\documentclass[10pt, draftclsnofoot, journal, onecolumn]{IEEEtran}
\pdfoutput=1

\usepackage{algorithm} 
\usepackage[noend]{algpseudocode} 
\usepackage{amssymb, amsmath, amsthm} 

\usepackage[english]{babel}
\usepackage{bbm} 
\usepackage[url=false, doi=false, eprint=false, isbn=false,style=ieee, backend=biber]{biblatex}
\usepackage{booktabs} 

\usepackage{caption}
\usepackage[aboveskip=-3pt]{subcaption}
\usepackage[babel]{csquotes}

\usepackage{enumerate} 

\usepackage{glossaries} 

\usepackage{hyperref} 

\usepackage{mathtools} 
\usepackage{multirow} 

\usepackage{placeins}

\usepackage{siunitx} 
\usepackage{subcaption} 
\usepackage{stackengine}

\usepackage{tikz}

\usepackage{upgreek} 
\newtheorem{conj}{Conjecture}
\newtheorem{theorem}{Theorem}

\def\presuper#1#2%
{\mathop{}%
	\mathopen{\vphantom{#2}}^{#1}%
	\kern-\scriptspace%
	#2}

\newacronym{ar}{AR}{autoregressive}

\newacronym{bh}{BH}{Benjamini-Hochberg procedure}

\newacronym{cca}{CCA}{canonical correlation analysis}
\newacronym{cdf}{CDF}{cumulative distribution function}
\newacronym{lfdrmultcost}{LFDR-MULT-COST}{\textbf{\gls{lfdr}}-based \textbf{mu}ltiple hypothesis \textbf{t}esting procedure for complete \textbf{co}rrelation \textbf{st}ructure identification}

\newacronym{edf}{EDF}{empirical distribution function}
\newacronym{evd}{EVD}{eigenvalue decomposition}

\newacronym{fmri}{fMRI}{functional magnetic resonance imaging}
\newacronym{fdr}{FDR}{false discovery rate}
\newacronym{fwer}{FWER}{family-wise error rate}

\newacronym{iid}{i.i.d.}{independent and identically distributed}

\newacronym{lfdr}{lfdr}{local false discovery rate}

\newacronym{mcca}{mCCA}{multiset \gls{cca}}
\newacronym{mht}{MHT}{multiple hypothesis testing}

\newacronym{os}{OSP}{one-step procedure}

\newacronym{pdf}{PDF}{probability density function}

\newacronym{rmt}{RMT}{random matrix theory}

\newacronym{snr}{SNR}{signal-to-noise ratio}

\newacronym{ts}{TSP}{two-step procedure}
\newcommand{\actMat}{\boldsymbol{M}}
\newcommand{\actMatEst}{\widehat{\mathbf{M}}}
\newcommand{\actEl}[2]{M_{#1}^{(#2)}}
\newcommand{\actElEst}[2]{\hat{\mathrm{M}}_{#1}^{(#2)}}
\newcommand{\alpFAStOne}{{}_\text{I}\alpha_{\text{FA}}}
\newcommand{\alpFAStTwo}{{}_\text{II}\alpha_{\text{FA}}}
\DeclareMathOperator*{\argmax}{argmax}

\newcommand{\btIdx}{b}

\newcommand{\canNumActSrc}{d}
\newcommand{\chkNormExp}[2][\setIdx]{\mu_{#1}^{(#2)}}

\newcommand{\cmpCohMat}{\boldsymbol{C}}
\newcommand{\cmpCorMat}{\boldsymbol{R}}
\newcommand{\cmpCorMatEst}{\corMatEst}
\newcommand{\cmpBlkCorMat}{\boldsymbol{R}_\text{D}}
\newcommand{\cmpBlkCorMatEst}{\corMatEst[\text{D}]}
\newcommand{\cmpCohMatEst}{\widehat{\mathbf{C}}}
\newcommand{\cmpCohMatEstBt}[1][\btIdx]{{}^{#1}\cmpCohMatEst^\ast}
\newcommand{\corCoeffSym}{\rho}
\newcommand{\corCoeffSrc}[2][\setIdx, \setIdx^\prime]{\corCoeffSym_{#1}^{(#2)}}

\newcommand{\corMat}[1][]{\boldsymbol{R}_{#1}}
\newcommand{\corMatEst}[1][]{\widehat{\mathbf{R}}_{#1}}
\newcommand{\corSet}[1][\srcIdx]{\overline{\mathcal{K}}^{(#1)}}

\newcommand{\E}{\mathsf{E}}
\newcommand{\eigVal}[1][\srcIdx]{\lambda^{(#1)}}
\newcommand{\eigValEst}[1][\srcIdx]{\hat{{\uplambda}}^{(#1)}}
\newcommand{\eigValEstBt}[2][\btIdx]{{}^{#1}\hat{{\uplambda}}^{(#2)^\ast}}
\newcommand{\eigValPdf}[1][\srcIdx]{f_{\hat{{\uplambda}}^{(#1)}}({\lambda})}
\newcommand{\eigVec}[1][\srcIdx]{{\boldsymbol{\eigVecSym}^{(#1)}}}
\newcommand{\eigVecEst}[1][\srcIdx]{{\hat{\mathbf{\eigVecSym}}^{(#1)}}}
\newcommand{\eigVecEstBt}[2][\btIdx]{{}^{#1}{\hat{\mathbf{\eigVecSym}}^{(#2)^\ast}}}
\newcommand{\eigVecChk}[2][\setIdx]{{\boldsymbol{\eigVecSym}_{#1}^{(#2)}}}
\newcommand{\eigVecChkNorm}[2][\setIdx]{{\mathrm{c}_{#1}^{(#2)}}}
\newcommand{\eigVecChkNormBt}[3][\btIdx]{{^{#1}\mathrm{c}_{#2}^{(#3)^\ast}}}
\newcommand{\eigVecChkEst}[2][\setIdx]{{\hat{\mathbf{\eigVecSym}}_{#1}^{(#2)}}}
\newcommand{\eigVecChkEstBt}[3][\btIdx]{{}^{#1}{\hat{\mathbf{\eigVecSym}}_{#2}^{(#3)^\ast}}}
\newcommand{\eigVecChkPdf}[2][\setIdx]{f_{\hat{\mathbf{\eigVecSym}}_{#1}^{(#2)}}(\boldsymbol{\eigVecSym})}
\newcommand{\eigVecSym}{u}

\newcommand{\fdr}{\mathsf{FDR}}
\newcommand{\fdrCmp}{\mathsf{FDR}_\text{cmp}}
\newcommand{\fdrThrAtom}{\alpha}
\newcommand{\fdrThrCmp}{\alpha_\text{cmp}}

\newcommand{\HAltAtom}[2][\setIdx]{\overline{H}_{{#1}}^{(#2)}}
\newcommand{\HNulAtom}[2][\setIdx]{H_{{#1}}^{(#2)}}
\newcommand{\HAltCmp}[1][\srcIdx]{\overline{H}^{(#1)}}
\newcommand{\HNulCmp}[1][\srcIdx]{H^{(#1)}}
\newcommand{\HSetAlt}[2][\setIdx]{{}_\text{II}\overline{H}_{#1}^{(#2)}}
\newcommand{\HSetNul}[2][\setIdx]{{}_\text{II}H_{#1}^{(#2)}}

\newcommand{\HSigAlt}[1][\canNumActSrc]{{}_\text{I}\overline{H}^{(#1)}}
\newcommand{\HSigNul}[1][\canNumActSrc]{{}_\text{I}H^{(#1)}}

\newcommand{\idc}[1]{\mathbbm{1}\!\left\{#1\right\}}
\newcommand{\iMat}[1]{\boldsymbol{I}_{#1\times#1}}

\newcommand{\mixMat}[1][\setIdx]{\boldsymbol{A}_{#1}}

\newcommand{\mixMatEl}[3][\setIdx]{a_{#1}^{(#2, #3)}}

\newcommand{\nonCorSet}[1][\srcIdx]{{\mathcal{K}}^{(#1)}}
\newcommand{\numActSrc}{D}
\newcommand{\numActSrcEst}{\hat{\numActSrc}}
\newcommand{\numBtSam}{B}
\newcommand{\numRej}{\mathrm{R}}
\newcommand{\numRejCmp}{\mathrm{R}_\text{cmp}}
\newcommand{\numSam}{N}
\newcommand{\numSrc}{J}
\newcommand{\numSet}{K}
\newcommand{\numSub}{I}
\newcommand{\numSubPerSet}[1][\setIdx]{I_{\setIdx}}
\newcommand{\numFalPos}{\mathrm{V}}
\newcommand{\numFalPosCmp}{\mathrm{V}_\text{cmp}}
\newcommand{\numTruPos}{\mathrm{S}}
\newcommand{\numTruPosCmp}{\mathrm{S}_\text{cmp}}
\newcommand{\nulFrc}{\pi_0}

\newcommand{\obsCorMat}[1][\setIdx]{\corMat[{\obsRVec[#1]\obsRVec[#1]}]}
\newcommand{\obsCorMatEst}[1][\setIdx]{\corMatEst[{\obsRVec[#1]\obsRVec[#1]}]}
\newcommand{\obsEl}[3][\setIdx]{{\obsSym_{#1, #2}^{(#3)}}}
\newcommand{\obsIdx}[1][]{i_{#1}}
\newcommand{\obsSym}{x}
\newcommand{\obsRVec}[1][\setIdx]{\mathbf{\obsSym}_{#1}}

\newcommand{\obsRvEl}[2][\setIdx]{{\mathrm{\obsSym}_{#1}^{(#2)}}}
\newcommand{\obsVec}[2][\setIdx]{\boldsymbol{\obsSym}_{#1, #2}}

\newcommand{\obsMat}[1][\setIdx]{\boldsymbol{X}_{#1}}
\newcommand{\obsMatBt}[2][\btIdx]{{}^{#1}\boldsymbol{X}^\ast_{#2}}

\newcommand{\p}[2][\setIdx]{p_{#1}^{(#2)}}
\newcommand{\pval}{$p$-value}
\newcommand{\pSet}[2][\setIdx]{p_{#1}^{(#2)}}
\newcommand{\pSig}[1][\canNumActSrc]{p^{(#1)}}

\newcommand{\samIdx}{n}
\newcommand{\setIdx}{k}
\newcommand{\srcEl}[3][\setIdx]{{\srcSym_{#1, #2}^{(#3)}}}
\newcommand{\srcCorMat}[1][\setIdx]{\corMat[{\srcRVec[#1]\srcRVec[#1]}]}
\newcommand{\srcCrossCorMat}{\corMat[{\srcRVec\srcRVec[\setIdx^\prime]}]}
\newcommand{\srcRvEl}[2][\setIdx]{{\mathrm{\srcSym}_{#1}^{(#2)}}}
\newcommand{\srcIdx}{j}
\newcommand{\srcMat}[1][\setIdx]{\boldsymbol{S}_{#1}}
\newcommand{\srcRVec}[1][\setIdx]{\mathbf{\srcSym}_{#1}}
\newcommand{\srcVec}[2][\setIdx]{\boldsymbol{\srcSym}_{#1, #2}}

\newcommand{\srcSym}{s}

\newcommand{\T}[2][\setIdx]{{\TSymRVar}_{{#1}}^{(#2)}}
\newcommand{\TBt}[3][\btIdx]{^{#1}{\TSymRVar}_{{#2}}^{(#3)}}
\newcommand{\TPdf}[2][\setIdx]{f_{\T[#1]{#2}}(\TSym)}
\newcommand{\TPdfAlt}[2][\setIdx]{f_{\T[#1]{#2}|\HAltAtom{\srcIdx}}(\TSym)}
\newcommand{\TPdfNull}[2][\setIdx]{f_{\T[#1]{#2}|\HNulAtom{\srcIdx}}(\TSym)}

\newcommand{\TCdfNullEst}[3][\setIdx]{\hat{F}_{\T[#1]{#2}|\HNulAtom{\srcIdx}}\left(#3\right)}
\newcommand{\TSetRVar}[2][\setIdx]{{_{\text{II}}}{\TSymRVar}_{{#1}}^{(#2)}}

\newcommand{\TSetPdfNull}[2][\setIdx]{f_{\eigVecChkNorm[#1]{#2}|\HSetNul{\srcIdx}}(c)}
\newcommand{\TSigRVar}[1][\canNumActSrc]{{\TSymRVar}^{(#1)}}

\newcommand{\TSigCdfNull}[1][\canNumActSrc]{F_{\TSigRVar[#1]|\HSigNul}(\TSym)}
\newcommand{\TSigPdfNull}[1][\canNumActSrc]{f_{\TSigRVar[#1]|\HSigNul}(\TSym)}
\newcommand{\TSym}{t}
\newcommand{\TSymRVar}{\mathrm{T}}

\graphicspath{{./images/},{./figures/}}

\begin{document}
	\title{Identifying the Complete Correlation Structure in Large-Scale High-Dimensional Data Sets with Local False Discovery Rates}
	
	\author{Martin~Gölz, Tanuj Hasija, Michael Muma and Abdelhak M. Zoubir %
		\thanks{This research was supported by the German Research Foundation (DFG) under grants SCHR 1384/3-2 and ZO 215/17-2}
		\thanks{M. Gölz and A. M. Zoubir are with the Signal Processing Group, TU Darmstadt, Germany. T. Hasija is with the Signal and Systems Group, Paderborn University, Germany. Michael Muma is with the Robust Data Science Group, TU Darmstadt, Germany.\\
		corresponding author e-mail: {goelz}@spg.tu-darmstadt.de}}
	
	\markboth{Submitted to Signal Processing}%
	{Gölz \MakeLowercase{\textit{et al.}}: TBD}
	
	\maketitle
	
	
	\begin{abstract}
		The identification of the dependent components in multiple data sets is a fundamental problem in many practical applications. The challenge in these applications is that often the data sets are high-dimensional with few observations or available samples and contain latent components with unknown probability distributions. A novel mathematical formulation of this problem is proposed, which enables the inference of the underlying correlation structure with strict false positive control. In particular, the false discovery rate is controlled at a pre-defined threshold on two levels simultaneously. The deployed test statistics originate in the sample coherence matrix. The required probability models are learned from the data using the bootstrap. Local false discovery rates are used to solve the multiple hypothesis testing problem. Compared to the existing techniques in the literature, the developed technique does not assume an a priori correlation structure and work well when the number of data sets is large while the number of observations is small. The simulation results underline that it can handle the presence of distributional uncertainties, heavy-tailed noise, and outliers.
	\end{abstract}


	\section{Notation}
\label{sec:not}
	Italic normal font letters $a$ and $A$ denote deterministic scalar quantities. Deterministic vectors and matrices are represented by bold italic $\boldsymbol{a}$ and $\boldsymbol{A}$, respectively. Upright $\mathrm{A}$, $\mathbf{a}$ and $\mathbf{A}$ symbolize random variables, vectors and matrices, respectively. $f_\mathrm{A}(a)$, $F_\mathrm{A}(a)$ denote the \gls{pdf} and \gls{cdf} of random variable $\mathrm{A}$ in dependence its realization $a$; $\E[\mathrm{A}]$ is its expected value. We write $g(\mathrm{A})$ for a genericc function of $\mathrm{A}$. Calligraphic $\mathcal{A}$ denotes an arbitrary set with complement $\overline{\mathcal{A}}$ and $[A]$ is the set of non-negative integers $1, \dots, A$. Multiletter abbreviations representing mathematical quantities come sans-serif, e.g. $\fdr$. The indicator function is $\idc{\cdot}$. We write $||\mathbf{a}||$ for the Euclidean norm of $\mathbf{a}$ and $\mathrm{diag}(\mathbf{a})$ is a diagonal matrix with the elements of $\mathbf{a}$ on its main diagonal. $\widehat{a}$ denotes both, the estimator and a sample estimate of $a$.	
	\section{Problem Formulation}
\label{sec:prb-for}
	\subsection{System Model}
	\label{sec:prb-for_sys-mod}
		We consider \(\numSet\) \textit{data sets} which are composed of zero-mean, real-valued random \textit{observation vectors} \(\{\obsRVec\}_{\setIdx\in[\numSet]}\). For notational simplicity, we assume that all observation vectors contain an equal number \(\numSub\) of \textit{observations} or \textit{subjects} \(\obsRvEl{\obsIdx}\), such that \(\obsRVec = \Big[\obsRvEl{1},\dots \obsRvEl{\numSub}\Big]^\top\forall\, \setIdx\in[\numSet]\). The observation vectors are assumed to be generated by the linear mixing of the latent \textit{component vectors} \(\{\srcRVec\}_{\setIdx\in[\numSet]}\in \mathbb{R}^\numSrc\). The component vectors each contain \(\numSrc\leq\numSub\) \textit{components}. The \(\srcIdx\)th component of data set \(\setIdx\) is denoted by \(\srcRvEl{\srcIdx}\). The relation between the observations and components is hence
		\begin{equation}
			\label{eq:theo-sig-mod}
			\underbrace{
				\begin{bmatrix}
					\obsRvEl{1}\\
					\vdots\\
					\obsRvEl{\numSub}
				\end{bmatrix}}_{\obsRVec} = 
				\underbrace{
					\begin{bmatrix}
						\mixMatEl{1}{1} & \cdots & \mixMatEl{1}{\numSrc}\\
						\vdots & \ddots & \vdots \\
						\mixMatEl{\numSub}{1} & \cdots & \mixMatEl{\numSub}{\numSrc}
					\end{bmatrix}}_{\mixMat} \cdot
				\underbrace{
					\begin{bmatrix}
						\srcRvEl{1}\\
						\vdots \\
						\srcRvEl{\numSrc}
					\end{bmatrix}}_{\srcRVec},
				\quad \setIdx \in[\numSet],
		\end{equation}
		where \(\mixMat \in \mathbb{R}^{{\numSub}\times \numSrc}\)  is an unknown deterministic mixing matrix with full column rank. Without loss of generality, the components are assumed to be zero-mean and unit variance, i.e.,  
		\begin{align}
			\E\Big[\srcRvEl{\srcIdx}\Big] 	&= {0}, \qquad \text{and}\\
			\E\Big[\srcRvEl{\srcIdx}^2\Big]	&= 1, \qquad \forall\,\setIdx\in[\numSet], \srcIdx\in[\numSrc].
		\end{align}
		For each \(\setIdx\in[\numSet]\), we observe \(\numSam\) \(\numSub\)-dimensional \gls{iid} realizations of \(\obsRVec\) that we also refer to as the \textit{observation samples} \(\obsVec{1}, \dots, \obsVec{\numSam}\), where $\obsVec{\samIdx} = \Big[\obsEl{\samIdx}{1}, \dots, \obsEl{\samIdx}{\numSub}\Big]^\top$. These are summarized in the \(k\)th \(\numSub\times\numSam\) \textit{observation matrix} \(\obsMat = [\obsVec{1}~ \dots~ \obsVec{\numSam}]\). The in practice unobservable \(k\)th \(\numSrc\times\numSam\) \textit{component matrix} \(\srcMat\) is defined analogously. The sample index \(\samIdx\in[\numSam]\) is consistent over the sets and components, i.e., \(\obsVec{\samIdx}, \srcVec{\samIdx}, \obsVec[\setIdx^\prime]{\samIdx}\) and \(\srcVec[\setIdx^\prime]{\samIdx}\) denote the paired realizations of random vectors \(\obsRVec[\setIdx], \srcRVec, \obsRVec[\setIdx^\prime],\) and \(\srcRVec[\setIdx^\prime]\) for \(\setIdx, \setIdx^\prime\in[\numSet]\). Hence, the sample-based equivalent to Eq.~\eqref{eq:theo-sig-mod} is
		\begin{equation}
			\label{eq:sam-sig-mod}
			\underbrace{\begin{bmatrix}
					\obsEl{1}{1} & \dots & \obsEl{\numSam}{1}\\
					\vdots & \ddots & \vdots \\
					\obsEl{1}{\numSub} & \dots & \obsEl{\numSam}{\numSub}\\
				\end{bmatrix}}_{\obsMat} =
						\underbrace{\begin{bmatrix}
							\mixMatEl{1}{1} & \cdots & \mixMatEl{1}{\numSrc}\\
							\vdots & \ddots & \vdots \\
							\mixMatEl{\numSub}{1} & \cdots & \mixMatEl{\numSub}{\numSrc}
						\end{bmatrix}}_{\mixMat} \cdot
						\underbrace{\begin{bmatrix}
							\srcEl{1}{1} & \dots & \srcEl{\numSam}{1}\\
							\vdots & \ddots & \vdots \\
							\srcEl{1}{\numSrc} & \dots & \srcEl{\numSam}{\numSrc}\\
						\end{bmatrix}}_{\srcMat},
						\quad \setIdx \in[\numSet].
		\end{equation}
	\subsection{Underlying correlation structure}
	\label{sec:prb-for_corr-str}
		The observation vectors \(\obsRVec\) are the result of a linear combination of the underlying components. In practice, the observations are the outcome of the interaction between the different physical processes \(\big\{\srcRvEl{\srcIdx}\big\}_{\srcIdx\in[\numSrc]}\) with an environment described by \(\mixMat\). The correlation structure between the underlying source components of different data sets defines the degree to which the observations in different sets originate in the same cause. Mathematically, this structure can be expressed using the component cross-covariance matrix between data sets \(\setIdx, \setIdx^\prime\in[\numSet]\)
		\begin{equation}
			\srcCrossCorMat = \E\Big[\srcRVec\srcRVec[\setIdx^\prime]^\top\Big]
		\end{equation}

		The entry at position \((\srcIdx, \srcIdx^\prime) \in [\numSrc] \times [\numSrc]\) of \(\srcCrossCorMat\) is the correlation coefficient \(\corCoeffSrc{\srcIdx, \srcIdx^\prime}\) between the \(\srcIdx\)th component of set \(\setIdx\) and the \(\srcIdx^\prime\)th component of set \(\setIdx^\prime\). If \(\corCoeffSrc{\srcIdx, \srcIdx^\prime}\neq 0\), the \(\srcIdx\)th component of set \(\setIdx\) is (partially) driven by the same underlying physical phenomenon as the \(\srcIdx^\prime\)th component of set \(\setIdx^\prime\).

			We impose the following assumptions to facilitate the learning of the structure of the latent components from the observations.
			\begin{enumerate}[I)]
				\item Intraset independence: the components are uncorrelated \textit{within} each set, i.e.,   
				\begin{equation}
					\srcCorMat = \E\Big[\srcRVec\srcRVec^\top\Big] = \iMat{\numSrc},
				\end{equation}
				where \(\iMat{\numSrc}\) is the \(\numSrc\times\numSrc\) identity matrix. This is a mild assumption: If there exists correlation between the components of a data set, those can be summarized as a single component that absorbs all components that are correlated within the set. Naturally, this reduces the dimension \(\numSrc\) of the component vector.
				\item Pairwise interset dependence: the components between any two data sets \(\setIdx, \setIdx^\prime\) may only be correlated pairwise. Thus, component $\srcRvEl[\setIdx^\prime]{\srcIdx}$ may correlate with component $\srcRvEl[\setIdx]{\srcIdx}$, but not with \(\srcRvEl[\setIdx^\prime]{\srcIdx^\prime}\), \(\srcIdx\neq\srcIdx^\prime\).
				This implies that the component cross-covariance matrix between data sets $\setIdx$ and $\setIdx^\prime \,\forall\,\setIdx, \setIdx^\prime\in[\numSet]:\setIdx\neq\setIdx^\prime$ is diagonal,
				\begin{equation}
					\srcCrossCorMat = \mathrm{diag}\Big(\Big[\corCoeffSrc{1}, \corCoeffSrc{2}, \ldots, \corCoeffSrc{\numSrc}\Big]\Big),
				\end{equation}
				where we use $\corCoeffSrc{\srcIdx} \equiv \corCoeffSrc{\srcIdx, \srcIdx}$ for readability. This assumption is common place in the literature on correlation analysis for multiple data sets, see \cite{Hasija2020} and the references therein. To the best of our knowledge, all existing methods to extract the correlation structure between multiple sets require this assumption. 
				\item The correlations are transitive. Hence, if $\corCoeffSrc{\srcIdx} > 0$ and $\corCoeffSrc[\setIdx, \setIdx^{\prime\prime}]{\srcIdx} > 0$, then also $\corCoeffSrc[\setIdx^{\prime}, \setIdx^{\prime\prime}]{\srcIdx} > 0\,\forall\, \setIdx, \setIdx^\prime, \setIdx^{\prime\prime} \in[\numSet]$.
			\end{enumerate}
			The correlation structure is fully specified by the correlation coefficients \(\corCoeffSrc{\srcIdx}, \forall\,\srcIdx\in[\numSrc]\) on the main diagonal of the correlation matrices \(\srcCrossCorMat, \setIdx, \setIdx^\prime\in[\numSet], \setIdx\neq\setIdx^\prime\), since the off-diagonal entries are all zero due to Assumption~II) and \(\srcCorMat\) is identity according to Assumption~I. We define the \(\numSrc\times\numSet\) \textit{activation matrix}
			\begin{equation}
				\label{eq:act-mat-def}
				\actMat =
				\begin{bmatrix}
					\actEl{1}{1} & \dots & \actEl{\numSet}{1}\\
					\vdots & \ddots & \vdots\\
					\actEl{1}{\numSrc} & \dots & \actEl{\numSet}{\numSrc}
				\end{bmatrix}
			\end{equation}
			to indicate which components are correlated across which data sets. If the \(\srcIdx\)th component is correlated between data set \(\setIdx\in[\numSet]\) and at least one other data set \(\setIdx^\prime\in[\numSet]:\setIdx\neq\setIdx^\prime\), i.e., $\exists\, \corCoeffSrc{\srcIdx}>0$, then \(\actEl{\setIdx}{\srcIdx} = \actEl{\setIdx^\prime}{\srcIdx} = 1\). Otherwise, \(\actEl{\setIdx}{\srcIdx} = 0\). Thus, either \(\sum_{k = 1}^K\actEl{k}{j} \geq 2\) or \(\sum_{k = 1}^K\actEl{k}{j} = 0\).
We write \(\corSet,\srcIdx\in[\numSrc]\) to denote the collection of those sets whose \(\srcIdx\)th component is correlated. Its complement is \(\nonCorSet\). 
If all $\srcCrossCorMat, \,\forall\,\setIdx,\setIdx^\prime\in[\numSet]$ were observable, $\actMat$ could be deduced directly from their non-zero elements. We illustrate the relation between the assumptions, the correlation matrices and the activation matrix by a low-dimensional example with \(\numSet = 4\) data sets in Fig.~\ref{fig:cor-struc-ex}.

			\begin{figure}
				\begin{tabular}{cc}
					\begin{minipage}{0.45\textwidth}
						\begin{flushleft}
							\begin{tabular}{c| c c c c c c}
								\(\srcIdx\) & \(\corCoeffSrc[1, 2]{\srcIdx}\) & \(\corCoeffSrc[1, 3]{\srcIdx}\) & \(\corCoeffSrc[1, 4]{\srcIdx}\) & \(\corCoeffSrc[2, 3]{\srcIdx}\) & \(\corCoeffSrc[2, 4]{\srcIdx}\) & \(\corCoeffSrc[3, 4]{\srcIdx}\)\\\toprule
								\(\mathit{1}\) & \(.7\) & \(.8\) & \(.8\) & \(.7\) & \(.8\) & \(.9\) \\
								\(\mathit{2}\) & \(.9\) & \(0\) & \(.75\) & \(0\) & \(.8\) & \(0\) \\
								\(\mathit{3}\) & \(0\) & \(0\) & \(0\) & \(.9\) & \(0\) & \(0\) \\
								\(\mathit{4}\) & \(0\) & \(.4\) & \(.35\) & \(0\) & \(0\) & \(.4\) \\
								\(\mathit{5}\) & \(0\) & \(0\) & \(0\) & \(.5\) & \(0\) & \(0\) \\
								\(\mathit{6}\) & \(0\) & \(0\) & \(0\) & \(0\) & \(0\) & \(0\)\\
								\(\mathit{7}\) & \(0\) & \(0\) & \(0\) & \(0\) & \(0\) & \(0\)
							\end{tabular}
					  \end{flushleft}
					\end{minipage}&
					\begin{minipage}{.05\textwidth}
						\begin{equation*}
							\Longleftrightarrow
						\end{equation*}
					\end{minipage}
					\begin{minipage}{0.45\textwidth}
						\vphantom{j}
							\begin{equation*}
								M = 
								\begin{bmatrix}
									1 & 1 & 1 & 1 \\
									1 & 1 & 0 & 1 \\
									0 & 1 & 1 & 0 \\
									1 & 0 & 1 & 1 \\
									0 & 1 & 1 & 0 \\
									0 & 0 & 0 & 0 \\
									0 & 0 & 0 & 0 	
								\end{bmatrix}
							\end{equation*}
					\end{minipage}
				\end{tabular}
				\caption{An example with \(\numSet=4\) data sets and \(7\) components per set. The correlation coefficients are given in the table on the left and the corresponding activation matrix is shown on the right. Such correlation structures arise in a large number of practical applications. For instance, the different data sets may correspond to recordings by different cameras that are deployed in different locations. The components may then be the visual signatures of different visual objects and identifying correlated components may then help to track the movement of a certain object across the surveillance area. In environmental monitoring, each data set may correspond to the measurements obtained with a sensor that records multiple environmental measures such as temperature, humidity or pollution. The components could then correspond to different metheorological or man-made phenomena and knowing the correlation structure across different senors would be useful in both, identifying the presence of such underlying phenomena and localizing the areas which they affect.}
				\label{fig:cor-struc-ex}
			\end{figure}
			\paragraph*{Remark on the component indices}
				The methods discussed in this work assume that the "absolute" index $\srcIdx\in[\numSrc]$ of a component is of limited interest. The covered procedures generally assume that the components are sorted such that component with index $\srcIdx = 1$ exhibits the strongest correlation across all data sets and with $\srcIdx = \numSrc$ the weakest. This is well-justified in practical applications where the underlying component and mixing matrices are typically unknown and the absolute ordering of rows of $\mixMat$ and columns of $\srcMat$ does not play a role.

	\section{Correlation Structure Estimation based on the Coherence Matrix}
\label{sec:rel-wrk}
	We intend to identify the unknown structure between the latent components across the different data sets.
	In practice, $\actMat$ is not directly observable and must be estimated based on the sample-based model given in Eq.~\eqref{eq:sam-sig-mod}. We denote its estimate by $\actMatEst$, with entries $\actElEst{\setIdx}{\srcIdx}$. The problem of estimating $\actMat$ was first considered in . The multiple testing procedure proposed in this work utilizes some of the theoretical findings from \cite{Hasija2020}. Hence, we summarize the procedure developed in \cite{Hasija2020} in what follows. 

	\subsection{The \gls{ts} \cite{Hasija2020}}
	\label{sec:ts-prc}
		The number of components that are correlated across at least two data sets be denoted by $\numActSrc\leq\numSrc$. $\numActSrc$ is equivalent to the total number of rows in $\actMat$ with at least two non-zero elements. The \gls{ts} from \cite{Hasija2020} bases upon the eigenvalues and eigenvectors of the so-called \textit{composite coherence matrix}
		\begin{equation}
		\label{eq:coh-mat}
			\cmpCohMat = \cmpBlkCorMat^{-\frac{1}{2}}\cmpCorMat\cmpBlkCorMat^{-\frac{1}{2}} =\cmpBlkCorMat^{-\frac{1}{2}}\!\cdot\!\E\bigg[\Big[\obsRVec[1]^\top, \dots, \obsRVec[\numSet]^\top\Big]^\top\cdot\Big[\obsRVec[1]^\top, \dots, \obsRVec[\numSet]^\top\Big]\bigg]\cdot\cmpBlkCorMat^{-\frac{1}{2}}.
		\end{equation}
		$\cmpBlkCorMat = \mathsf{blkdiag}(\obsCorMat[1], \dots, \obsCorMat[\numSet])$ is block-diagonal with blocks $\obsCorMat = \E\Big[\obsRVec\obsRVec^\top\Big]$ and $(\cdot)^{-\frac{1}{2}}$ denotes the inverse of the matrix square root.
		 
		The eigenvalues of $\cmpCohMat$ sorted in descending order be denoted by $\eigVal[1], \dots, \eigVal[\numSrc\numSet]$ with corresponding eigenvectors \(\eigVec[1], \dots, \eigVec[\numSrc\numSet]\). A series of theorems in \cite{Hasija2020} proves that under certain conditions, exactly one out of the $\numSet$ eigenvalues associated with the $\srcIdx$th correlated component is greater than $1$, $\srcIdx\in[\numActSrc]$. The remaining $\numSet-1$ eigenvalues of component \(\srcIdx\in[\numActSrc]\) are all \(\leq1\). For the $\srcIdx$th component $\numActSrc<\srcIdx\leq \numSrc$ that is uncorrelated across all $\numSet$ data sets, all associated eigenvalues are equal to $1$. In total, exactly $\numActSrc$ eigenvalues of $\cmpCohMat$ are greater than $1$. Hence, the correlation structure analysis in \cite{Hasija2020} focuses exclusively on the largest $\numSrc$ eigenvalues $\eigVal$ and the corresponding $\numSrc\cdot\numSet$-dimensional eigenvectors $\eigVec$, $\srcIdx\in[\numSrc]$. The latter can be decomposed as $\eigVec = \bigg[\eigVecChk[1]{\srcIdx}^\top, \dots, \eigVecChk[\numSet]{\srcIdx}^\top\bigg]^\top$, where the $\numSrc$-dimensional \textit{eigenvector chunk} $\eigVecChk{\srcIdx}$ summarizes the contribution of data set $\setIdx$ to eigenvalue $\eigVal$, $\srcIdx\in[\numSrc]$. With $\boldsymbol{0}_{\numSrc}$ denoting a $\numSrc$-dimensional all-zero vector, the following has been shown about the eigenvector chunks \cite{Hasija2020}:
		\begin{equation}
			\eigVecChk{\srcIdx}
				\begin{cases}
					= \boldsymbol{0}_{\numSrc}, 	& \text{if $\srcIdx\in[\numActSrc]$ and $\setIdx \in \corSet$, i.e., the $\srcIdx$th component is correlated between set $\setIdx$ and another set},\\
					\neq \boldsymbol{0}_{\numSrc},	& \text{if $\srcIdx\in[\numActSrc]$ and $\setIdx \in \nonCorSet$, i.e., the $\srcIdx$th component is not correlated between set $\setIdx$ and any other set}.
				\end{cases}
		\end{equation}
		This property of the coherence matrix eigenvector chunks has lead to a two-step procedure for identifying correlations across sets in \cite{Hasija2020}. 
		\begin{enumerate}[I)]
			\item Find all $\numActSrc$ eigenvalues of $\cmpCohMat$ for which $\eigVal>1,\, \srcIdx\in[\numActSrc]$ and set $\actEl{\setIdx}{\srcIdx} = 0\, \forall\, \srcIdx > \numActSrc, \setIdx\in[\numSet]$. 
			\item The remaining entries of $\actMat$ are found with $\actEl{\setIdx}{\srcIdx} = \idc{\Big\vert\Big\vert\eigVecChk{\srcIdx}\Big\vert\Big\vert^2 > 0},\,\forall\, \srcIdx\in[\numActSrc], \setIdx\in[\numSet]$.
		\end{enumerate}
		In practice, $\cmpCohMat$ is unknown and has to be estimated from the data $\obsMat,\setIdx\in[\numSet]$. The estimate $\cmpCohMatEst$ \cite{Hasija2020} is a random matrix. The sample covariance matrix is deployed to estimate the quantities in Eq.~\eqref{eq:coh-mat}, 
		\begin{align}
		\label{eq:sam-coh-est}
			\cmpBlkCorMatEst 	&= \mathsf{blkdiag} (\obsCorMatEst[1], \dots, \obsCorMatEst[\numSet]),\\	
			\obsCorMatEst		&= \frac{1}{\numSam} \sum_{\samIdx = 1}^{\numSam}\obsVec{\samIdx}\obsVec{\samIdx}^\top,\\
			\cmpCorMatEst		&= \frac{1}{\numSam} \sum_{\samIdx = 1}^{\numSam} \Big[\boldsymbol{\obsSym}_{1, \samIdx}^\top, \dots, \boldsymbol{\obsSym}_{\numSet, \samIdx}^\top\Big]^\top\cdot\Big[\boldsymbol{\obsSym}_{1, \samIdx}^\top, \dots, \boldsymbol{\obsSym}_{\numSet, \samIdx}^\top\Big].
		\end{align}
		The eigenvalues $\eigValEst$ and eigenvectors $\eigVecEst$ of $\cmpCohMatEst$ are random variables and vectors, respectively. As a consequence, neither the eigenvalues of $\cmpCohMat$ associated with uncorrelated components are strictly $\leq1$, nor are the chunk norms $\Big\vert\Big\vert\eigVecChkEst{\srcIdx}\Big\vert\Big\vert^2, \srcIdx\in[\numActSrc]$ exactly equal to $0$ if $\setIdx\in\nonCorSet$. 
		Instead, the eigenvalues and eigenvector chunk norms follow \gls{pdf}s $\eigValPdf$ and $\eigVecChkPdf{\srcIdx}$. To infer the correlation structure from those random quantities, the authors of \cite{Hasija2020} propose to formulate both steps as separate hypothesis testing problems. 
		\paragraph*{Step~I}
		By iterating over $\canNumActSrc\in[\numSrc]$, perform a series of binary tests between the $\canNumActSrc$ \text{total correlations} null hypothesis $\HSigNul$ and the $\canNumActSrc$ \text{total correlations} alternative $\HSigAlt$  where
			\begin{align}
			\label{eq:s1-hyp}
				\HSigNul:  &\qquad\canNumActSrc = \numActSrc,\quad \text{i.e., there are exactly $\canNumActSrc$ components with non-zero inter-set correlation},\\
				\HSigAlt:  &\qquad\canNumActSrc > \numActSrc,\quad \text{i.e., there are more than $\canNumActSrc$ components with non-zero inter-set correlation}.
			\end{align}
			The decisions between \(\HSigNul\) and \(\HSigAlt\) are based on test statistics \(\TSigRVar = g\big(\eigValEst[1],\dots, \eigValEst[\canNumActSrc]\big)\), where \(g: d\rightarrow 1\) is a function of the \(\canNumActSrc\) largest eigenvalues of $\cmpCohMatEst$ \cite{Hasija2020}. 
			The test statistic \gls{pdf} and \gls{cdf} under the $\canNumActSrc$th component number null hypothesis $\HSigNul$ be denoted by $\TSigPdfNull$ and \(\TSigCdfNull\). 
			The decisions between \(\HSigNul\) and \(\HSigAlt\) are made by
			thresholding the \text{total correlations} \pval~\(\pSig\) at false alarm level \(\alpFAStOne\), i.e.,
			\begin{equation}
			\label{eq:s1-test}
				\int_{\TSigRVar}^{\infty}\TSigPdfNull\mathrm{d}{\TSym} = \pSig \underset{\HSigAlt}{\overset{\HSigNul}{\gtrless}} \alpFAStOne.
			\end{equation}
			If \(\HSigNul\) the estimated number of correlated sources is \(\numActSrcEst = d\). Hence, set $\actElEst{\setIdx}{\srcIdx}=0,\,\forall\,\srcIdx\in\{\numActSrcEst + 1, \dots, \numSrc\}$, \(\setIdx\in[\numSet]\).
		\paragraph*{Step~II} Given $\numActSrcEst$ from Step~I) such that $\numActSrcEst\leq \numActSrc$, define a set of $\numSet$ binary hypotheses for each $\srcIdx\in[\numActSrcEst]$,
			\begin{align}
				\label{eq:s2-hyp}
				\HSetNul{\srcIdx}:  &\qquad\setIdx\in\nonCorSet:\text{ The $\srcIdx$th component is uncorrelated between the sets $\setIdx,\setIdx^\prime,\,\forall\, \setIdx^\prime\in[\numSet],\setIdx\neq\setIdx^\prime$},\\
				\HSetAlt{\srcIdx}:  &\qquad\setIdx\in \corSet:\text{ $\exists\,\setIdx^\prime\in[\numSet],\setIdx\neq\setIdx^\prime$ s.t. the $\srcIdx$th component is correlated between the sets $\setIdx$ and $\setIdx^\prime$}.
			\end{align}
			We refer to $\HSetNul{\srcIdx}$ and $\HSetAlt{\srcIdx}$ as the \text{set-wise} $\srcIdx$\text{th} \text{component null hypothesis} and \text{set-wise} $\srcIdx$th \text{component alternative}, respectively. As test statistics, the norms of the estimated eigenvector chunks $\eigVecChkNorm{\srcIdx} = \Big\vert\Big\vert\eigVecChkEst{\srcIdx}\Big\vert\Big\vert^2$ are deployed. 
			Under the set-wise component null hypothesis, the test statistic \gls{pdf} is $\TSetPdfNull{\srcIdx}$. 
			The decisions between the \(\HSetNul{\srcIdx}\) and \(\HSetAlt{\srcIdx}\) are made by thresholding the set-wise component \pval s \(\pSet{\srcIdx}\), i.e., 
			\begin{equation}
			\label{eq:s2-test}
				\int_{\eigVecChkNorm{\setIdx}}^{\infty}\TSetPdfNull{\srcIdx}\mathrm{d}{c} = \pSet{\srcIdx} \underset{\HSetAlt{\srcIdx}}{\overset{\HSetNul{\srcIdx}}{\gtrless}} \alpFAStTwo.
			\end{equation}
			with the user-defined Step~II false alarm probability level $\alpFAStTwo$. 
			If $\HSetNul{\srcIdx}$ is accepted, then $\actElEst{\setIdx}{\srcIdx}= 0$. Otherwise, $\actElEst{\setIdx}{\srcIdx} = 1$. This completes the \gls{ts} estimator of \(\actMat\).
			
	\section{The Proposed One-Step Procedure with False Positive Control}
\label{sec:prop-os}
	In this section, we propose a novel holistic approach to estimating the activation matrix $\actMat$. We first motivate this approach by highlighting the benefits of a single detection step over the state-of-the art \gls{ts} from \cite{Hasija2020}. We then provide the required theory with a particular focus on the false positive control properties of the proposed \gls{os}.
	\subsection{Motivation}
	\label{sec:prop-os_mot}
		A \gls{ts} to identifying the complete correlation structure between components from multiple data sets like the one from \cite{Hasija2020} exhibits two major problems. Firstly, false alarms are controlled in both steps individually, but not for the combination of the two steps. Secondly, controlling the false alarm probability is not well-suited to limit false positives when many binary hypothesis tests are performed simultaneously. In what follows, we discuss those issues in detail and propose solutions that form the basis of the proposed \gls{os}.
		
		\subsubsection{Error propagation} 
			The \gls{ts} attempts to control the false positives in both steps independently by thresholding the respective \pval s at nominal false alarm levels $\alpFAStOne$ and $\alpFAStTwo$ in Eqs.~\eqref{eq:s1-hyp}, \eqref{eq:s2-hyp}. However, the false positives are not controlled for the sequential combination of these two steps. The method proposed in \cite[Theorem~1]{Hasija2020} infers that a data set \(\setIdx\) belongs to the collection of correlated sets \(\corSet\) of the \(\srcIdx\)th component, if the corresponding eigenvector chunk norm \(\eigVecChkNorm{\srcIdx} = ||\eigVecChkEst{\srcIdx}||^2\) is significantly non-zero, since \(||\eigVecChk{\srcIdx}||^2 \neq 0\) if \(\setIdx\in\corSet\) and \(||\eigVecChk{\srcIdx}||^2=0\) otherwise. However, this assumption does not hold if the \(\srcIdx\)th component is uncorrelated across all sets. Since \(\alpFAStOne > 0\), 
			false alarms can occur in Step~I. If Step~I concludes with a false alarm, the number of components with correlation between at least two sets is overestimated, $\numActSrcEst>\numActSrc$. Then, the chunk norms associated with the \(\srcIdx\)th component, \(\numActSrc < \srcIdx\leq \numActSrcEst\) all are significantly non-zero and the tests in Step~II are based on wrong assumptions. The heatmap in Fig.~\ref{fig:chunk-ex-hm} shows the average chunk norm value obtained for the toy example introduced previously in Fig.~\ref{fig:cor-struc-ex}, with $\numSam = 500$ realizations of Gaussian distributed component and noise vectors. The \gls{snr} is \(5\unit{dB}\). Since \(\numActSrc = 5\), the sixth and seventh component are entirely uncorrelated. However, the average chunk norms in the sixth and seventh row are far from zero. Instead, their values are similar to those of the first component, which is correlated across all sets. In addition, the \glspl{edf} of the \pval s for Step~II are provided for some components and sets in Fig.~\ref{fig:chunk-ex-edf}. If the \(\srcIdx\)th component of set \(\setIdx\) is uncorrelated with the other sets, the \pval s should closely follow a uniform distribution \(\mathcal{U}(0, 1)\) to guarantee that the false alarm probability is bounded by the nominal level \(\alpFAStTwo\). The nominal false alarm level is fulfilled for the \(\srcIdx\)th components, \(\srcIdx\in\numActSrc\), since the \pval~distribution is approximately uniform in the typical \([0, .2]\) range of \(\alpFAStTwo\). As an example, we provide the \gls{edf} for the second component and the third set on the top left of Fig.~\ref{fig:chunk-ex-edf}. On the top right, the \gls{edf} of the \pval s for the fourth component of set \(\setIdx = 1\) is shown, which is correlated with the fourth component of the third and the fourth set. Here, a lot of mass is located close to zero, making it highly likely to discover the correlation. The second row of Fig.~\ref{fig:chunk-ex-edf} shows some \pval~distributions obtained for the sixth and seven component, which are uncorrelated across all sets. Nevertheless, these \glspl{edf} look very similar to the one on the top right. Hence, falsely identifyig correlation is much more likely than the nominal false alarm level \(\alpFAStTwo\). The statistical behavior of the chunk norms $\eigVecChkEst{\srcIdx}$ for $\numActSrc < \srcIdx$ is similar to that of chunk norms of components which are correlated across all data sets. Hence, overestimation of \(\numActSrcEst>\numActSrc\) leads to uncontrolled false alarm probabilities in Step~II).

			A different problem arises if Step~I terminates with a Type~II error, i.e., a decision in favor of $\HSigNul$ while $\canNumActSrc<\numActSrc$. Then, regardless of the information contained in the eigenvector chunks, no correlations can be identified for the $\srcIdx$th components of any data sets, $\numActSrcEst<\srcIdx\leq\numActSrc$. A potentially significant amount of true correlations  remain unidentified even if the eigenvector chunks contain strong evidence in favor of additional correlated components. This is illustrated in Fig.~\ref{fig:fn-ex}. In this example, \(\numActSrc = 6\) out of $\numSrc = 10$ components are correlated across \(7, 6, 5, 4, 3\) and \(2\) out of $\numSet = 15$ data sets with correlation coefficients \(7, .7, .65, .6, .6, .55\), respectively. $\numSam = 300$ realizations are observed and the components and noise are Gaussian with \(\mathrm{SNR} = 5\unit{dB}\). The correlation coefficients are constant per component, i.e., if the \(\srcIdx=5\)th component is correlated across sets \(1, 2\) and \(3\), then \(\corCoeffSrc[1, 2]{5} = \corCoeffSrc[2, 3]{5} = \corCoeffSrc[1, 3]{5} = .6\). The results are averaged over \(100\) independent realizations of the experiment. The average chunk norms in Fig.~\ref{fig:fn-ex-hm} are again distinctly indicating the true correlation structure among the first five components. However, Step~I consistently underestimates \(\numActSrc\): The average of \(\numActSrcEst\) is \(1.67\) for the reasonable false alarm level \(\alpFAStOne = .1\). Since more than half of the true correlations occur beyond the second component, a significant percentage of true correlations is not even tested for in Step~II, as the detection probabilities in Fig.~\ref{fig:fn-corr-est} underline. 
			Hence, the \gls{ts} exhibits poor detection power and a large fraction of the true correlations remain undetected, despite that the chunk norms displayed in Fig.~\ref{fig:fn-ex-hm} yield considerable evidence on the underlying correlation structure. 

			\captionsetup{labelsep=endash}
			\begin{figure}
				\centering
				\begin{subfigure}{.48\textwidth}
					\centering
					\includegraphics[scale=.4]{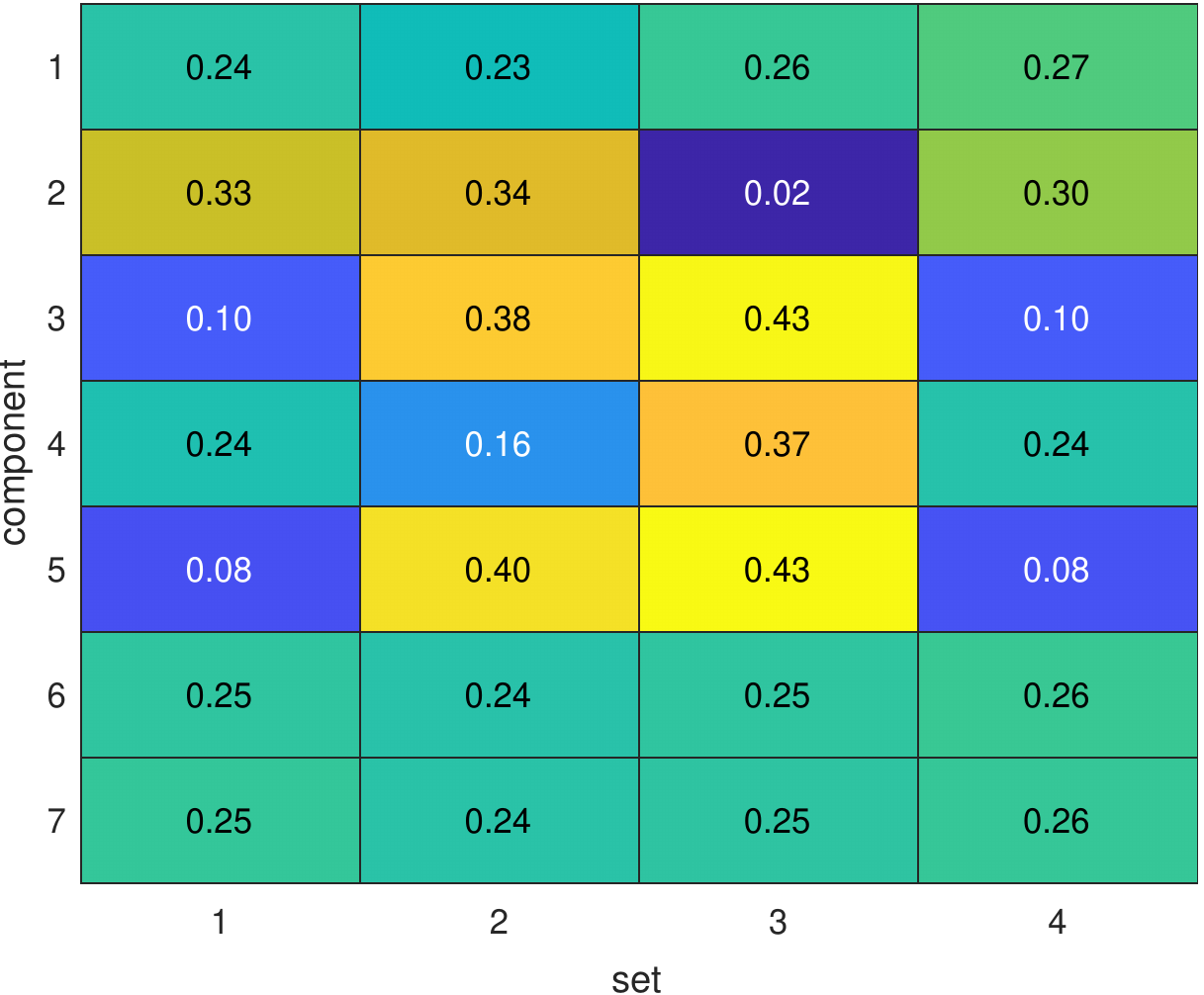}
					\caption{}
					\label{fig:chunk-ex-hm}
				\end{subfigure}
				\begin{subfigure}{.48\textwidth}
					\centering
					\includegraphics[scale=.4]{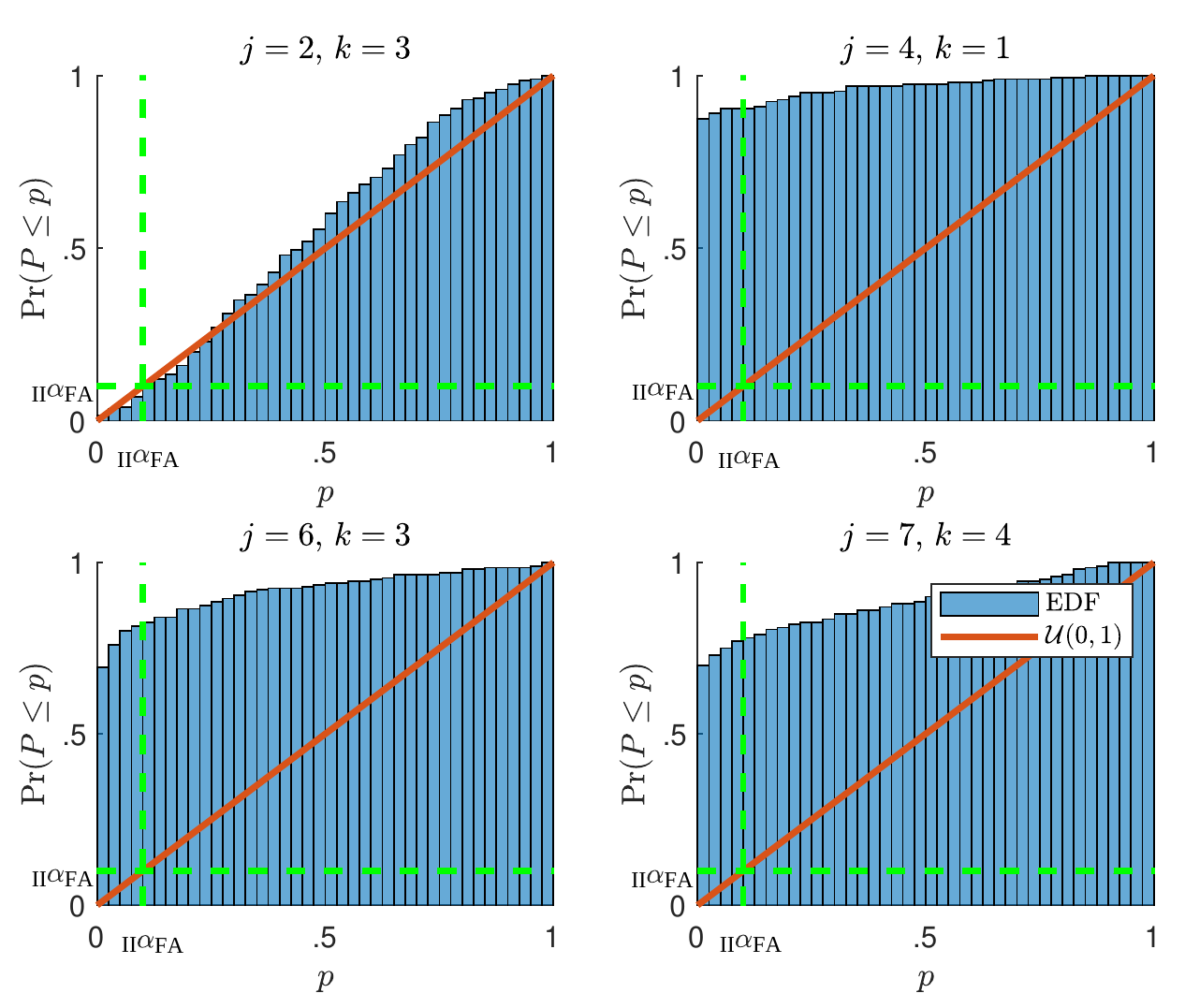}
					\caption{}
					\label{fig:chunk-ex-edf}
				\end{subfigure}
				\caption{(a): The average chunk norms extracted from four data sets with Gaussian components and noise, \(\mathrm{SNR} = 5\unit{dB}\) and the correlation structure given in Fig.~\ref{fig:cor-struc-ex}. The chunk norms associated with components that are uncorrelated across all sets are far from zero, which invalidates the assumption that the chunk norms are close to zero under \(\HSetNul{\srcIdx}\) testing procedure in Step~II of the \gls{ts}. -- (b): The \pval~\glspl{edf} for different sets and components. \(\HSetNul{\srcIdx}\) holds in the top left and since \(\srcIdx\leq \numActSrc\), false positives can be controlled at a typical nominal level \(\alpFAStTwo = .1\). The top right displays the \pval~\gls{edf} under \(\HSetAlt{\srcIdx}\), indicating a high probability of detection. The bottom row underlines that \(\alpFAStTwo\) is significantly violated if \(\srcIdx>\numActSrc\).}
				\label{fig:chunk-ex}
			\end{figure}

			\begin{figure}
				\centering
				\begin{subfigure}{.48\textwidth}
					\centering
					\includegraphics[scale=.5]{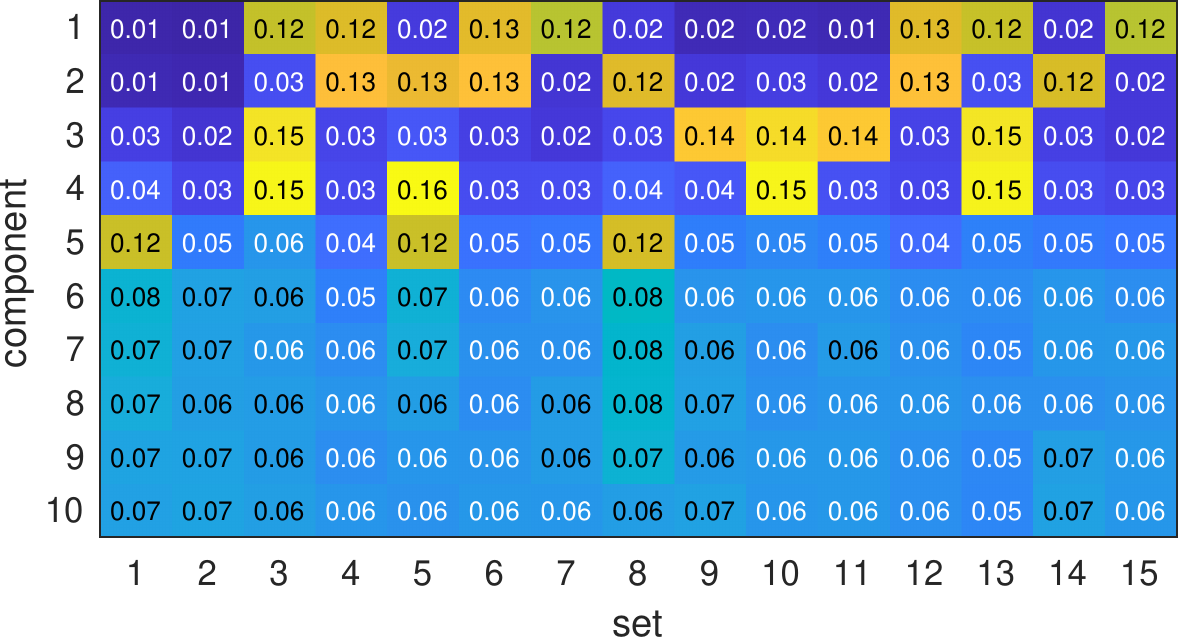}
					\caption{}
					\label{fig:fn-ex-hm}
				\end{subfigure}
				\begin{subfigure}{.48\textwidth}
					\centering
					\begin{subfigure}{\textwidth}
						\centering
						\includegraphics[scale=.5]{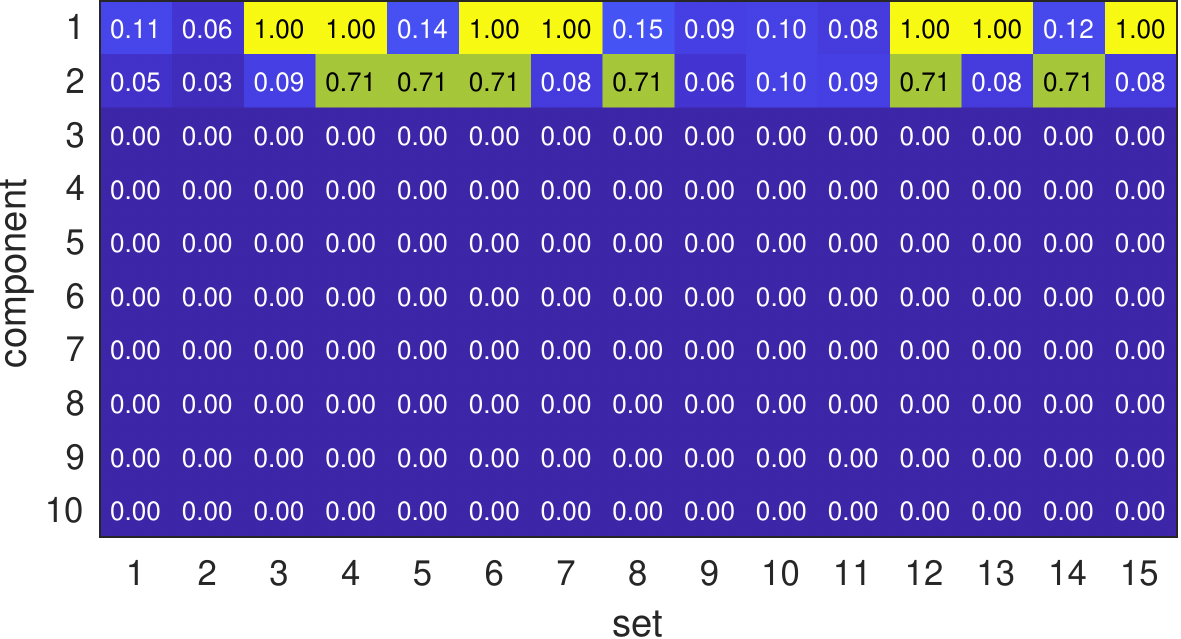}
						\caption{}
						\label{fig:fn-corr-est}
					\end{subfigure}
				\end{subfigure}

				\caption{(a) The average chunks norms for a scenario where the \(\srcIdx\)th component is correlated across \(7-\srcIdx\) sets, \(\srcIdx\in[\numActSrc]\), \(\numActSrc= 6\). The strength of the correlation also decreases with increasing component index \(\srcIdx\). The chunk norms associated with sets across which the \(\srcIdx\)th component is correlated are clearly visible \(\srcIdx\in[5]\). -- (b) The detection probabilities for \(\alpFAStOne = \alpFAStTwo=.1\). Step~I of the \gls{ts} yields on average a value of \(1.67\) for \(\numActSrcEst\). No correlations beyond the second component are identified.
				\label{fig:fn-ex}}
			\end{figure}
			\captionsetup{labelsep=colon}

		\subsubsection{False positive control}
		\label{sec:single-tests-multiple-tests}
			Assume now that Step~I yielded the correct number of correlated components, $\numActSrcEst= \numActSrc$. Hence, the detector in Step~II is based on test statistics $\TSetRVar{\srcIdx}$ for which the relations from \cite[Theorem~1]{Hasija2020} hold and the false alarm probability is controlled for each test at the nominal level $\alpFAStTwo$. Step~II performs in total $\numActSrcEst\cdot \numSet$ binary tests between $\HSetNul{\srcIdx}$ and $\HSetAlt{\srcIdx}$, $\forall\,\setIdx\in[\numSet], \srcIdx\in[\numActSrcEst]$. Let $\nulFrc$ denote the proportion of true set-wise component null hypotheses. Then, on average, $\nulFrc\cdot\numActSrcEst\cdot\numSet\cdot\alpFAStTwo$ true set-wise component null hypotheses get rejected. The fraction of true null hypotheses $\nulFrc$ is unknown in practice. Thus, if Step~II results in a total number of $\numRej$ \textit{rejections}, i.e., decisions in favor of the alternative, it is impossible to assess how many of the $\numRej$ accepted alternatives are false positives. In fact, all $\numRej$ rejections could be false, despite the individual false alarm probabilities being controlled at level $\alpFAStTwo$. A false positive corresponds to an erroneously identified correlation of components in at least two data sets. Falsely detected correlations can have a significant impact. Consider genome-wide association studies \cite{Uffelmann2021}. Falsely identified correlations between a genome and a bio-marker could trigger intensive follow-up research efforts. Hence, using a detector to identify interesting genomes which produces false positives in an uncontrolled manner can lead to a significant waste of resources \cite{Efron2010}. 
	
\subsection{Multiple hypothesis testing}
	Our proposed method identifies the complete correlation structure in a single step. Thus, it is not subject to the aforementioned error propagation between steps.
	To ensure that the identified correlations come with statistical error guarantees and are thus are meaningful, we resort to \gls{mht} \cite{Tukey1991} false positive measures. The common principle in \gls{mht} is to account for the multiplicity of binary decisions by a correction of the individual test statistics. Depending on the type of correction, \gls{mht} detectors are designed to control statistical performance measures that allow to quantify the reliability of the $\numRej$ rejections. The two most commonly used measures are the \gls{fwer} \cite{Hochberg1987} and the \gls{fdr} \cite{Benjamini1995}. The \gls{fwer} is the probability that at least one of the $\numRej$ discoveries is a false positive, while the \gls{fdr} is the expected fraction of false discoveries among all discoveries. Procedures that control the \gls{fwer} are particularly useful for problems where already a single false positive is very costly, while a missed discovery is less critical. A missed discovery occurs, whenever a decision in favor of the null hypothesis is made while the alternative is true. If the \gls{fwer} is controlled at the nominal level $\alpha$, the probability that at least one of the $\numRej$ discovery is a false discovery is $\leq \alpha$. In contrast, controlling the \gls{fdr} permits more false positives, if more correct discoveries are made: If a procedure controls the \gls{fdr} at nominal level $\alpha$, no more than on average $\alpha\cdot\numRej$ discoveries are false. For correlation structure identification, limiting the probability of a single false positive appears unnecessarily strict. A small fraction of false positives among all positives is sufficient, as this allows identifying more true positives while keeping results trustworthy. Thus, we focus on controlling the \gls{fdr} in this work.

	\subsection{Proposed \gls{mht} problem for correlation structure identification}
	\label{sec:prop-os_test}
	We first define a set of binary hypotheses ,
	\begin{align}
		\label{eq:os-hyp-atom}
		\HNulAtom{\srcIdx}:  &\qquad\setIdx\in\nonCorSet:\text{The $\srcIdx$th components of sets $\setIdx$ and $\setIdx^\prime$ are uncorrelated}, \srcIdx\in[\numSrc], \setIdx, \setIdx^\prime\in[\numSet], \setIdx\neq\setIdx^\prime,\\
		\HAltAtom{\srcIdx}:  &\qquad\setIdx\in\corSet:\text{The $\srcIdx$th components of sets $\setIdx$ and $\setIdx^\prime$ are correlated}, \srcIdx\in[\numSrc], \setIdx, \setIdx^\prime\in[\numSet], \setIdx\neq\setIdx^\prime.
	\end{align}
	We refer to $\HNulAtom{\srcIdx}$ and $\HAltAtom{\srcIdx}$ as the \text{atom null hypothesis} and \text{atom alternative}, respectively, since an \textit{atom} is the smallest indivisible unit. In contrast to the hypotheses defined in Eq.~\eqref{eq:s2-hyp} that are used in the second step of the \gls{ts}, the atom hypotheses do \textit{not} depend on an estimate for the total number of correlated components $\numActSrc$. We infer the true atom nulls and alternatives based on test statistics $\T{\srcIdx}\sim\TPdf{\srcIdx}$ that follow  $\TPdfNull{\srcIdx}$ and $\TPdfAlt{\srcIdx}$ under $\HNulAtom{\srcIdx}$ and $\HAltAtom{\srcIdx}$, respectively. The details on $\T{\srcIdx}$ are provided in Sec.~\ref{sec:test-stat}. Finally, the elements of the activation matrix are estimated as $\actElEst{\setIdx}{\srcIdx} = 0 \, \forall \, \setIdx\in[\numSet], \srcIdx\in[\numSrc]$ where $\HNulAtom{\srcIdx}$ is accepted and $\actElEst{\setIdx}{\srcIdx} = 1$ otherwise. 
	
	In addition, we define a set of binary hypotheses for the components
	\begin{align}
		\label{eq:os-hyp-row}
		\HNulCmp:  \qquad &\text{The $\srcIdx$th component is uncorrelated across all sets, }\srcIdx\in[\numSrc], \corSet=\emptyset,\\
		\HAltCmp:  \qquad &\text{The $\srcIdx$th component is correlated between at least two sets } \setIdx, \setIdx^\prime\in[\numSet], |\corSet| \geq 2, \srcIdx\in[J].
	\end{align}
	The \text{component null hypotheses} $\HNulCmp$ and \text{component alternatives} $\HAltCmp$ are unions of their respective component's atom hypotheses: If all $\HNulAtom{\srcIdx}\forall\,\setIdx\in[\numSet]$ hold for a $j\in[J]$, then $\HNulCmp$ holds as well. 

		To identify the activation matrix $\actMat$, the proposed \gls{os} requires $\numSrc\cdot\numSet$ binary tests. Be $\numRej\leq\numSrc\cdot\numSet$ the total number of times a decision in favor of the alternative is made, or, the number of \textit{atom discoveries}. Naturally, $\numRej = \numTruPos + \numFalPos$ with $\numTruPos$ and $\numFalPos$ the numbers of correct and false atom discoveries. $\numRej$ is observable, but $\numTruPos$ and $\numFalPos$ are not. Our proposed approach attempts to maximize $\E[\numTruPos] = \sum_{j=1}^{K}\sum_{k = 1}^{K}P\big(\actElEst{\setIdx}{\srcIdx} = 1 \big|\actEl{\setIdx}{\srcIdx} = 1\big)$ while controlling 
		the \textit{atom \gls{fdr}}
		\begin{equation}
		\label{eq:fdr}
			\fdr = \E\bigg[\frac{\numFalPos}{\numRej}\bigg],
		\end{equation}
		the expected number of atom false discoveries among all atom discoveries at a nominal level $\fdrThrAtom$. This guarantees that on average at least $(1-\fdrThrAtom)\%$ of the non-zero elements in $\actMatEst$ correspond to true correlations.
		
		In addition to controlling false positives on the atom-level, one may also be interested in controlling the component-level false positives. This is particularly useful if atom-level false positives differ in importance based on the corresponding component. It is often much more critical to avoid accidentally declared correlation between at least two sets for the \(\srcIdx\)th component that is uncorrelated between all sets than accidentally identifying correlation for the \(\srcIdx^\prime\)th component of set \(\setIdx\) if this component is correlated between other sets, \(\setIdx\notin\corSet[\srcIdx^\prime]\neq\emptyset\). We denote the number of component discoveries, false discoveries and true discoveries by $\numRejCmp$, $\numFalPosCmp$ and $\numTruPosCmp$, respectively. Then, the \textit{component \gls{fdr}} is 
		\begin{equation}
			\label{eq:fdr-row}
			\fdrCmp = \E\bigg[\frac{\numFalPosCmp}{\numRejCmp}\bigg],
		\end{equation}
	which we attempt to control at the nominal level $\fdrThrCmp$.


		Finally, we revisit the fact that either none or at least two atom hypotheses must be rejected per component with index $\srcIdx\in[\numSrc]$, since correlation \textit{in between} data sets is considered. This structural property has to be incorporated directly into the testing procedure to guarantee strict control of the atom \gls{fdr} at level $\fdrThrAtom$ and maximize the detection power. If it was enforced only \textit{after} the hypothesis testing procedure has been applied, i.e., by a posteriori accepting the atom null hypothesis $\HNulCmp$ for all atoms of components $\srcIdx\in[\numSrc]$ where $\sum_{k = 1}^{K}\actElEst{\setIdx}{\srcIdx} = |\corSet| = 1$, \gls{fdr} control is lost. If $|\corSet| = 1$ due to missed detection(s), i.e., if the actual number of sets across which the $\srcIdx$th component is correlated is $\geq2$, this posterior cleaning reduces the number of correct positives, thereby increasing the \gls{fdr}. In addition, if $|\corSet| = 1$ while the $\srcIdx$th component is uncorrelated across all sets, then the final \gls{fdr} is smaller than with the results of the \gls{fdr} control procedure and additional discoveries may have been possible.

		The complete \gls{mht}-based correlation structure identification problem is, hence,  
		\begin{equation}
		\label{eq:mht-prob-for}
			\begin{aligned}
				\actMatEst_\text{MHT} &= \argmax_{\hat{\mathbf{M}}} \sum_{\srcIdx = 1}^{\numSrc} \sum_{\setIdx=1}^{\numSet} P\big(\actElEst{\setIdx}{\srcIdx} = 1|\actEl{\setIdx}{\srcIdx} = 1\big),\\
				\mathrm{s.t. } &\qquad\fdr\leq \fdrThrAtom,\\
				&\qquad \fdrCmp\leq \fdrThrCmp,\\
				&\qquad \Bigg(\sum_{k = 1}^{K}\actElEst{\setIdx}{\srcIdx}\Bigg) \neq 1, \forall \,\srcIdx\in[\numSrc].
			\end{aligned}
		\end{equation}

	\subsection{Proposed empirical Bayes solution}
		Eq.~\eqref{eq:mht-prob-for} is a very challenging \gls{mht} problem. To the best of our knowledge, a solution does not yet exist in the open literature. The well-known \gls{bh} \cite{Benjamini1995}, for instance, which computes a $p$-value for each tested null hypothesis and then rank-order the $p$-values to identify those that indicate little support for their associated null hypotheses does neither maximize detection power, nor does it provide control on multiple FDRs simultaneously, nor can it fulfill structural conditions.
	
	We resort to \gls{mht} with \glspl{lfdr} \cite{Efron2005, Efron2010}. The \gls{lfdr} is the empirical Bayes probability for a null hypothesis to hold, given the observed data. In what follows, we design a probability-based MHT approach to solving Eq.~\eqref{eq:mht-prob-for}.
	
	The atom \glspl{lfdr} are
		\begin{equation}
		\label{eq:lfdr-atom}
			P\Big(\HNulAtom{\srcIdx} \big|\mathcal{P}\Big) = \mathsf{lfdr}_\setIdx^{(\srcIdx)} \equiv \mathsf{lfdr}\big(\p{\srcIdx}\big) = \frac{\pi_0}{f_P\big(\p{\srcIdx}\big)}, \qquad \forall\,\setIdx\in[\numSet], \srcIdx\in[\numSrc].
		\end{equation}
		$\nulFrc$ is the proportion of true atom null hypotheses. Random variable $P$ with \gls{pdf} $f_P(p)$ and its i.i.d. realizations $\mathcal{P} = \Big\{\p{\srcIdx}\Big\}_{j\in[J], k \in[K]}$ represent the atom \pval s

		\begin{equation}
			\label{eq:pval-atom}
				\p{\srcIdx} = \int_{\T{\srcIdx}}^{\infty} \TPdfNull{\srcIdx}\mathrm{d}\TSym,\qquad \srcIdx\in[\numSrc], \setIdx\in[\numSet]
		\end{equation}

		The details on the test statistics $\T{\srcIdx}$ and their \glspl{pdf} under the atom null $\TPdfNull{\srcIdx}$ are provided in Sec.~\ref{sec:test-stat}. For now, we assume that a valid set of \pval s $\mathcal{P}$ has been observed.

			Under the null hypothesis, $p$-values follow a uniform distribution. The distribution under the alternative may vary from atom to atom due to different levels of correlation and different probability models for the components. Thus, little is known about the exact shape of the \gls{pdf} \(f_{P|\overline{H}}\big(\p{\srcIdx}\big)\) representing those \pval s \(\p{\srcIdx}\in\big\{\mathcal{P}:\HAltAtom{\setIdx}= 1\big\}\) where the alternative is in place. 
			\begin{equation}
				f_P(\p{\srcIdx}) = \nulFrc + 
				(1-\nulFrc)\cdot f_{P|\overline{H}}\big(\p{\srcIdx}\big).
			\end{equation}
			
			The atom \gls{lfdr} is $\mathsf{lfdr}_\setIdx^{(\srcIdx)}$ the posterior probability that the $\srcIdx$th component is uncorrelated between set $k\in[\numSet]$ and all other sets $\setIdx^\prime\in[K]\setminus \setIdx$. Since the component null $\HNulCmp$ holds for the $j$th component if it is uncorrelated across all sets, the posterior probability of the component null hypotheses can be expressed through the atom \glspl{lfdr},
			\begin{equation}
			\label{eq:lfdr-cmp}
				P\Big(\HNulCmp\big|\mathcal{P}\Big) = \prod_{\setIdx = 1}^{\numSet} \mathsf{lfdr}_\setIdx^{(\srcIdx)}.
			\end{equation}
			In general, if a detector rejects a set of null hypotheses, the average probability of the null across this subset is an estimate for the resulting \gls{fdr} \cite{Efron2010}. Hence, for an activation matrix estimate $\actMatEst$, the estimated atom and component \glspl{fdr} are 
			\begin{equation}
				\widehat{\fdr}^\ast = \frac{\sum_{\srcIdx = 1}^{\numSrc}\sum_{\setIdx=1}^{\numSet}\actElEst{\setIdx}{\srcIdx}\cdot \mathsf{lfdr}_\setIdx^{(\srcIdx)}}{{\sum_{\srcIdx = 1}^{\numSrc}\sum_{\setIdx=1}^{\numSet}\actElEst{\setIdx}{\srcIdx}}}, \qquad \widehat{\fdr}_\text{cmp}^\ast = \frac{\sum_{\srcIdx = 1}^\numSrc \idc{\sum_{\setIdx=1}^{\numSet}\actElEst{\setIdx}{\srcIdx}>0}\cdot \prod_{\setIdx = 1}^{\numSet} \mathsf{lfdr}_\setIdx^{(\srcIdx)}}{\sum_{\srcIdx = 1}^\numSrc \idc{\sum_{\setIdx=1}^{\numSet}\actElEst{\setIdx}{\srcIdx}>0}}.
			\end{equation}
			In the definition of \(\widehat{\fdr}^\ast\) and \(\widehat{\fdr}^\ast_\text{cmp}\), we use that $\actElEst{\setIdx}{\srcIdx} = 1$ whenever the atom null hypothesis $\HNulAtom{\setIdx}$ is rejected.
			We propose to exploit this relation between \glspl{lfdr} and \glspl{fdr} to simultaneously control the \glspl{fdr} on the atom and component level. The objective is to detect as many correlations as possible while controlling the atom and component \glspl{fdr} at the respective nominal levels $\fdrThrAtom$ and $\fdrThrCmp$. Thus, we search for the activation matrix estimate $\actMatEst^\ast$ with the largest number of non-zero entries such that a) the average null probability across the non-zero entries is $\leq\fdrThrAtom$ and b) the average component null probability across the set of all components for which correlation between some data sets was identified is $\leq\fdrThrCmp$. The resulting activation matrix estimator is
			\begin{equation}
				\label{eq:sol-fdrs}
				\actMatEst^\ast = \argmax_{\actMatEst}\Bigg\{\sum_{\srcIdx = 1}^{\numSrc}\sum_{\setIdx=1}^{\numSet}\actElEst{\setIdx}{\srcIdx}\,:\,\widehat{\fdr}^\ast \leq \fdrThrAtom, \,\widehat{\fdr}^\ast_\text{cmp}\leq \fdrThrCmp\Bigg\}.
			\end{equation}

			To include the constraint of either none or at least two atom discoveries per row, i.e., $\sum_{\setIdx = 1}^{\numSet}\neq 1$, we propose the following procedure. 
			We define the \textit{modified atom \gls{lfdr}}, which replaces the smallest two atom \glspl{lfdr} by their average, 
			\begin{equation}
				\label{eq:lfdr-mod}
				\widetilde{\mathsf{lfdr}}_\setIdx^{(\srcIdx)} = \begin{cases}
					\frac{\mathsf{lfdr}_\setIdx^{(\srcIdx)} + \min_{\setIdx^\prime\in[\numSet]\setminus\setIdx}\mathsf{lfdr}_{\setIdx^\prime}^{(\srcIdx)}}{2} &\text{if } \setIdx=\arg\min_{\setIdx^{\prime\prime}\in[\numSet]}\mathsf{lfdr}_{\setIdx^{\prime\prime}}^{(\srcIdx)}\text{ or }\setIdx=\arg\min_{\setIdx^{\prime\prime}\in[\numSet]\setminus\setIdx^\prime}\mathsf{lfdr}_{\setIdx^{\prime\prime}}^{(\srcIdx)}\\
					\mathsf{lfdr}_\setIdx^{(\srcIdx)}, & \text{otherwise}.
				\end{cases}
			\end{equation}
			Then, we replace $\mathsf{lfdr}_\setIdx^{(\srcIdx)}$ in Eqs.~\eqref{eq:fdr}, \eqref{eq:fdr-row} by $\widetilde{\mathsf{lfdr}}_\setIdx^{(\srcIdx)}$ to obtain  

			\begin{equation}
				\widehat{\fdr} = \frac{\sum_{\srcIdx = 1}^{\numSrc}\sum_{\setIdx=1}^{\numSet}\actElEst{\setIdx}{\srcIdx}\cdot \widetilde{\mathsf{lfdr}}_\setIdx^{(\srcIdx)}}{{\sum_{\srcIdx = 1}^{\numSrc}\sum_{\setIdx=1}^{\numSet}\actElEst{\setIdx}{\srcIdx}}}, \qquad \widehat{\fdr}_\text{cmp} = \frac{\sum_{\srcIdx = 1}^\numSrc \idc{\sum_{\setIdx=1}^{\numSet}\actElEst{\setIdx}{\srcIdx}>0}\cdot \prod_{\setIdx = 1}^{\numSet} \widetilde{\mathsf{lfdr}}_\setIdx^{(\srcIdx)}}{\sum_{\srcIdx = 1}^\numSrc \idc{\sum_{\setIdx=1}^{\numSet}\actElEst{\setIdx}{\srcIdx}>0}}.
			\end{equation}

			Then, the solution to Eq.~\eqref{eq:mht-prob-for} is 

			\begin{equation}
				\label{eq:sol-all-cond}
				\actMatEst_\text{MHT} = \argmax_{\actMatEst}\Bigg\{\sum_{\srcIdx = 1}^{\numSrc}\sum_{\setIdx=1}^{\numSet}\actElEst{\setIdx}{\srcIdx}\,:\,\widehat{\fdr} \leq \fdrThrAtom, \,\widehat{\fdr}_\text{cmp}\leq \fdrThrCmp\Bigg\}.
			\end{equation}
			Since $\widehat{\fdr}\leq \widehat{\fdr}^\ast\leq\fdrThrAtom$ and $\widehat{\fdr}_\text{cmp}^\ast \leq \widehat{\fdr}_\text{cmp}\leq\fdrThrCmp$, atom and component \gls{fdr} control carries over from Eq.~\eqref{eq:sol-fdrs}.

			In practice, $\actMatEst_\text{MHT}$ can be determined as follows. The given set of \pval s $\mathcal{P}$ contains one \pval~per atom, which quantifies the evidence that the $\srcIdx$th component of set $\setIdx$ is correlated with at least one other set. First compute the atom \glspl{lfdr} from Eq.~\eqref{eq:lfdr-atom} and the modified \glspl{lfdr} from Eq.~\eqref{eq:lfdr-mod}. Sort them in ascending order. Then, find the  index $m\in[\numSet\cdot \numSrc]$ as the largest integer such that the cumulative average over the smallest $m$ modified \gls{lfdr}'s is below the nominal atom \gls{fdr} level $\fdrThrAtom$. If the $m$-th largest $\widetilde{\mathsf{lfdr}}_\setIdx^{(\srcIdx)}$ is equal to its corresponding atom \gls{lfdr}, rejecting the null hypothesis for those atoms corresponding to the smallest $m$ modified \gls{lfdr}'s guarantees that at least two discoveries per component are made. If the $m$-th largest modified \gls{lfdr} is different from its corresponding \gls{lfdr}, then this atom is one of the atoms with the two smallest \glspl{lfdr} of a component. One then has to make sure that either the null hypothesis for this second atom with one of the smallest \gls{lfdr}'s of that component gets rejected as well, or that none of the two get rejected.
		
	The component \gls{fdr} is estimated subsequently by averaging the component null probabilities over those components for which correlations have been identified. The resulting estimate $\widehat{\fdr}_\text{cmp}$ is compared to the nominal component \gls{fdr} level $\fdrThrCmp$. If $\fdrCmp\leq\fdrThrCmp$, both atom and component \gls{fdr} are controlled as desired and the estimate $\actMatEst_\text{MHT}$ that solves Eq.~\eqref{eq:mht-prob-for} has been found. If, on the other hand, $\fdrCmp>\fdrThrCmp$, we iteratively remove components from the set of discoveries. In each iteration, we remove the discoveries for the component with the smallest number of atom discoveries to reduce the component \gls{lfdr} while minimizing the loss in detection power on the atom level. Then, we again rank order the atom-level modified \glspl{lfdr} from the remaining components for which correlation has been detected. This leaves room for additional discoveries, since the removal of discoveries for one component has also reduced $\widehat{\fdr}$. The procedure terminates as soon as $\fdrCmp$ falls below the nominal level $\fdrThrCmp$. The details are given in Alg.~\ref{alg:proposed-lfdr-atom-cmp}.

			\begin{algorithm}
				\caption{The proposed \gls{lfdr}-based correlation structure detector with atom and component \gls{fdr} control}
				\label{alg:proposed-lfdr-atom-cmp}
				\hspace*{\algorithmicindent} \textbf{Input}: $\mathsf{lfdr}_\setIdx^{(\srcIdx)}\,\forall\srcIdx\in[\numSrc],\setIdx\in[\numSet]$, $\fdrThrAtom$, $\fdrThrCmp$\\
				\hspace*{\algorithmicindent} \textbf{Output}: The activation matrix estimate $\actMatEst_\text{MHT}$.
				\begin{algorithmic}[1]
					\vspace{1pt}
					\State Initialize $\mathcal{J}_\text{cand} = [\numSrc]$
					\State Compute all $\widetilde{\mathsf{lfdr}}_\setIdx^{(\srcIdx)}$ from Eq.~\eqref{eq:lfdr-mod}
					\State Define $|\mathcal{J}_\text{cand}|\cdot \numSet$ tuples of component and set indexes $(\srcIdx_\mu, \setIdx_\mu)\in\mathcal{J}_\text{cand}\times [\numSet]$ s.t. $\widetilde{\mathsf{lfdr}}_{\setIdx_\mu}^{(\srcIdx_\mu)}$ is the $\mu$th largest modified atom \gls{lfdr}
					\State Find $m = \max\Big\{{\mu}\in[|\mathcal{J}_\text{cand}|\cdot\numSet]:\mu^{-1}\cdot \sum_{m = 1}^{\mu}\widetilde{\mathsf{lfdr}}_{\setIdx_\mu}^{(\srcIdx_\mu)}\leq\fdrThrAtom\Big\}$
\If{$\widetilde{\mathsf{lfdr}}_{k_{m}}^{(j_m)}\neq{\mathsf{lfdr}}_{k_{m}}^{(j_m)}$}
					\If{$\nexists\, \mu<m:j_\mu = j_m$}
					\State $m = m \pm 1$ where $+$ is selected at random with probability $.5$
					\EndIf
					\EndIf
					\State Define the preliminary set of rejected atom null hypotheses $\overline{\mathcal{H}} = \Big\{\HAltAtom[\setIdx_\mu]{\srcIdx_\mu}\Big\}_{\mu\leq m}$
					\State Find the set of component rejections $\overline{\mathcal{H}}_\text{cmp} = \Big\{\HAltCmp\,\forall\, \srcIdx\in\mathcal{J}_\text{cand}: \exists\, \HAltAtom{\srcIdx}\in\overline{\mathcal{H}}\Big\}$
					\State Compute $\widehat{\fdr}_\text{cmp} = \sum_{j:\HAltCmp\in\overline{\mathcal{H}}_\text{cmp}}\prod_{\setIdx = 1}^{\numSet} \widetilde{\mathsf{lfdr}}_\setIdx^{(\srcIdx)} $
					\If{$\widehat{\fdr}_\text{cmp}>\fdrThrCmp$}
					\State Define $\nu = \min_{\srcIdx:\HAltCmp\in\overline{H}_\text{cmp}}\sum_{\setIdx\in[\numSet]}\idc{\HAltAtom{\srcIdx}\in\overline{H}}$
					\State Find $\srcIdx^\ast= \arg\max_{\srcIdx:\HAltCmp\in\overline{\mathcal{H}}_\text{cmp}}\Big\{\prod_{\setIdx = 1}^{\numSet} \widetilde{\mathsf{lfdr}}_\setIdx^{(\srcIdx)}:\sum_{\setIdx\in[\numSet]}\idc{\HAltAtom{\srcIdx}\in\mathcal{H}} = \nu\Big\}$
					\State $\mathcal{J}_\text{cand} = \mathcal{J}_\text{cand}\setminus j^\ast$
					\State Jump to Step 3 with the remaining $\widetilde{\mathsf{lfdr}}_\setIdx^{(\srcIdx)}, \setIdx\in[\numSet], \srcIdx\in\mathcal{J}_\text{cand}$
					\EndIf
					\State Determine the entries of $\actMatEst_\text{MHT}$ with $\actElEst{\setIdx}{\srcIdx} = \idc{\HAltAtom{\srcIdx}\in \overline{\mathcal{H}}} \forall\,\setIdx \in[\numSet], \srcIdx\in[\numSrc]$
				\end{algorithmic}
			\end{algorithm}

			So far, we have assumed that a set $\mathcal{P}$ of \pval s was provided. In the following section, we present our proposed test statistic that extracts the evidence for correlation between components across data sets from the coherence matrix of the data.

			\subsection{The proposed test statistic}
			\label{sec:test-stat}
			We exploit the properties of the eigenvectors of the composite coherence matrix $\cmpCohMat$ from \cite{Hasija2020} that we summarized in Sec.~\ref{sec:ts-prc}. Assume that $\sum_{\setIdx = 1}^{\numSet}\actEl{\setIdx}{\srcIdx}\neq 0$, i.e., there is at least one pair of sets $\setIdx, \setIdx^{\prime}\in[\numSet], \setIdx\neq\setIdx^{\prime}$ such that the $\srcIdx$th component of set $\setIdx$ is correlated with the $\srcIdx$th component of set $\setIdx^{\prime}$. Then, according to \cite[Theorem~2]{Hasija2020}, $\eigVecChk{\srcIdx} = \mathbf{0}_{\numSrc}$ if and only if $\actEl{\setIdx}{\srcIdx} = 0$  and $\big|\big|\eigVecChk{\srcIdx}\big|\big|^2 = 0$. In contrast, if $\actEl{\setIdx}{\srcIdx}= 1$, $\big|\big|\eigVecChk{\srcIdx}\big|\big|^2 > 0$ holds. In addition, for any eigenvector $\boldsymbol{\eigVecSym}$ of any arbitrary matrix  general, $||\boldsymbol{\eigVecSym}||^2 = 1$ holds.
			
			Since the true $\cmpCohMat$ is not available, our access is limited to the eigenvector chunks $\eigVecChkEst{\srcIdx}$ of the estimate $\cmpCohMatEst$. $\cmpCohMatEst$ is estimated from finite sample data and is hence subject to random estimation error. Due to this noise, its eigenvector chunks $\eigVecChkEst{\srcIdx}\neq\mathbf{0}_\numSrc\,\forall\,\srcIdx\in[\numSrc], \setIdx\in[\numSet]$ and thus also $\big|\big|\eigVecChkEst{\srcIdx}\big|\big|^2\neq0\,\forall\,\srcIdx\in[\numSrc], \setIdx\in[\numSet]$. In what follows, we denote the squared estimated eigenvector chunk norm by $\eigVecChkNorm{\srcIdx} = \big|\big|\eigVecChkEst{\srcIdx}\big|\big|^2$. 
			
			\begin{theorem}
			\label{theo:equal-chunk-norms-h0}
				Consider that the $\srcIdx$th components of all data sets are uncorrelated, i.e., $\srcIdx\notin[\numActSrc]$. Then, the expected value of the squared chunk norms $\E\Big[\eigVecChkNorm{\srcIdx}\Big]$ is identical $\forall\,\setIdx\in[\numSet]$. With $\left|\left|\eigVecEst\right|\right|^2 = 1$ follows
				\begin{equation}
					\label{eq:equal-chunk-norms}
					\E\left[\eigVecChkNorm[1]{\srcIdx}\right] = \dots = \E\left[\eigVecChkNorm[\numSet]{\srcIdx}\right] = \frac{1}{\numSet}.
				\end{equation}
				The same relation holds, if the $\srcIdx$th components of all sets are correlated with an equal correlation coefficient.
			\end{theorem}
		
			\begin{proof}
				The norm of any eigenvector of any arbitrary matrix is one. If the $\srcIdx$th component is uncorrelated across all sets, its eigenvector $\eigVecEst$ does not exhibit a specific structure. The same holds, if $\srcIdx$th component is equally strongly correlated across all data sets. Then, its elements follow a maximum entropy distribution with zero mean and equal variance \cite{Laloux1999, Plerou2002}, i.e., they are uniformly distributed on the $\numSrc\cdot \numSet$-dimensional unit sphere \cite{Bai2007}. Hence, the expected value of the $\numSrc$ sums of distinct $\numSet$ entries of $\eigVecEst$ all have the same expected value and must sum to one. 	
			\end{proof}
			
			In the examples of Fig.~\ref{fig:chunk-ex} and Fig.~\ref{fig:fn-ex}, Theorem~\ref{theo:equal-chunk-norms-h0} can also be observed empirically.

			\begin{conj}
				\label{conj:exp-of-chunk-norms-cor-cmp}
				Consider that the $\srcIdx$th component of set $\setIdx\in[\numSet]$ is correlated with the $\srcIdx$th component of at least one other set $\setIdx^{\prime}\in[\numSet], \setIdx\neq\setIdx^{\prime}$, $\srcIdx\in[\numSrc]$. $\corSet$ denotes the set of all data sets whose $\srcIdx$th components are correlated and $\nonCorSet$ its complement. Then,
				\begin{equation}
					\E\left[\eigVecChkNorm[\setIdx]{\srcIdx}\right] > \E\left[\eigVecChkNorm[\setIdx^{\prime}]{\srcIdx}\right] \qquad\forall\,\setIdx\in\corSet, \setIdx^{\prime}\in\nonCorSet.
				\end{equation}
				Since $\left|\left|\eigVecEst\right|\right|^2 = 1$, $\E\left[\eigVecChkNorm[\setIdx^{\prime}]{\srcIdx}\right] < \frac{1}{\numSet}$ holds for $\setIdx^{\prime}\in\nonCorSet$. 
			\end{conj}
			
			This conjecture follows from Theorem~\ref{theo:equal-chunk-norms-h0} in combination with \cite[Theorem~2]{Hasija2020}. The eigenvectors of the estimated coherence matrix $\eigVecEst$ are noisy versions of the eigenvectors $\eigVec$ of the true unavailable $\cmpCohMat$. For the true $\eigVecChk{\srcIdx} = \mathbf{0}_{\numSrc}$ if and only if $\actEl{\setIdx}{\srcIdx} = 0$  and $\big|\big|\eigVecChk{\srcIdx}\big|\big|^2 = 0$. In contrast, if $\actEl{\setIdx}{\srcIdx}= 1$, $\big|\big|\eigVecChk{\srcIdx}\big|\big|^2 > 0$ holds. The chunk norms of the eigenvectors $\eigVecChkNorm[\setIdx]{\srcIdx}$ of $\cmpCohMatEst$ are noisy versions of the $\big|\big|\eigVecChk{\srcIdx}\big|\big|^2$. Conjecture \ref{conj:exp-of-chunk-norms-cor-cmp} states that these estimates, while not perfectly zero for uncorrelated sets, are expected to have a smaller expected value than those chunk norms associated with true correlations.

	We propose to exploit these differences in the expected value of the chunk norms to identify the true correlations. Hence, we deploy the test statistic
			\begin{equation}
				\label{eq:test-stat}
				\T{\srcIdx} = {\eigVecChkNorm{\srcIdx} - \E\left[\eigVecChkNorm[\setIdx]{\srcIdx}\right]} = {\eigVecChkNorm{\srcIdx} - \mu_\setIdx^{(\srcIdx)}},\qquad \forall\,\srcIdx\in[\numSrc], \setIdx\in[\numSet].
			\end{equation}
				$\chkNormExp{\srcIdx}$ denotes the expectation $\eigVecChkNorm{\srcIdx}$. We know that $\chkNormExp{\srcIdx}\leq \frac{1}{\numSet}$ under the atom null hypothesis, but its exact value depends on the underlying data structure. In addition, to compute the atom \pval s $\p{\srcIdx}$ from Eq.~\eqref{eq:pval-atom}, which are needed as inputs for our \gls{mht} correlation structure detector in Alg.~\ref{alg:proposed-lfdr-atom-cmp}, the \glspl{pdf} $\TPdfNull{\srcIdx}$ under the atom null hypothesis is required for all $\T{\srcIdx}\, \forall \,\setIdx\in[\numSet], \srcIdx\in[\numSrc]$. 
				The literature on the distributional properties of the spectrum of finite-sample second order statistic matrices is limited. In general, tools from \gls{rmt} need to be applied. While some work on estimators for the eigenvalues and eigenvectors of random covariance matrices exist, e.g. \cite{Mestre2008b}, the results in the literature concerning the statistical properties of the eigenvectors are not applicable to our problem at hand. Often, the assumptions on the underlying component distributions are too restrictive \cite{Forrester2007, Mestre2008a, Simm2017}. The \gls{rmt} overview paper \cite{ORourke2016} provides interesting insights into the distribution of the squared chunk norms under mild conditions. However, their model describes a chunk norm random variable that is marginalized over the corresponding eigenvalue magnitude. In this work, the eigenvalues are sorted such that the first eigenvector corresponds to the largest eigenvalue. The statistical properties of an eigenvector chunk norm conditioned on it being associated with the $\srcIdx$th largest eigenvalue differ from those of an eigenvector chunk norm associated with the $\obsIdx$th largest eigenvalue. We conclude that a general valid analytical model for distribution of the eigenvector chunk norms as utilized in this work does not exist and determining the required distributions analytically is too challenging. Instead, we resort to learning  $\TPdfNull{\srcIdx}$ and $\chkNormExp{\srcIdx}\,\forall\,\srcIdx\in[\numSrc], \setIdx\in[\numSet]$ from the data.

		
			\subsection{Learning the test statistic distribution from the data}
				We have access to exactly one realization of each eigenvector chunk norm $\eigVecChk{\srcIdx}$. Thus, we deploy the bootstrap \cite{Efron1994, Zoubir1998} to obtain artificial realizations that can be used to approximate the underlying probability model. The bootstrap is a standard tool to approximate the distribution of a test statistic under the null hypothesis \cite{Zoubir2001, Golz2017a} for the given number of samples. Hence, the bootstrap is not only useful for estimating theoretically unknown distributions, but can also be applied when a small sample size prohibits the use of asymptotic results \cite{Golz2017b}.
				The bootstrap can be deployed parametrically or non-parametrically. The former assumes a parametric data model, estimates the model parameters from the observations and then resamples from the estimated distribution. The latter resamples directly $\numBtSam$ times with replacement from the observation sample. Since little is known about shape of the distribution of the eigenvector chunk norms, we stick to the non-parametric bootstrap. The details on how we bootstrap the estimated eigenvector chunk norms are provided in Alg.~\ref{alg:chk-norm-bt}.
		
				\begin{algorithm}
					\caption{The eigenvector chunk norm bootstrap}
					\label{alg:chk-norm-bt}
					\hspace*{\algorithmicindent} \textbf{Input}: Observation matrices $\obsMat\,\forall\,\setIdx\in[\numSet]$, $\numBtSam$\\
					\hspace*{\algorithmicindent} \textbf{Output}: Bootstrapped squared eigenvector chunk norms $\eigVecChkNormBt{\setIdx}{\srcIdx}$\\
					\hspace*{\algorithmicindent} \textit{Remark}: The indices take values $\srcIdx\in[\numSrc], \setIdx\in[\numSet]$, $\btIdx\in[\numBtSam]$
					\begin{algorithmic}[1]
						\State Resample $\numBtSam$ times from the rows of $\obsMat$ to obtain $\obsMatBt{\setIdx}$
						\State Compute $\cmpCohMatEstBt$ with $\obsMatBt{\setIdx}$ from Eqs.~\eqref{eq:sam-coh-est}, \eqref{eq:coh-mat}
						\State Find the \gls{evd} of $\cmpCohMatEstBt$: $\eigValEstBt{1}\geq\dots\geq\eigValEstBt{\numSrc\numSet}$ and $\eigVecEstBt{1},\dots,\eigVecEstBt{\numSrc\numSet}$
						\State Compute $\eigVecChkNormBt{\setIdx}{\srcIdx} = \left|\left|\eigVecChkEstBt{\setIdx}{\srcIdx}\right|\right|^2$
					\end{algorithmic}
				\end{algorithm}
			\subsection{Local false discovery rate estimation}
				A serious challenge when working with \gls{lfdr}-based inference methods is that $f_P(p)$ is most often unavailable in practice, since the exact distribution of the $p$-values under the alternative cannot be specified exactly and must hence be learned from the data. Methods to estimate the \gls{lfdr}'s from the data exist in the literature. Under the alternative, $p$-values closer to zero become more likely. Hence, joint $p$-value \glspl{pdf} like $f_P(p)$ have been modeled as a mixture of a uniform and a single-parameter beta distribution component \cite{Pounds2003} or, more recently, as a mixture of multiple single-parameter beta distributions \cite{Goelz2022a, Goelz2022c}. In particular, the spectral method of moments-based \gls{lfdr} estimator proposed in \cite{Goelz2022a} and its maximum likelihood extension \cite{Goelz2022c} enable accurate \gls{lfdr}-based inference with \gls{fdr} control even when only few handfuls of $p$-values are available. Thus, we deploy the \gls{lfdr} estimator LFDR-SMOM-EM from \cite{Goelz2022c} for estimating the \glspl{lfdr} in this work.
			\subsection{The complete proposed algorithm}
				We present the complete \textbf{\gls{lfdr}}-based \textbf{mu}ltiple hypothesis \textbf{t}esting procedure for complete \textbf{co}rrelation \textbf{st}ructure identification (LFDR-MULT-COST) in detail in Alg.~\ref{alg:lfdr-mult-cost}. First, we compute the test statistics $\T{\srcIdx}\,\forall\,\setIdx\in[\numSet],\srcIdx\in[\numSrc]$ under the assumption that the null hypothesis holds everywhere. Then, we bootstrap the observation matrices to obtain an estimate $\TCdfNullEst{\srcIdx}{\TSym}$ for the \gls{cdf} of the test statistics under the null hypothesis. Subsequently, a \pval~$\p{\srcIdx}$ is computed  $\forall\,\setIdx\in[\numSet], \srcIdx\in[\numSrc]$ to express the confidence in each local null hypothesis $\HNulAtom{\srcIdx}$. Then, the \glspl{lfdr} are estimated. Finally, the proposed correlation structure detector from Sec.~\ref{sec:prop-os_test} is applied. The \gls{fdr} is controlled on the atom and component level, i.e., $\fdr\leq \fdrThrAtom$ and $\fdrCmp\leq \fdrThrCmp$ while $\sum_{k = 1}^{K}\actElEst{\srcIdx}{\setIdx}\neq 1\, \forall \srcIdx\in[\numSrc]$.

				\begin{algorithm}
					\caption{Proposed \gls{lfdrmultcost}}
					\label{alg:lfdr-mult-cost}
					\hspace*{\algorithmicindent} \textbf{Input}: Observation matrices $\obsMat$, $\setIdx\in[\numSet]$, $\numBtSam$, $\fdrThrAtom$, $\fdrThrCmp$\\
					\hspace*{\algorithmicindent} \textbf{Output}: Activation matrix estimate $\actMatEst_\text{MHT}$\\
					\hspace*{\algorithmicindent} \textit{Remark}: The indices take values $\srcIdx\in[\numSrc], \setIdx\in[\numSet]$, $\btIdx_0\in[\numBtSam_0]$, $\btIdx_1\in[\numBtSam_1]$
					\begin{algorithmic}[1]
						\vspace{1pt}
						\Statex\textbf{Step 1:} \textit{Computation of the chunk norms}
						\State Compute $\cmpCohMatEst$ from Eqs.~\eqref{eq:sam-coh-est}, \eqref{eq:coh-mat}
						\State Find the \gls{evd} of $\cmpCohMatEst$: $\eigValEst[1]\geq\dots\geq\eigValEst[\numSrc\numSet]$ and $\eigVecEst[1],\dots,\eigVecEst[\numSrc\numSet]$
						\State Divide $\eigVecEst[\srcIdx]$ into $\numSet$ chunks $\eigVecChkEst{\srcIdx}$
						\State Compute $\eigVecChkNorm{\srcIdx}= \left|\left|\eigVecChkEst{\srcIdx}\right|\right|^2$
						\Statex\textbf{Step 2:} \textit{Estimation of the distribution of the test statistic}
						\State Obtain $\numBtSam$ bootstrapped $\eigVecChkNormBt[\btIdx]{\setIdx}{\srcIdx}$ from Alg.~\ref{alg:chk-norm-bt}
						\State Estimate ${}^{\btIdx}{\chkNormExp{\srcIdx}}^\ast$ by the sample means of the $\eigVecChkNormBt[\btIdx]{\setIdx}{\srcIdx}$
						\State Compute the bootstrapped $\TBt[\btIdx]{\setIdx}{\srcIdx} = {\eigVecChkNormBt[\btIdx]{\setIdx}{\srcIdx} - {}^{\btIdx}{\chkNormExp{\srcIdx}}^\ast}$
						\State Form the null distribution function estimates $\TCdfNullEst{\srcIdx}{\TSym}$ by sorting the bootstrapped test statistics in ascending order	$\TBt[(1)]{\setIdx}{\srcIdx}\leq \dots\leq \TBt[(\numBtSam)]{\setIdx}{\srcIdx}$

						\Statex\textbf{Step 3:} \textit{Computation of the test statistic}
						\State Compute $\T{\srcIdx}|\HNulAtom{\srcIdx}$ from Eq.~\eqref{eq:test-stat} with $\hat{\mu}_\setIdx^{(\setIdx)}|\HNulAtom{\setIdx} = \min\Big[\frac{1}{B}\cdot \sum_{b = 1}^B \eigVecChkNormBt[\btIdx]{\setIdx}{\srcIdx}, \frac{1}{\numSet}\Big] $ 
						\Statex\textbf{Step 4:} \textit{Hypothesis testing}
						\State Compute the \pval s $\p[\setIdx]{\srcIdx} = \TCdfNullEst{\srcIdx}{\T{\srcIdx}}$
						\State Estimate the atom \glspl{lfdr} $\mathsf{lfdr}_\setIdx^{(\srcIdx)}$ with \cite[Alg.~2]{Goelz2022c}
						\State Find the activation matrix estimate $\actMatEst_\text{MHT}$ with Alg.~\ref{alg:proposed-lfdr-atom-cmp}
					\end{algorithmic}
				\end{algorithm}

	\section{Simulation Results}
\label{sec:sim-res}
In this section, we numerically evaluate the performance of the proposed \gls{lfdr}-based multiple testing approach to complete correlation structure identification. We simulate a variety of scenarios. Our results underline that our proposed method is very effective in identifying correlations. In particular in challenging scenarios, such as when data sets are high-dimensional and only few samples are available, it outperforms its competitors significantly in terms of detection power, while the \gls{fdr} is controlled at the nominal level. We compare the proposed method's performance to the two existing approaches for multiset correlation structure identification from the literature. These are the previously summarized \gls{ts} from \cite{Hasija2020}, which first estimates the number of components $\numActSrc$ correlated between at least some sets before identifying the precise sets for which correlation exists. We use equal nominal false alarm probability levels $\alpFAStOne=\alpFAStTwo=0.1$. 
For \gls{lfdrmultcost}, we present results for $\fdrThrAtom = \fdrThrCmp = 0.1$ and for $\fdrThrAtom = .1, \fdrThrCmp = 1$. While the former combination uses the widely used nominal \gls{fdr} level of $0.1$ for atoms and components alike, the latter combination allows for any proportion of false discoveries on the component level. Throughout our experiments, the achieved detection power is very similar for $\fdrThrCmp = 0.1$ and $\fdrThrCmp = 1$. Hence, the additional statistical guarantee, which increases the credibility for the identified correlation structure by controlling the \gls{fdr} on the component level versus having \gls{fdr} control only on the atom level comes at little cost. The second competitor originates from \cite{Marrinan2018}, where the underlying components are estimated via \gls{mcca} before each possible pair of components is tested for correlation. This results in a large number of tests if the number of data sets is large. Thus, we use a false alarm level of $0.001$, as the authors of \cite{Marrinan2018} suggested. 

The synthetic data is generated according to the model in Eq.~\eqref{eq:sam-sig-mod}. We randomly generate orthogonal mixing matrices $\mixMat$. The components in each data set have an equal variance of $1$. We use both, Gaussian and Laplacian distributed components to illustrate the insensitivity of our method to the underlying data distributions. The additive noise is i.i.d. Gaussian, unless specified otherwise. The noise variance is constant across data sets and computed from the \gls{snr} and the component variance. For the bootstrap, we use $B = 300$ resamples, which is sufficient in most applications \cite{Zoubir2001}.

To quantitatively evaluate the performance, we monitor the empirically obtained \gls{fdr} on the atom and component level, which should both not exceed their respective nominal levels for our proposed methods. In addition, we compare the detection power of the different methods, that is, the percentage of detected true correlations, both, on the atom and the component level. Finally, we also provide some exemplary plots of the averaged estimated activation matrices $\actMatEst$ in comparison to the true activation matrix $\actMat$. This is useful for understanding the origin of the performance differences. We analyze the behavior in dependence of \gls{snr}, sample size, number of data sets, proportion of true correlations. In addition, we investigate the impact of outliers through $\epsilon$-contaminated noise \cite{Huber2009} in the Appendix \ref{app:res}. The underlying correlation structures vary in the different experiments. All results are averaged over $100$ independent repetitions.

In Experiment~1, we revisit the example from Fig.~\ref{fig:fn-ex}. \(\numActSrc = 6\) out of $\numSrc = 10$ components are correlated across \(7, 6, 5, 4, 3\) and \(2\) out of $\numSet = 15$ data sets with correlation coefficients \(7, .7, .65, .6, .6, .55\), respectively. The correlation coefficients are constant per component, \(\corCoeffSrc[\setIdx, \setIdx^\prime]{\srcIdx} = \corCoeffSrc[\setIdx^{\prime\prime}, \setIdx^{\prime\prime\prime}]{\srcIdx}\,\forall\,\setIdx, \setIdx^{\prime}, \setIdx^{\prime\prime}, \setIdx^{\prime\prime\prime}\in[\numSet]\).

For the remaining experiments, the correlation structure be randomized as follows. We define $\pi_0$ as the fraction of zeros in $\actMat$, that is, the proportion of true atom level null hypotheses among all hypotheses. Based on $\pi_0$, $\numActSrc$ is computed such that the $\numActSrc$th component is correlated across at least two sets. The $\numActSrc-1$st component is correlated across one more set, the $\numActSrc-2$nd again correlated across one more etc. The sets across which a component is correlated are selected uniformly at random. The precise value of $\numActSrc$ depends on $\pi_0$, $\numSet$ and $\numSrc$. This randomization of the correlation allows to obtain results for a variety of different correlation structures while guaranteeing its identifiability via the coherence matrix \cite{Hasija2020}. In addition, the parameter $\pi_0$ allows to flexibly tune the sparsity of the true correlations, that is, the number of non-zero entries of activation matrix $\actMat$. 

We define an average correlation coefficient of $0.85$ for the first component with the strongest correlation and $0.5$ for the $\numActSrc$th component with the weakest correlation. The average correlation decreases linearly as $\srcIdx\in[\numActSrc]$ grows. We sample the correlation coefficients $\corCoeffSrc{\srcIdx}$ from a Gaussian distribution with expectation $\mu_\rho = 0.85 - \frac{(0.85-0.5)}{\numActSrc-1}\cdot (\srcIdx - 1)$ and a standard deviation $\sigma_\rho = 0.33 \cdot \frac{(0.85-0.5)}{\numActSrc-1}$. The objective was to create a challenging, yet solvable scenario. We found that correlations with correlation coefficient $\leq 0.5 $ are barely identifiable with any of the deployed methods for the small sample sizes we consider in this work. The randomness in the value of the correlation coefficient values attributes to the real-world, where the correlation strength between the $\srcIdx$th component of different sets may vary.

\subsubsection*{Experiment~1} In this experiment, the number of samples \(\numSam = 300\) is small. The top and bottom rows of Fig.~\ref{fig:exp-snr_perf} shows the performance measures on the atom and component level as a function of \gls{snr}, respectively. The atom \gls{fdr} for mCCA-HT is higher than the axis limit $0.3$. Such a high proportion of false discoveries makes the results of unreliable. This is also confirmed in the detection pattern for mCCA-HT on the right top of Fig.~\ref{fig:exp-snr_examples} for $\mathrm{SNR} = 10 \unit{dB}$. Many false positives occur with high probability. \gls{ts} produces less false positives than mCCA-HT, but is also far more conservative, i.e., finds considerably less true correlations for $\mathrm{SNR}\geq 5\unit{dB}$. As the bottom left of Fig.~\ref{fig:exp-snr_examples} illustrates, this is due to an early termination in Step~I. Indeed, for the displayed detection pattern with $\mathrm{SNR}=10\unit{dB}$, the average number of correlated components estimated by \gls{ts} is $\numActSrcEst = 2$. This is far from the true $\numActSrc = 6$. Hence, the error propagation from Step~I to Step~II causes low detection power. Our proposed \gls{lfdrmultcost} yields the best performance. The empirical \gls{fdr} is well below the nominal level and the detection power is high. Our proposed single-step approach excels in particular for stronger signals, due to its ability to freely detect true discoveries in all components. Removing the constraint on the component level \gls{fdr} by setting $\fdrThrCmp = 1$ has little impact on the detection power. While increases the component level \gls{fdr}, the empirical $\fdrCmp$ remains low.

\begin{figure}
	\begin{subfigure}{.54\textwidth}
		\begin{subfigure}{.49\textwidth}
			\includegraphics[scale=.28]{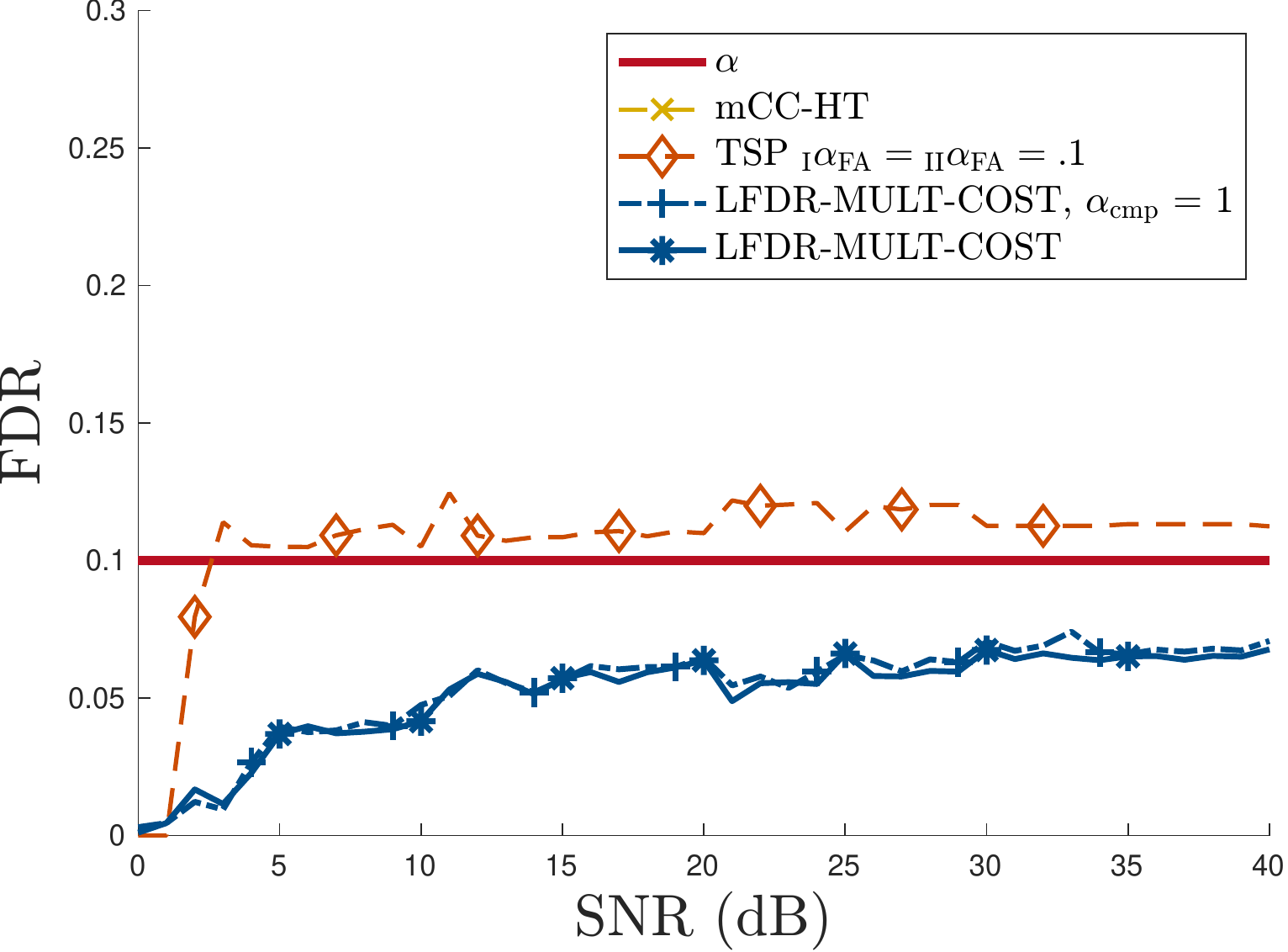}\smallskip
			\caption*{\(\fdr\)}
			\label{fig:exp-snr_atom-fdr}
		\end{subfigure}
		\begin{subfigure}{.49\textwidth}
			\includegraphics[scale=.28]{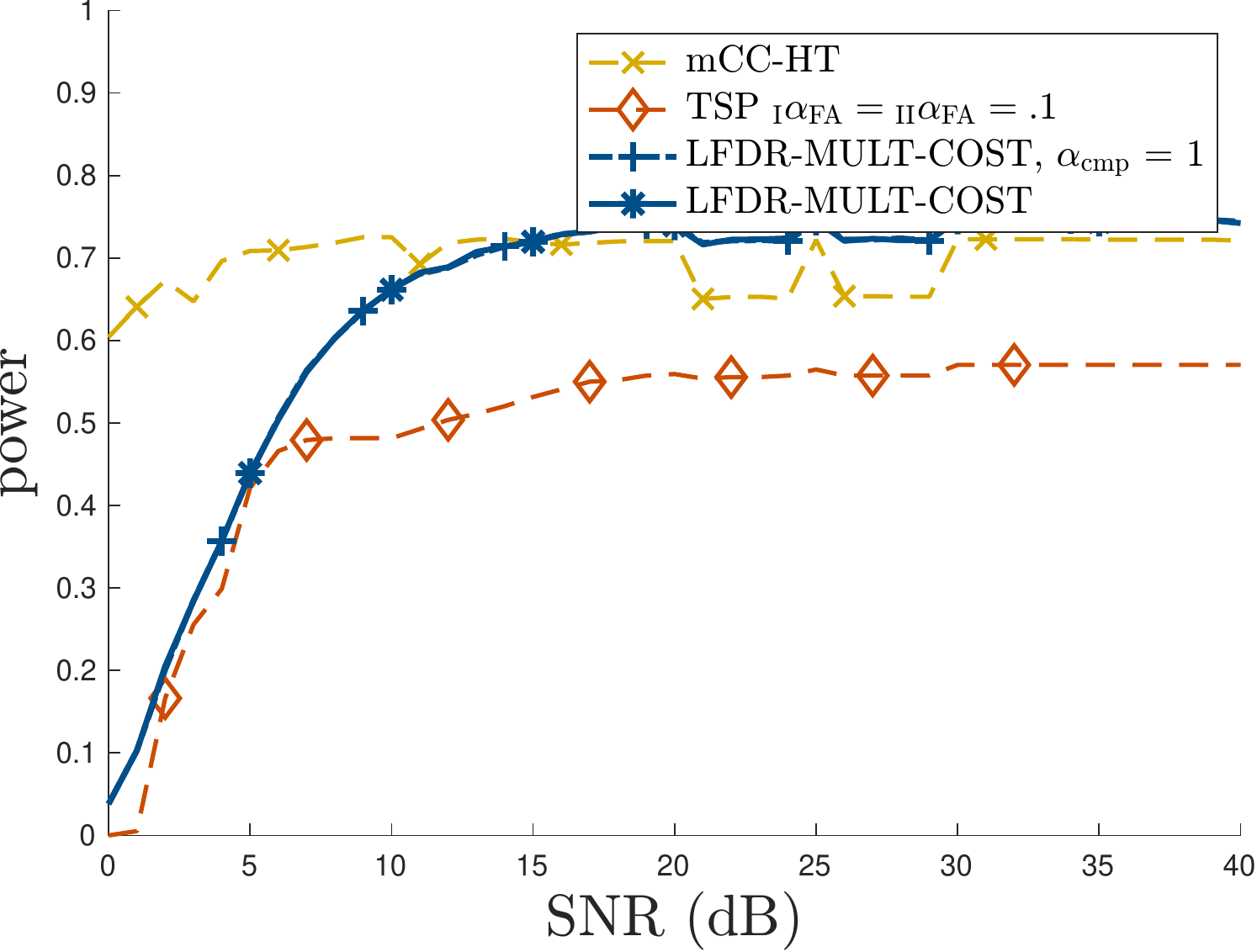}\smallskip
			\caption*{$\mathrm{Power}$}
			\label{fig:exp-snr_atom-pow}
		\end{subfigure}\medskip
		\label{dummy}
		\begin{subfigure}{.49\textwidth}
			\includegraphics[scale=.28]{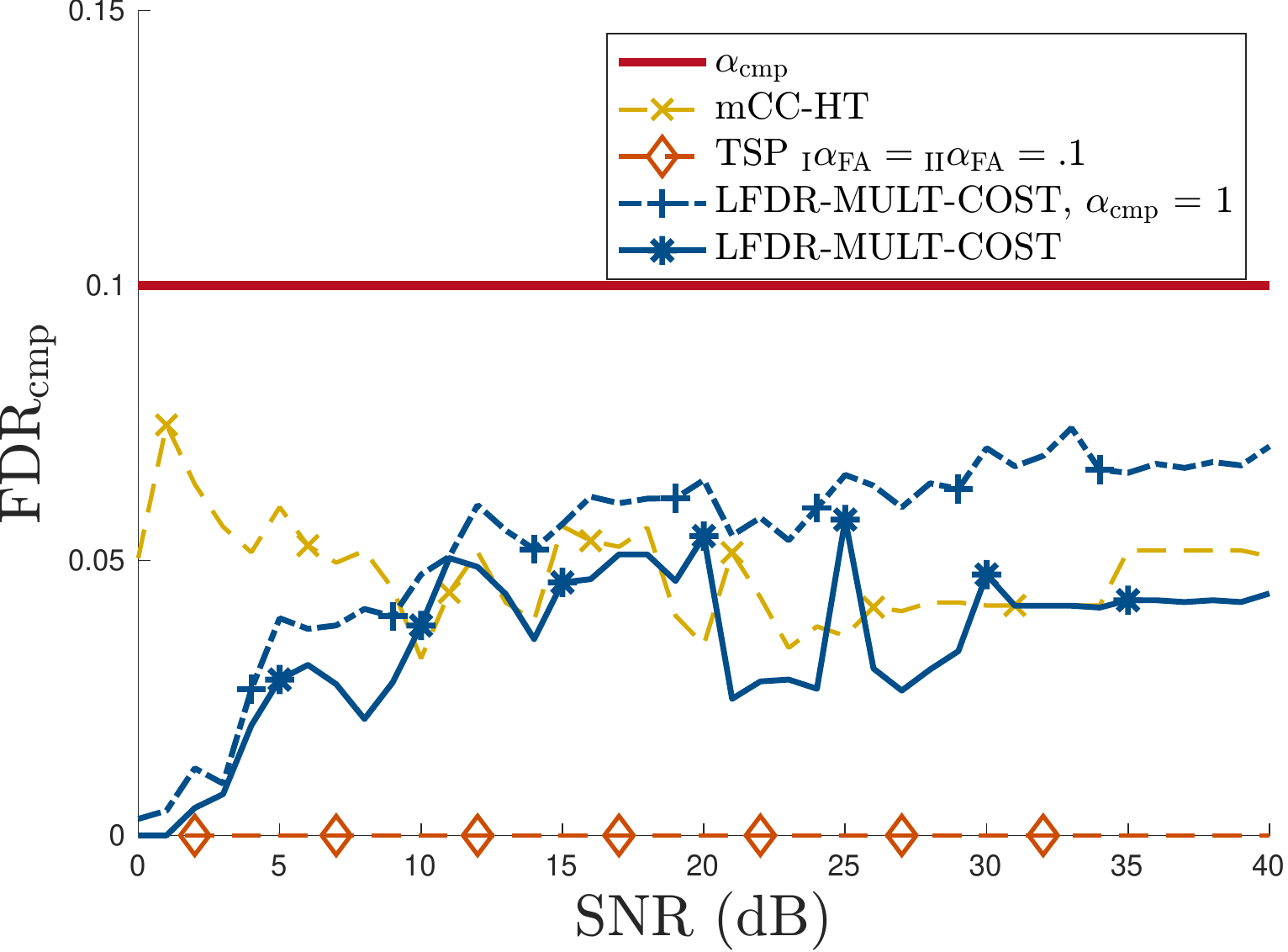}\smallskip
			\caption*{\(\fdrCmp\)}
			\label{fig:exp-snr_cmp-fdr}
		\end{subfigure}
		\begin{subfigure}{.49\textwidth}
			\includegraphics[scale=.28]{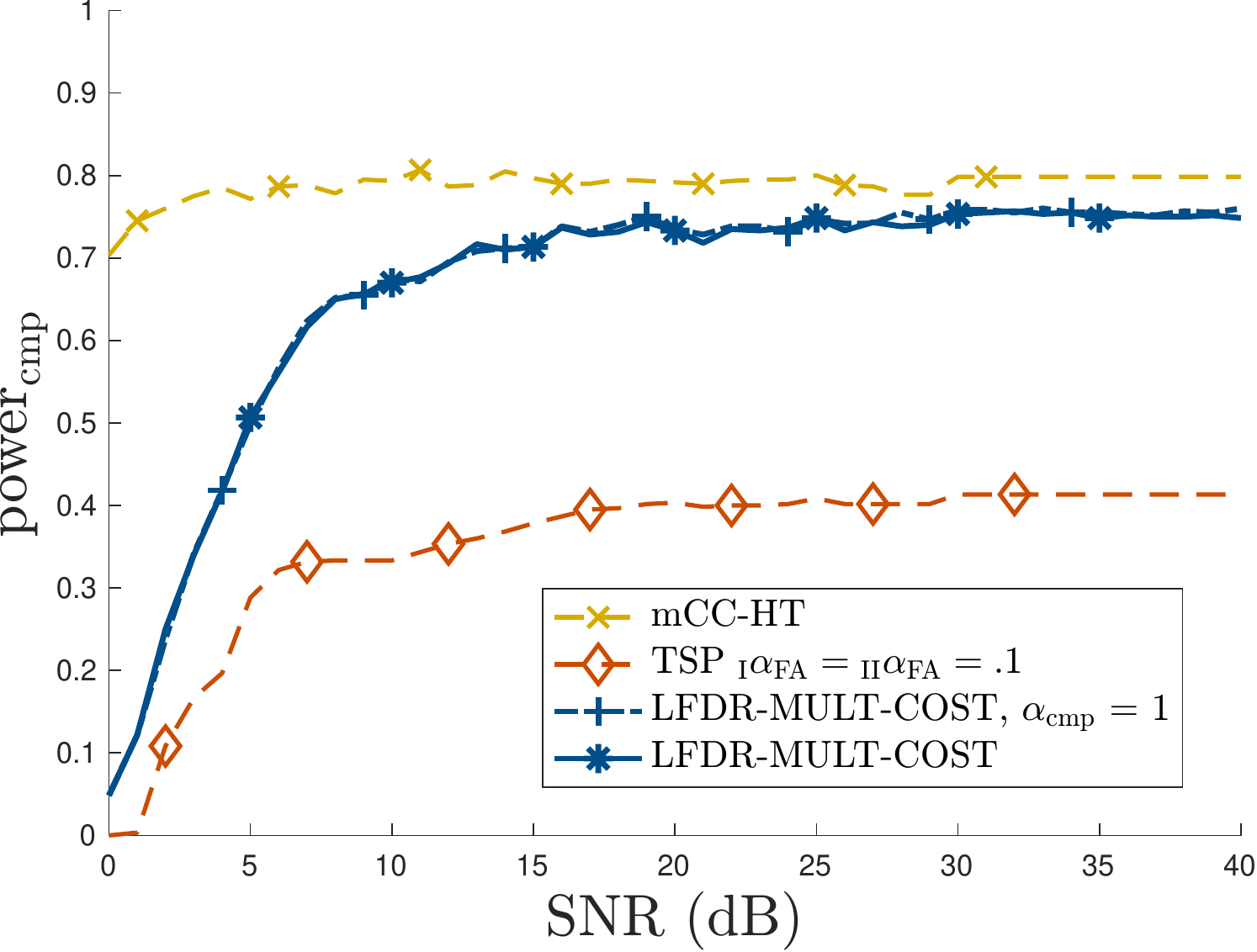}\smallskip
			\caption*{\(\mathrm{Power}_\text{cmp}\)}
			\label{fig:exp-snr_cmp-pow}
		\end{subfigure}
		\medskip
		\caption{performance measures}
		\label{fig:exp-snr_perf}
	\end{subfigure}
	\begin{subfigure}{.44\textwidth}\centering
		\begin{subfigure}{.47\textwidth}
			\centering
			\includegraphics[scale=.2]{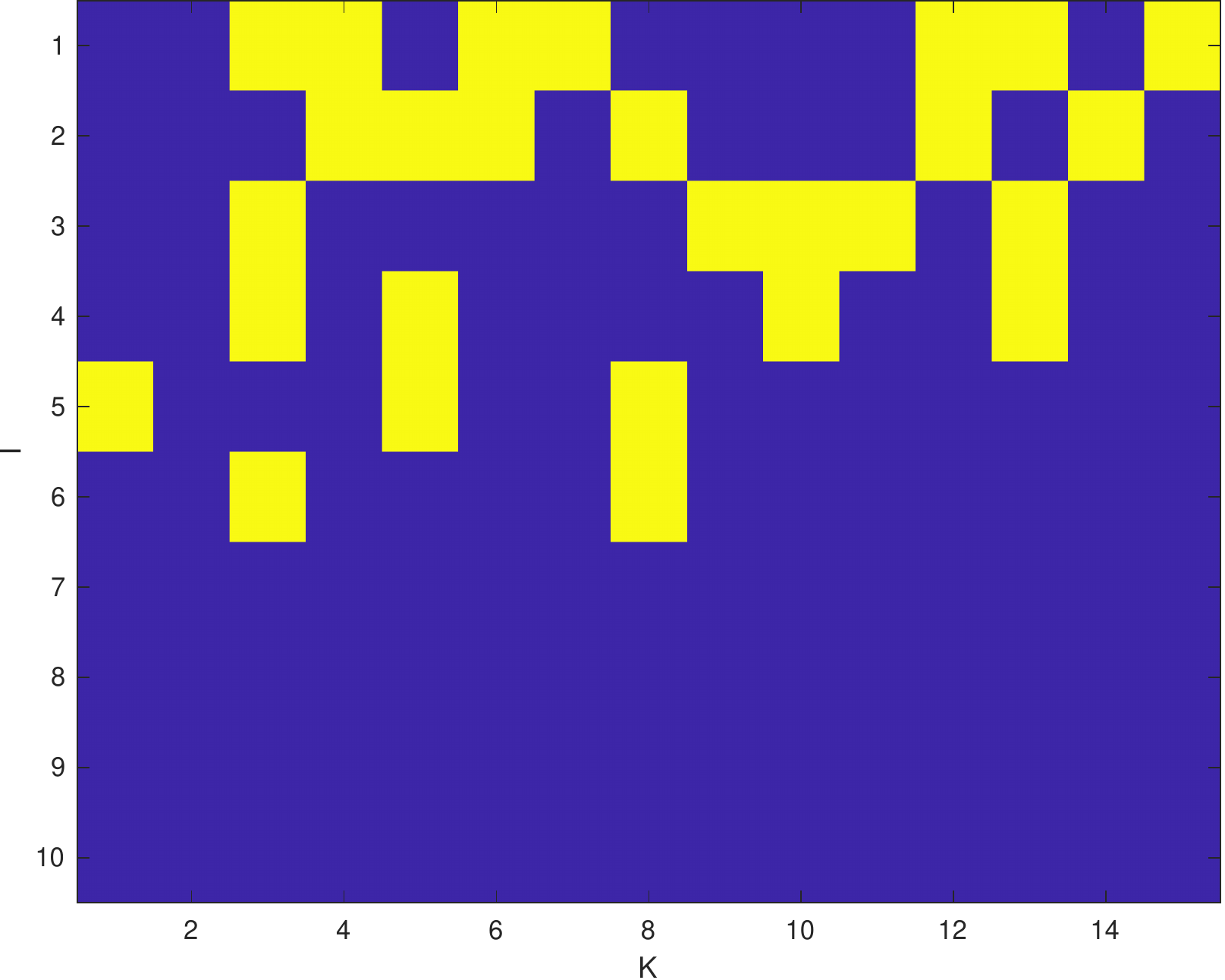}\smallskip
			\caption*{ground truth}
			\label{fig:exp-snr_ex-tru}
		\end{subfigure}
		\begin{subfigure}{.48\textwidth}
			\centering
			\includegraphics[scale=.2]{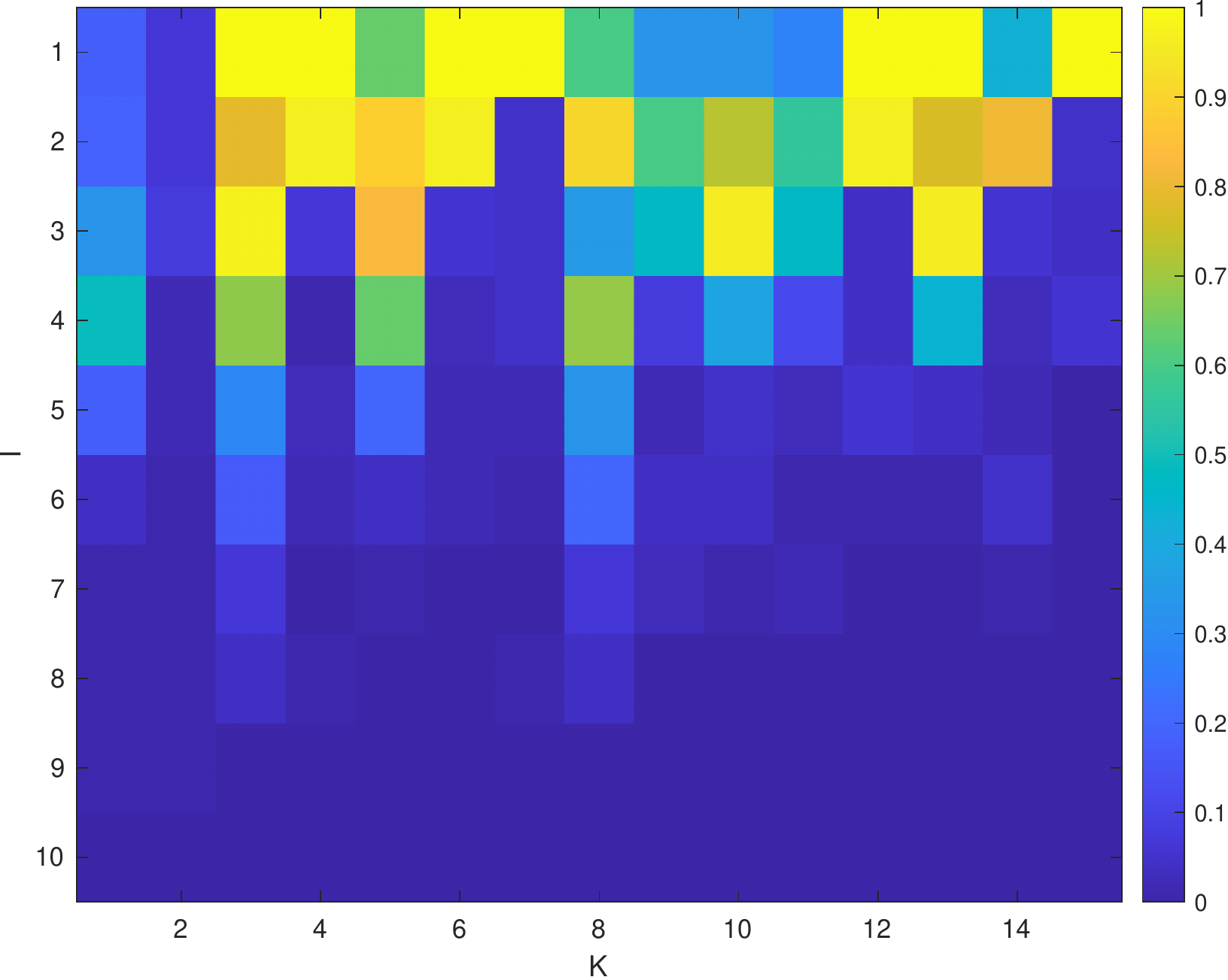}\smallskip
			\caption*{mCCA-HT \cite{Marrinan2018}}
			\label{fig:exp-snr_ex-mar}
		\end{subfigure}\bigskip
		\label{dummy2}
		\begin{subfigure}{.48\textwidth}
			\centering
			\includegraphics[scale=.2]{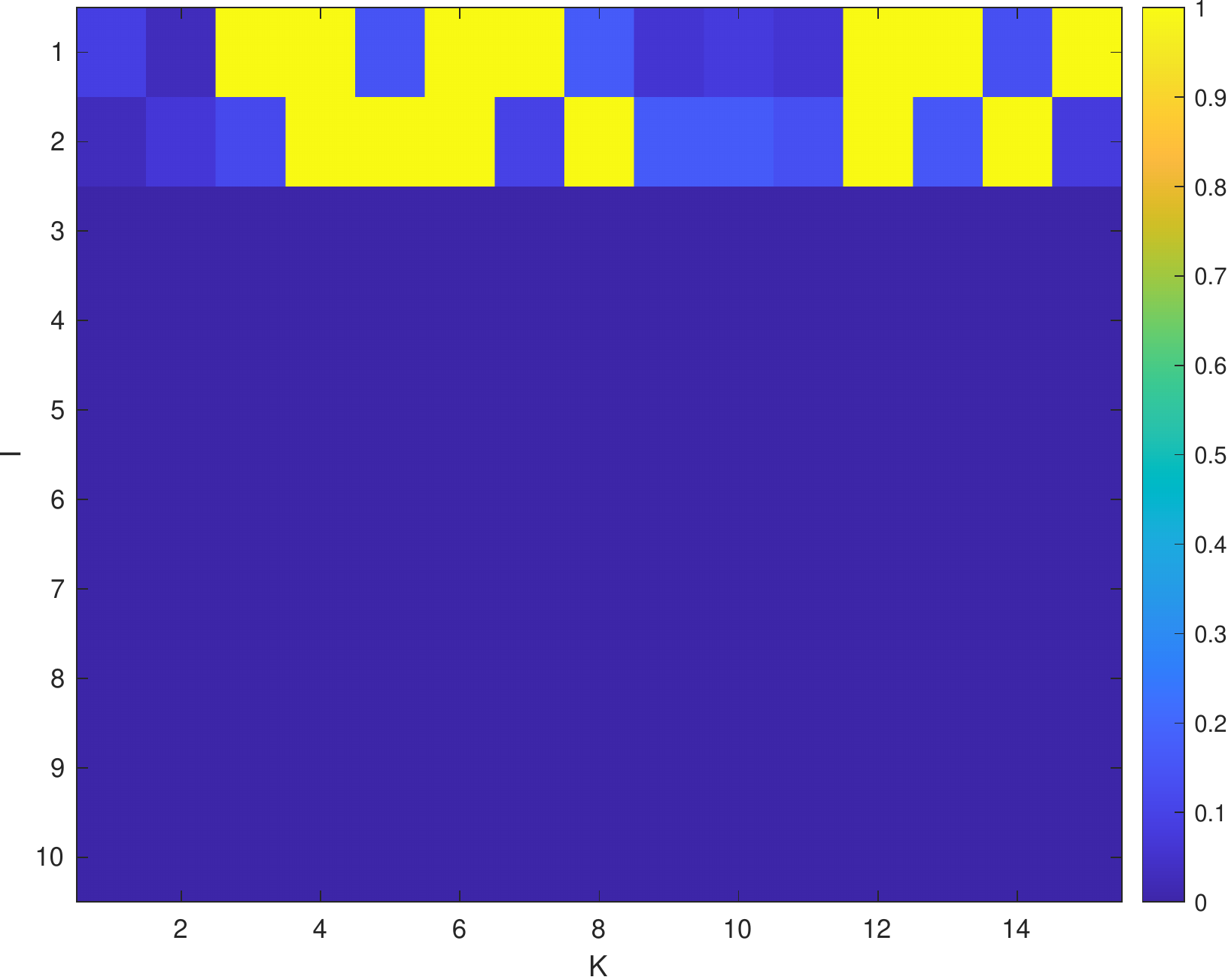}\smallskip
			\caption*{\gls{ts} \cite{Hasija2020}}
			\label{fig:exp-snr_ex-tsp}
		\end{subfigure}
		\begin{subfigure}{.48\textwidth}
			\centering
			\includegraphics[scale=.2]{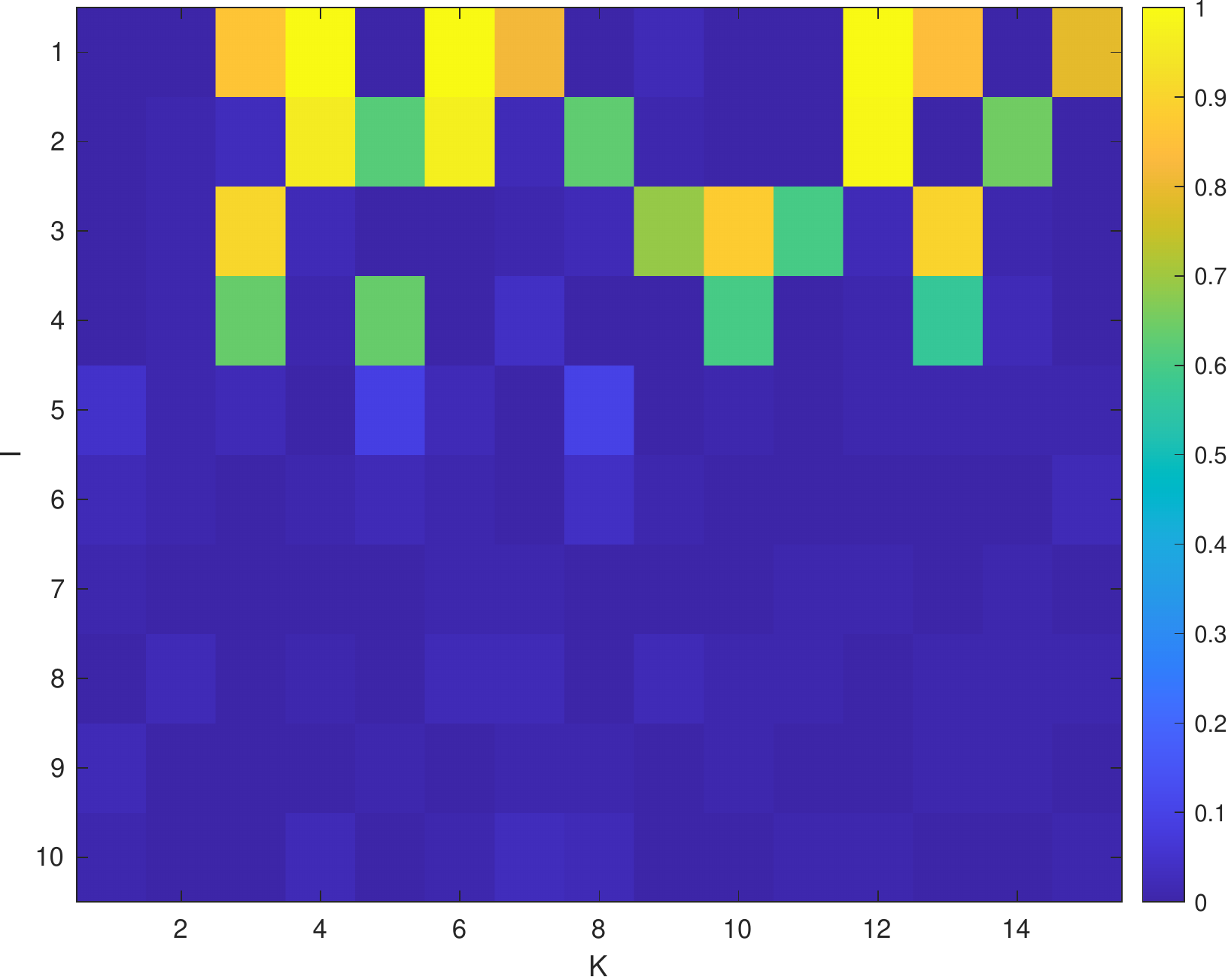}\smallskip
			\caption*{\gls{lfdrmultcost}}
			\label{fig:exp-snr_ex-osp}
		\end{subfigure}
		\bigskip
		\caption{detected correlation structures for $\mathrm{SNR} = 10\unit{dB}$}
		\label{fig:exp-snr_examples}
	\end{subfigure}
	\caption{Experiment~1 - Performance comparison depending on the \gls{snr}. mCCA-HT exhibits an unreasonably high false discovery rate. \gls{ts} produces less erroneously detected correlations, but has low detection power for higher SNR values. Our proposed \gls{lfdrmultcost} fulfills its nominal \gls{fdr} level while $\fdrThrAtom = 0.1$ while yielding high detection power. Setting \(\fdrThrCmp = 1\) leads to slightly more false positives on the component level.}
	\label{fig:exp-snr}
\end{figure}

\subsubsection*{Experiment~2} The $\mathrm{SNR} = 5\unit{dB}$ is now fixed, along $\numSet = 20$ the number of data sets and $\numSrc = 10$ the number of components. The components and the noise are both Gaussian distributed. The proportion of true atom null hypothesis is $\pi_0 = .7$. All components are correlated between at least few sets, so there cannot be any component level false positives. Again, the atom \gls{fdr} for mCCA-HT is higher than the axis limit. \gls{ts} yields a high \gls{fdr} when the number of samples is small. As its corresponding detection pattern for $\numSam = 175$ on the bottom left of Fig.~\ref{fig:exp-A_examples} reveals, $\numActSrc$ is underestimated. In those components with $\srcIdx\leq\numActSrcEst$, falsely identified correlations occur frequently. Since the small $\numSet$ introduces a large estimation error in the coherence matrix and thus, a lot of variation in the test statistics, our proposed method is conservative while the number of samples is small. Nevertheless, it provides highly reliable discoveries, since the \gls{fdr} is controlled even for extremely small sample sizes.

\begin{figure}
	\begin{subfigure}{.54\textwidth}
		\begin{subfigure}{.49\textwidth}
			\includegraphics[scale=.28]{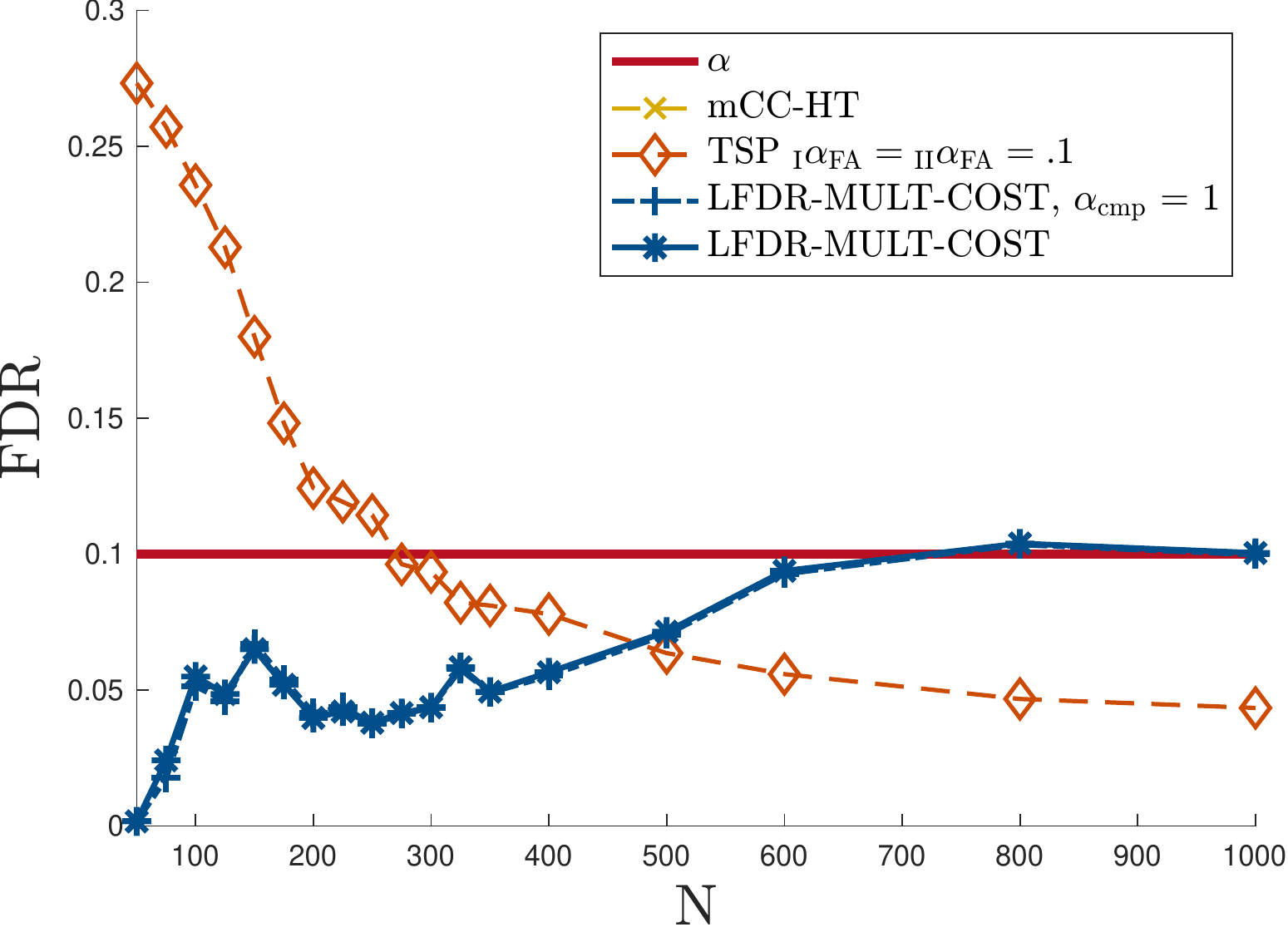}\smallskip
			\caption*{\(\fdr\)}
			\label{fig:exp-A_atom-fdr}
		\end{subfigure}
		\begin{subfigure}{.49\textwidth}
			\includegraphics[scale=.28]{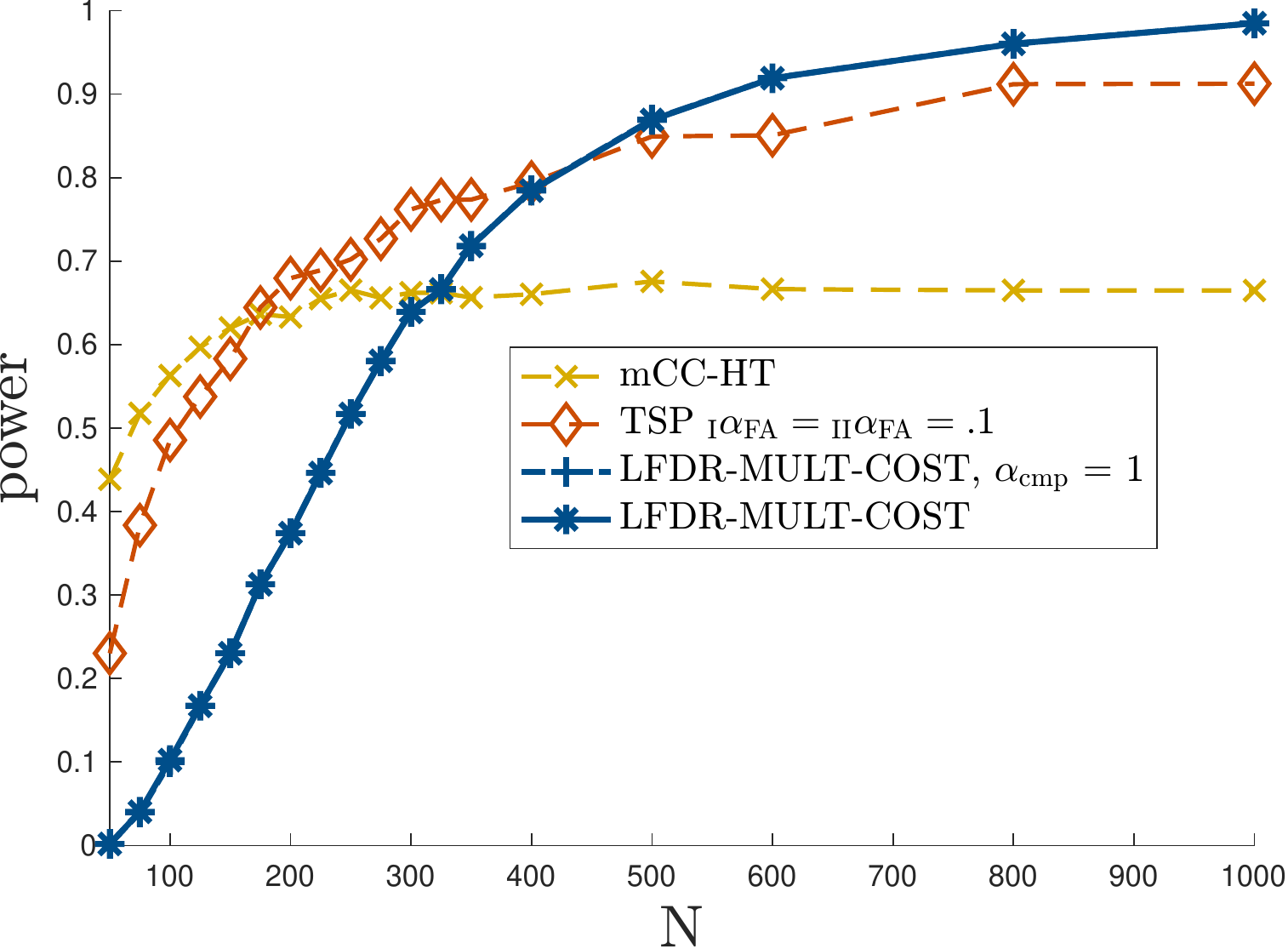}\smallskip
			\caption*{$\mathrm{Power}$}
			\label{fig:exp-A_atom-pow}
		\end{subfigure}\medskip
		\label{dummy}
		\begin{subfigure}{.49\textwidth}
			\includegraphics[scale=.28]{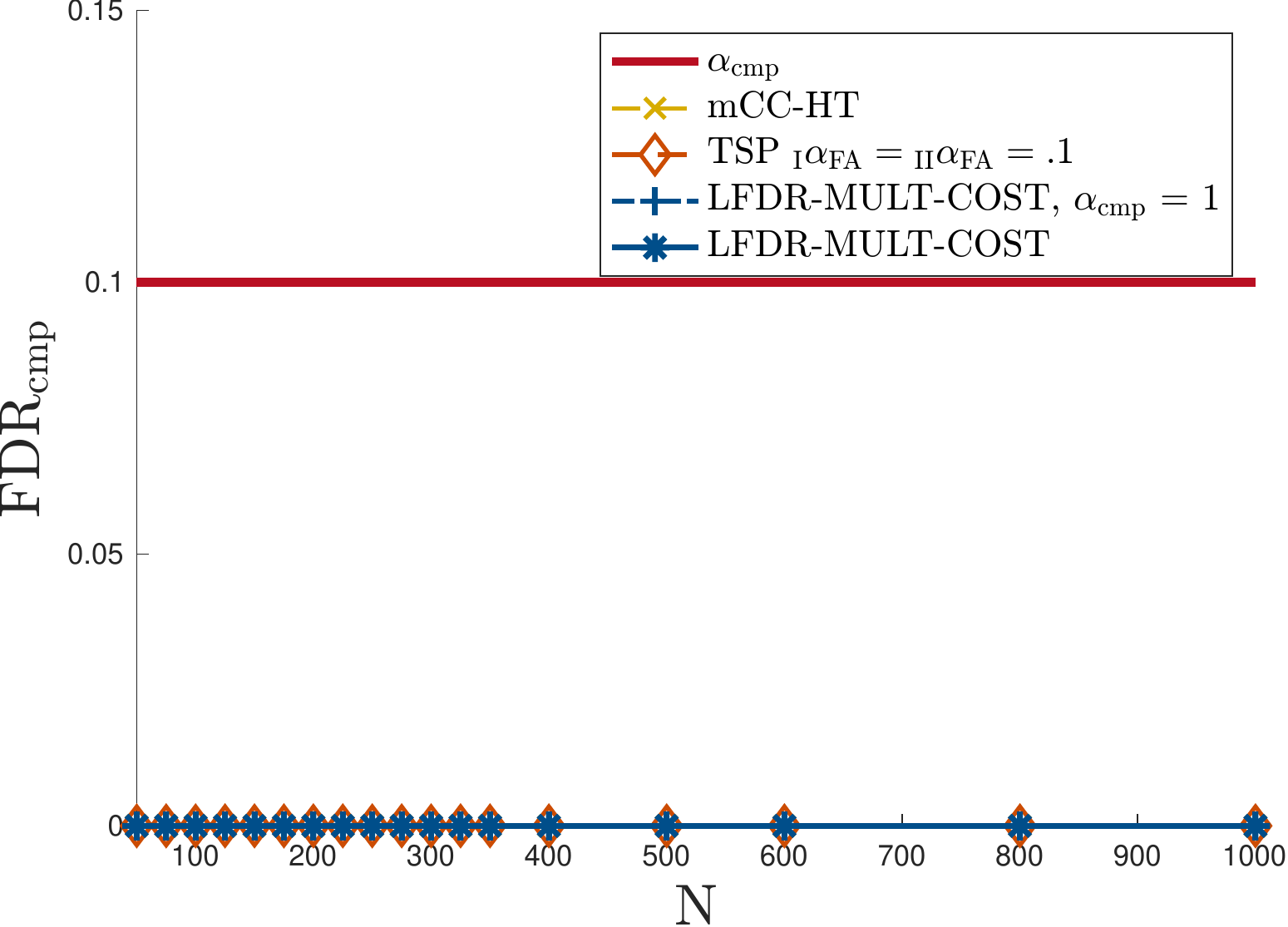}\smallskip
			\caption*{\(\fdrCmp\)}
			\label{fig:exp-A_cmp-fdr}
		\end{subfigure}
		\begin{subfigure}{.49\textwidth}
			\includegraphics[scale=.28]{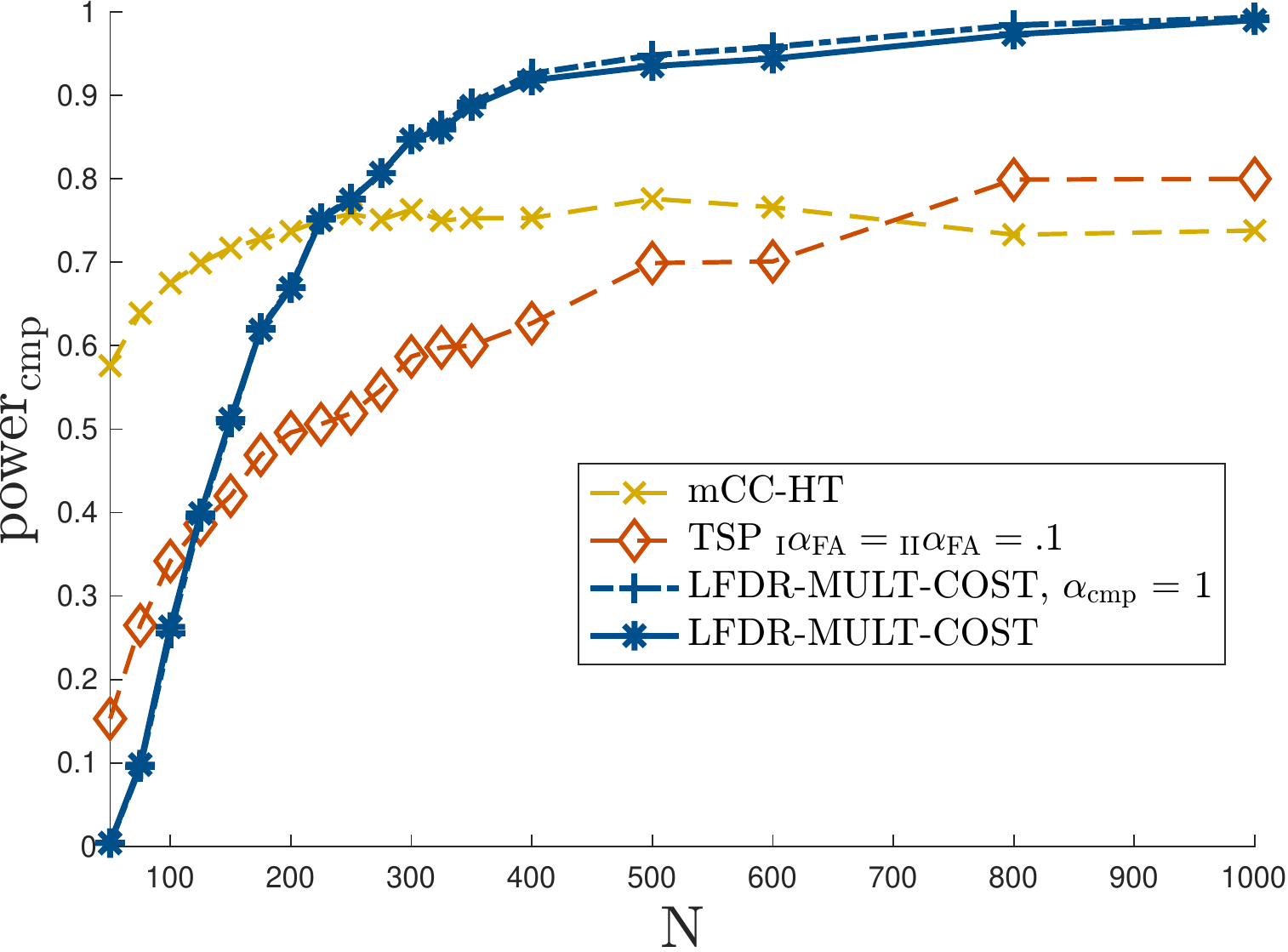}\smallskip
			\caption*{\(\mathrm{Power}_\text{cmp}\)}
			\label{fig:exp-A_cmp-pow}
		\end{subfigure}
		\medskip
		\caption{performance measures}
		\label{fig:exp-A_perf}
	\end{subfigure}
	\begin{subfigure}{.44\textwidth}\centering
		\begin{subfigure}{.47\textwidth}
			\centering
			\includegraphics[scale=.2]{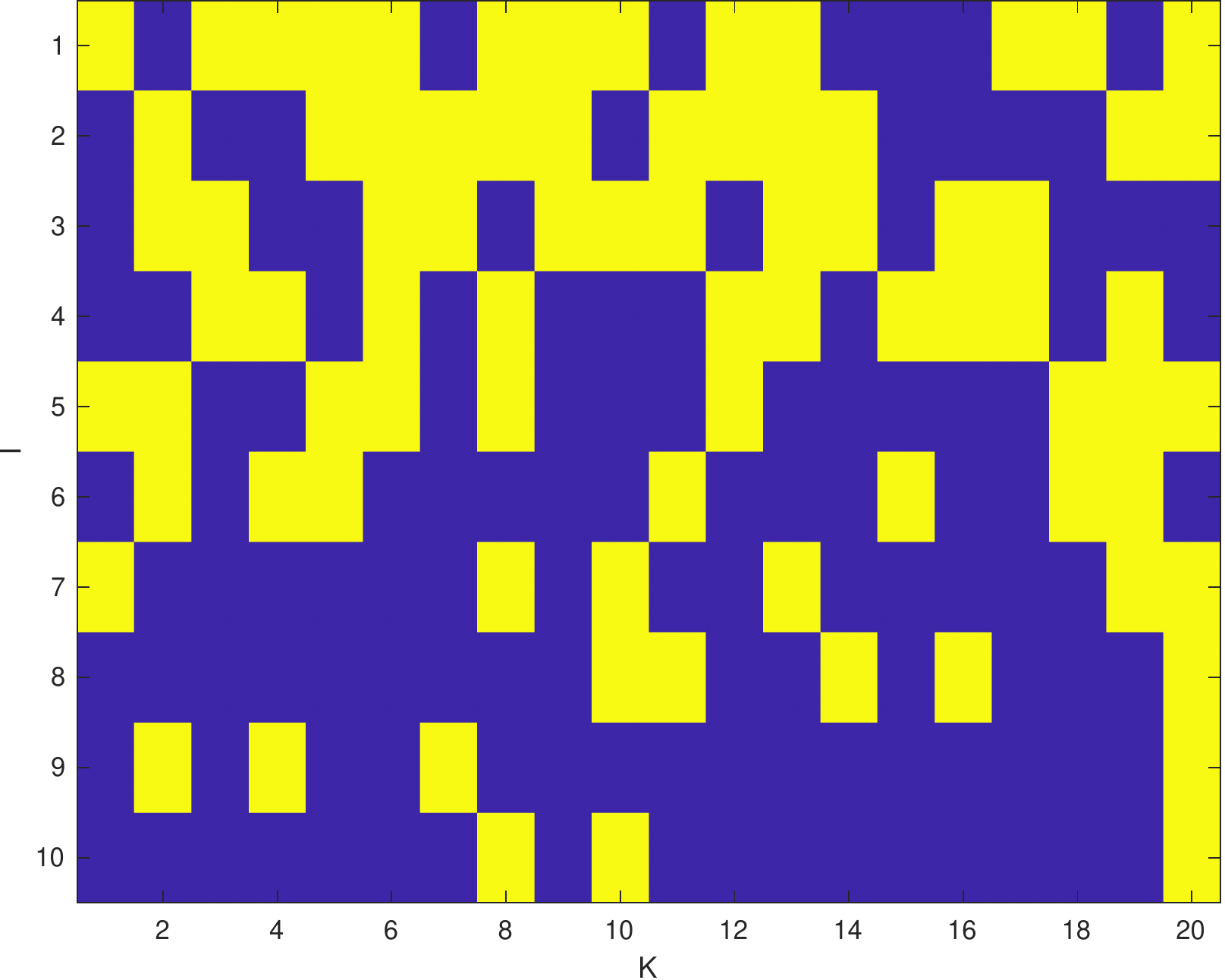}\smallskip
			\caption*{ground truth}
			\label{fig:exp-A_ex-tru}
		\end{subfigure}
		\begin{subfigure}{.48\textwidth}
			\centering
			\includegraphics[scale=.2]{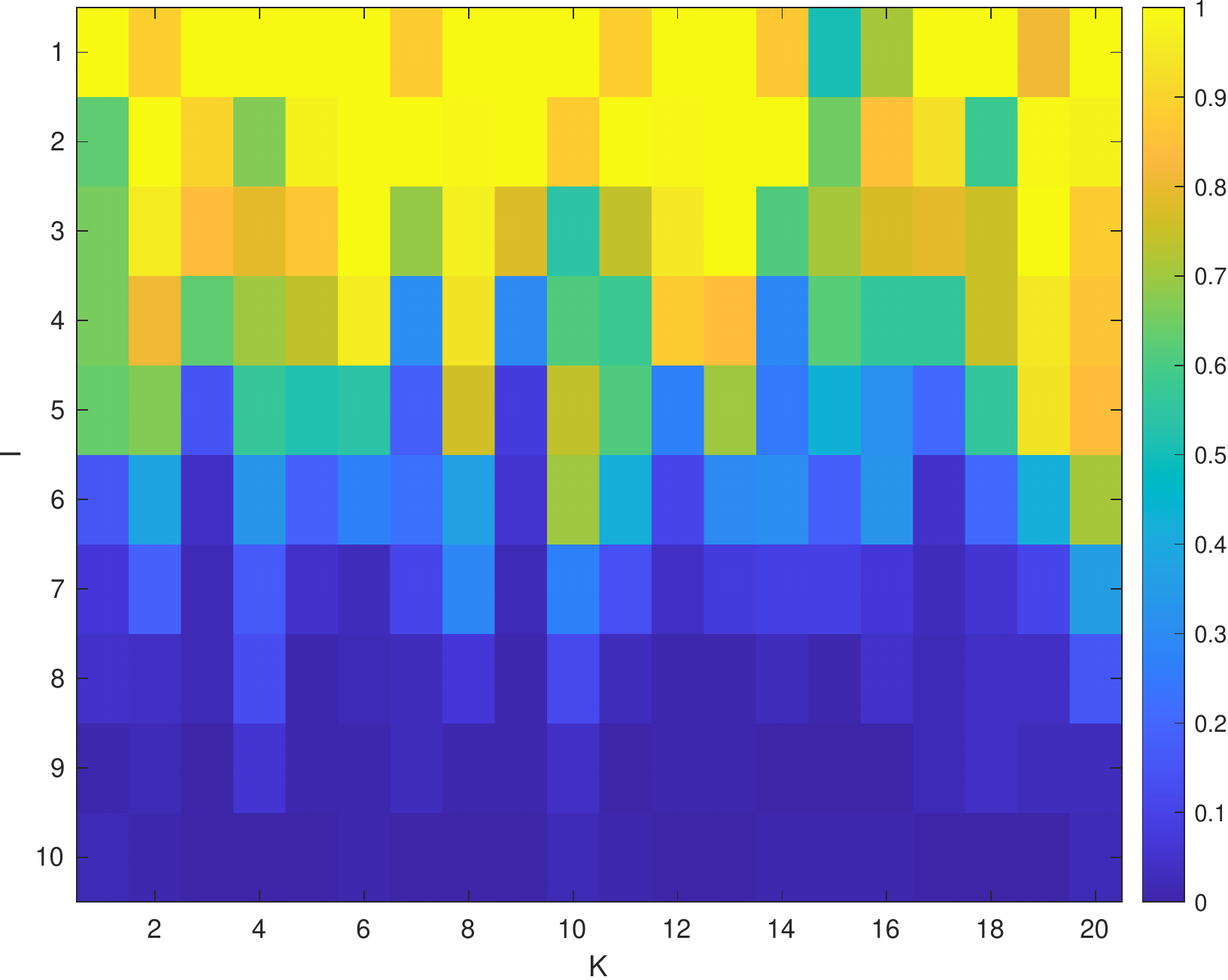}\smallskip
			\caption*{mCCA-HT \cite{Marrinan2018}}
			\label{fig:exp-A_ex-mar}
		\end{subfigure}\bigskip
		\label{dummy2}
		\begin{subfigure}{.48\textwidth}
			\centering
			\includegraphics[scale=.2]{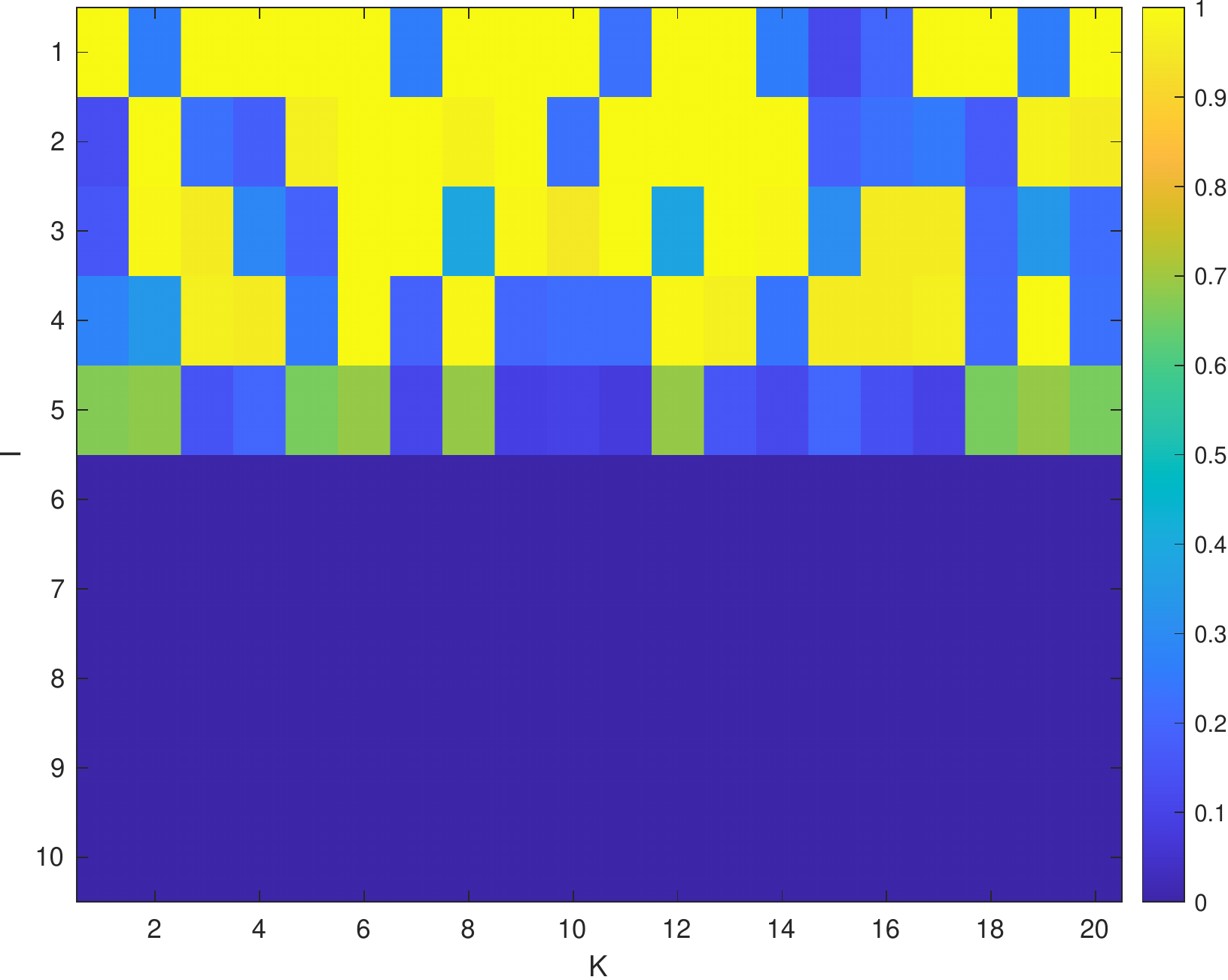}\smallskip
			\caption*{\gls{ts} \cite{Hasija2020}}
			\label{fig:exp-A_ex-tsp}
		\end{subfigure}
		\begin{subfigure}{.48\textwidth}
			\centering
			\includegraphics[scale=.2]{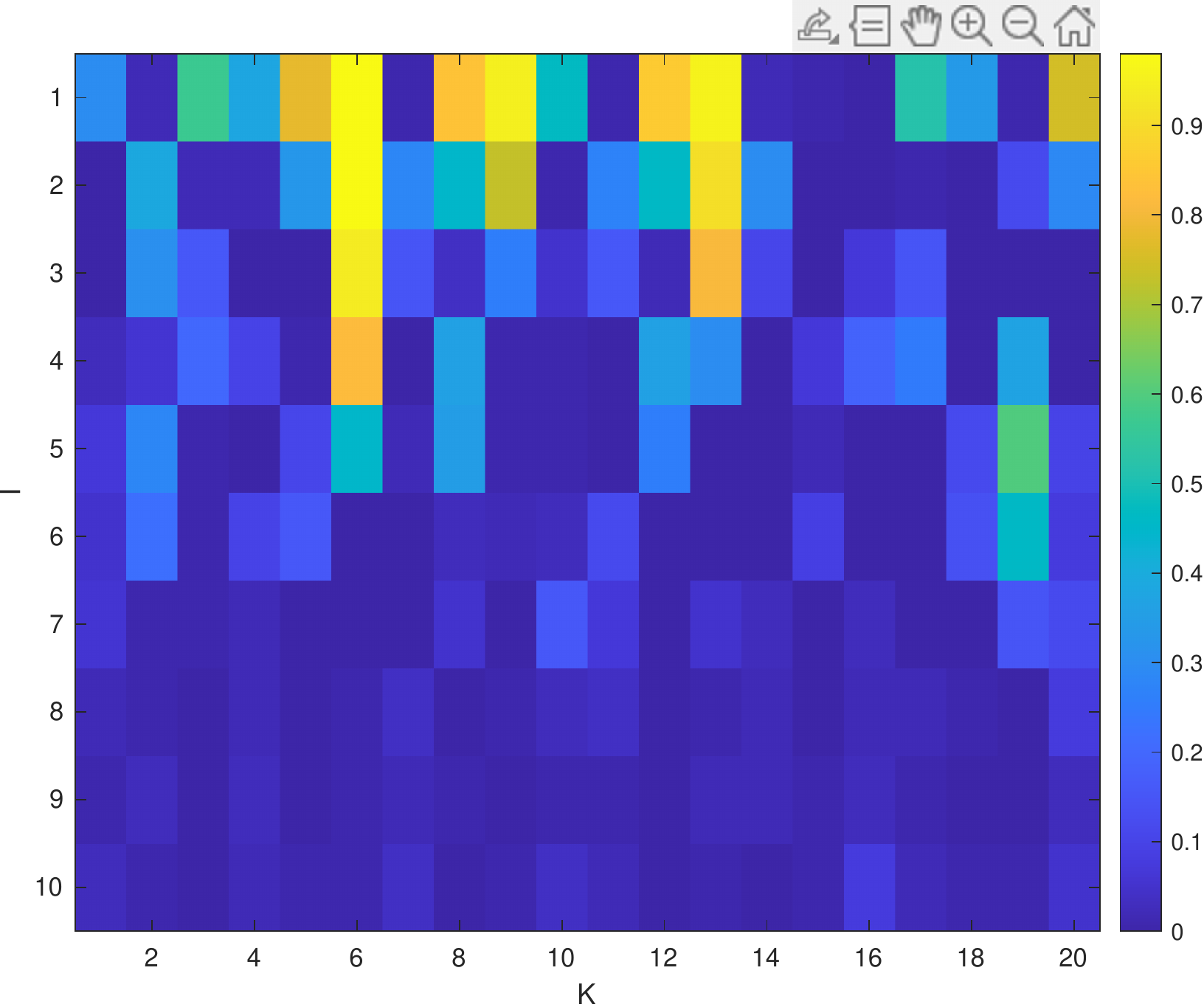}\smallskip
			\caption*{\gls{lfdrmultcost}}
			\label{fig:exp-A_ex-osp}
		\end{subfigure}
		\bigskip
		\caption{detected correlation structures for $\numSam = 175$}
		\label{fig:exp-A_examples}
	\end{subfigure}
	\caption{Experiment 2 - Performance comparison depending on the sample size. mCCA-HT yields unreasonably many false positives. \gls{ts} yields many falsely detected correlations when the sample size is very small. Our proposed \gls{lfdrmultcost} produces a much lower detection power, but performs consistently for all sample sizes. \gls{fdr} control is maintained even at very small sample sizes. All components are correlated between at least few sets. Hence, the value of  \(\fdrThrCmp\) has no impact.}
	\label{fig:exp-A}
\end{figure}

\subsubsection*{Experiment~3} Again, $\mathrm{SNR} = 5\unit{dB}$ and $\numSrc = 10$ components. Now, $\numSet = 500$ is fixed and the number of data sets $\numSet$ varies. The components are Laplacian distributed, but the additive noise is Gaussian. The results with $\pi_0 = 0.8$ are shown in Fig.~\ref{fig:exp-b08}, additional result for  $\pi_0 = 0.9$ are provided in the Appendix \ref{app:res}.

mCCA-HT again results in a too high proportion of false positives. The discoveries of \gls{ts} contain more and more false positives, as the number of data sets $\numSet$ grows. Inspection of the bottom left of Fig.~\ref{fig:exp-b08_examples} reveals that \gls{ts} underestimates the $\numActSrc$, preventing discoveries anywhere in components below, but commits more false positives within the first $\numActSrcEst$ components than our proposed approach \gls{lfdrmultcost}. This is the effect discussed in Section~\ref{sec:single-tests-multiple-tests} of conducting many binary tests without correcting for false multiple testing, which leads to an ever increasing number of false positives when the number of tested null hypotheses increases. Our proposed method performs well, exhibiting a nearly constant atom level \gls{fdr} over $\numSet$ and providing high detection power. The power only starts to slowly decrease when $\numSet$ becomes fairly large. This is due to the relative sample size: The dimensions of the coherence matrix $\numSrc\cdot\numSet \times \numSrc\cdot\numSet$ increase rapidly in $\numSet$, while the number of available samples $\numSam = 500$ is fixed. The effects of small sample size were discussed in Experiment 2. \gls{lfdrmultcost} with $\fdrThrCmp = 1$ has a slightly increased $\fdrCmp>0.1$ for $\numSet = 9$, but yields detection power nearly identical to \gls{lfdrmultcost} with $\fdrThrCmp = 0.1$. Hence, the additional component level false positive control comes at little cost.

\begin{figure}
	\begin{subfigure}{.54\textwidth}
		\begin{subfigure}{.49\textwidth}
			\includegraphics[scale=.28]{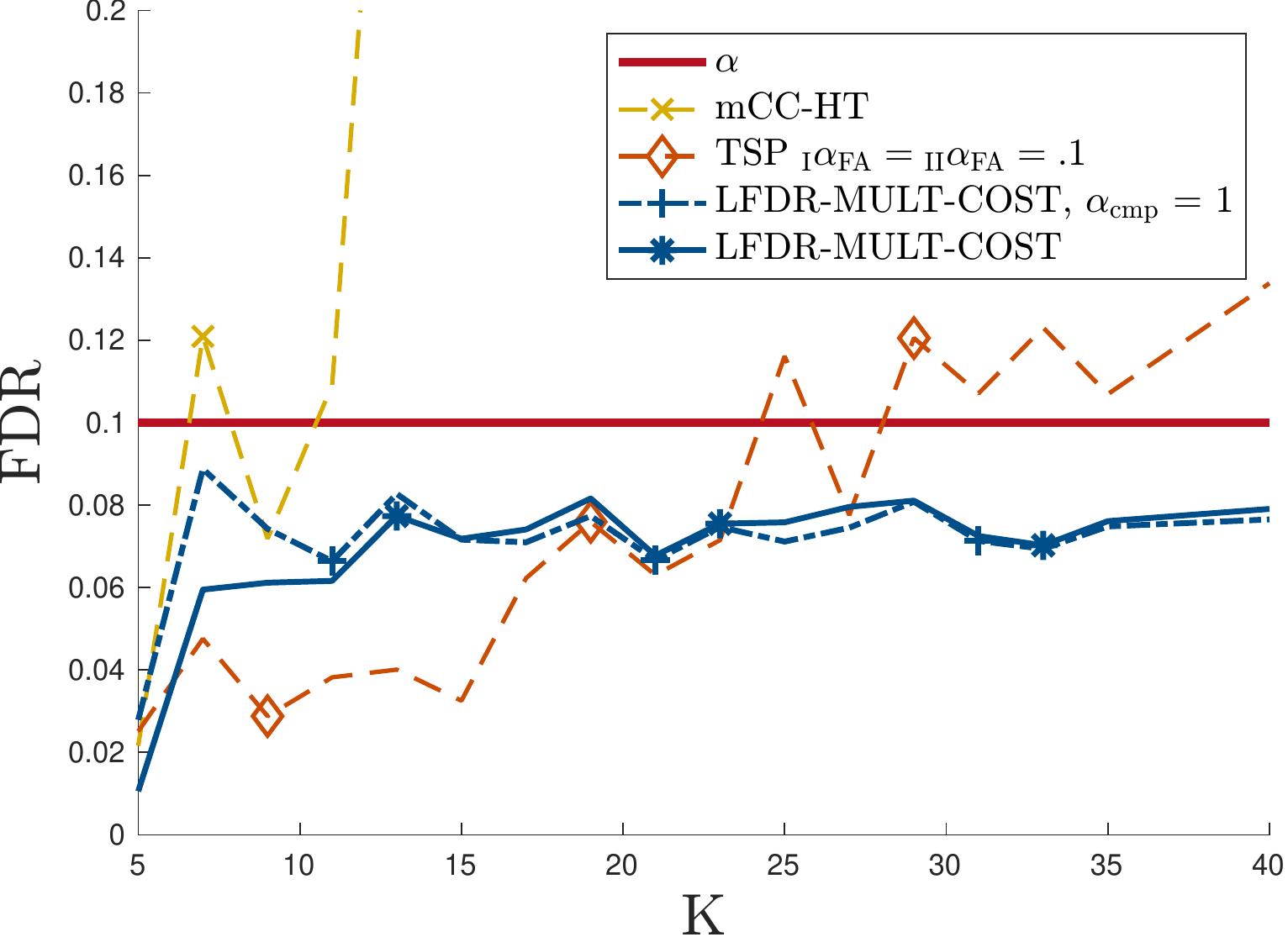}\smallskip
			\caption*{\(\fdr\)}
			\label{fig:exp-b08_atom-fdr}
		\end{subfigure}
		\begin{subfigure}{.49\textwidth}
			\includegraphics[scale=.28]{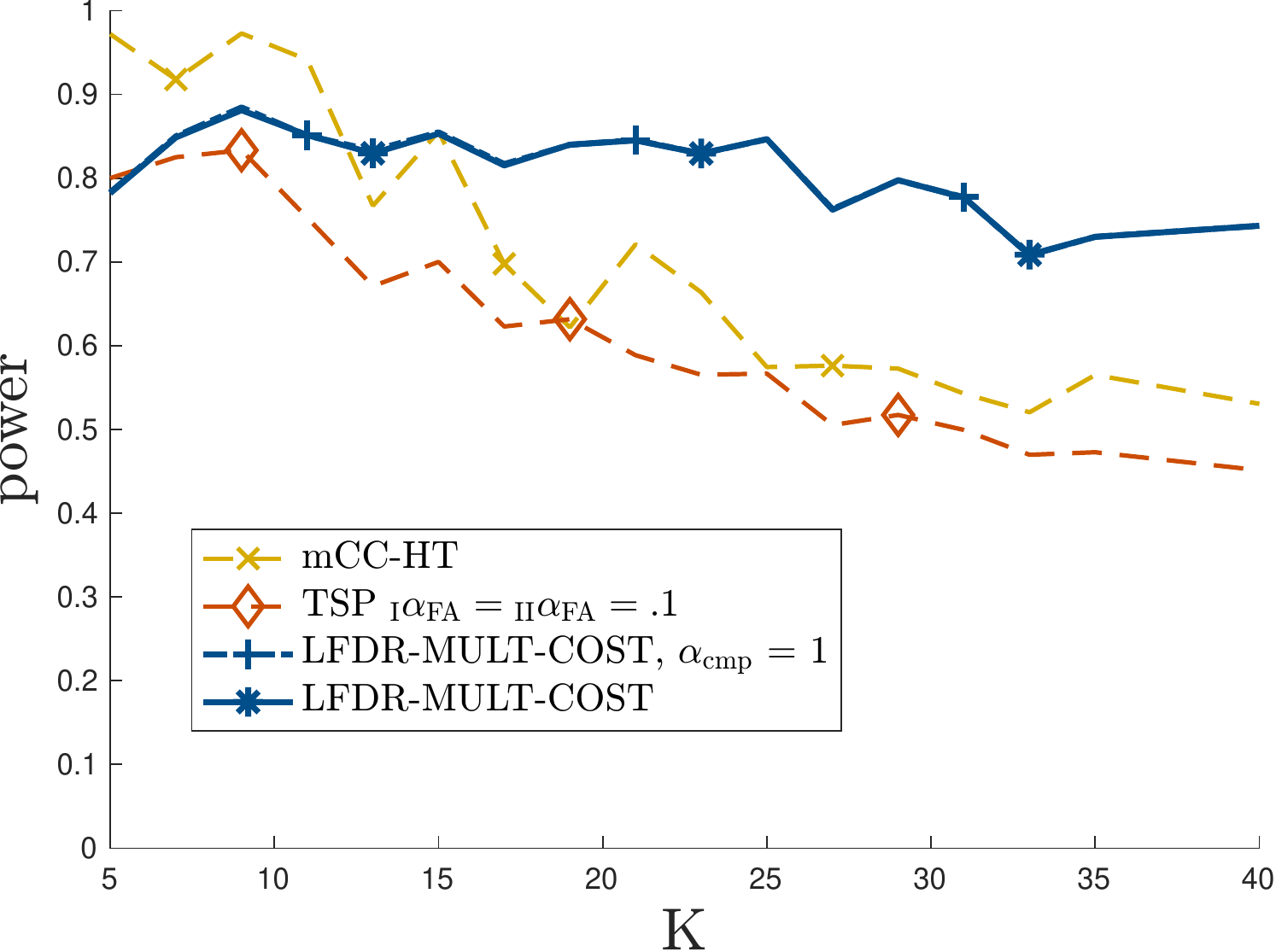}\smallskip
			\caption*{$\mathrm{Power}$}
			\label{fig:exp-b08_atom-pow}
		\end{subfigure}\medskip
		\label{dummy}
		\begin{subfigure}{.49\textwidth}
			\includegraphics[scale=.28]{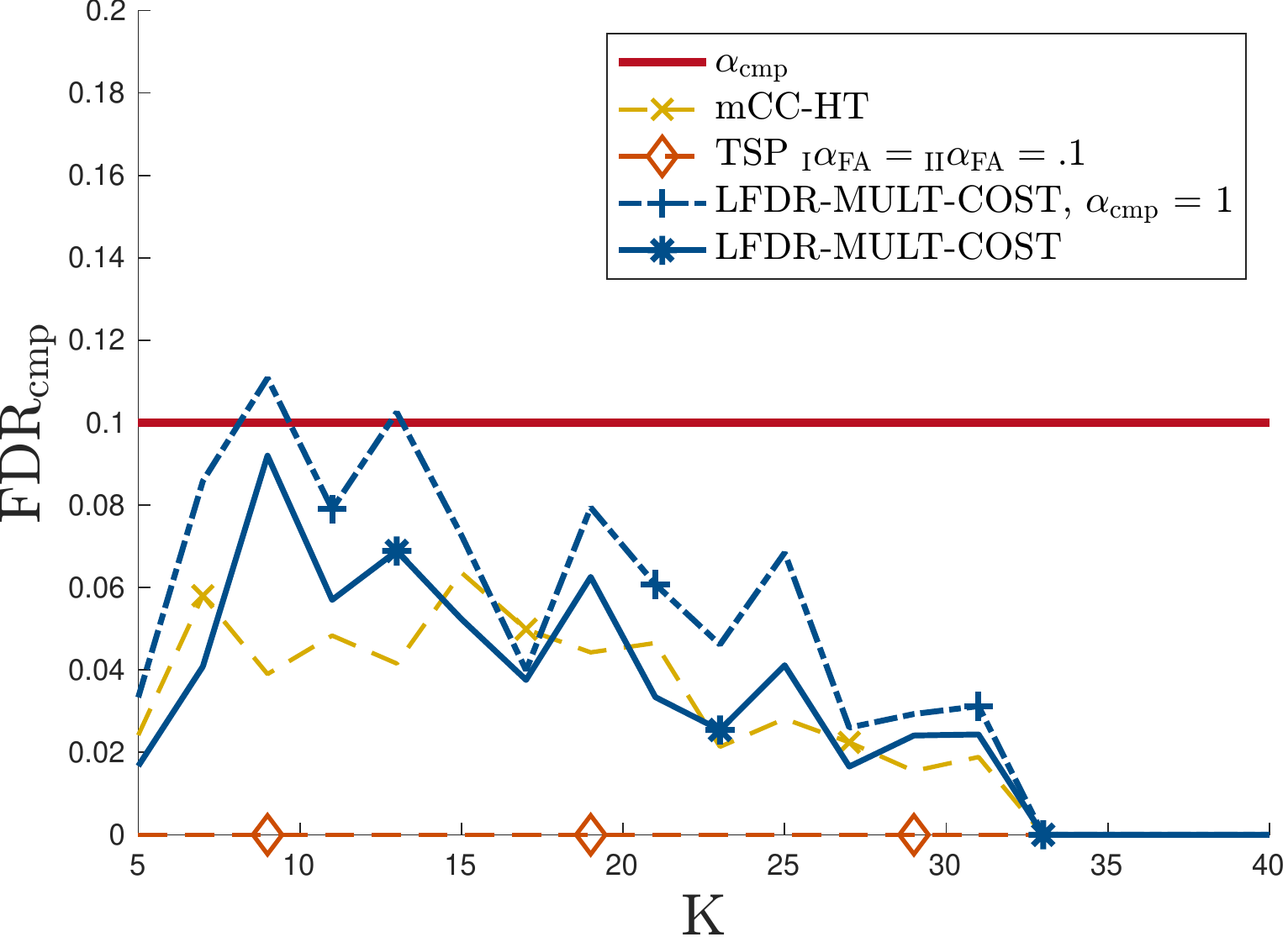}\smallskip
			\caption*{\(\fdrCmp\)}
			\label{fig:exp-b08_cmp-fdr}
		\end{subfigure}
		\begin{subfigure}{.49\textwidth}
			\includegraphics[scale=.28]{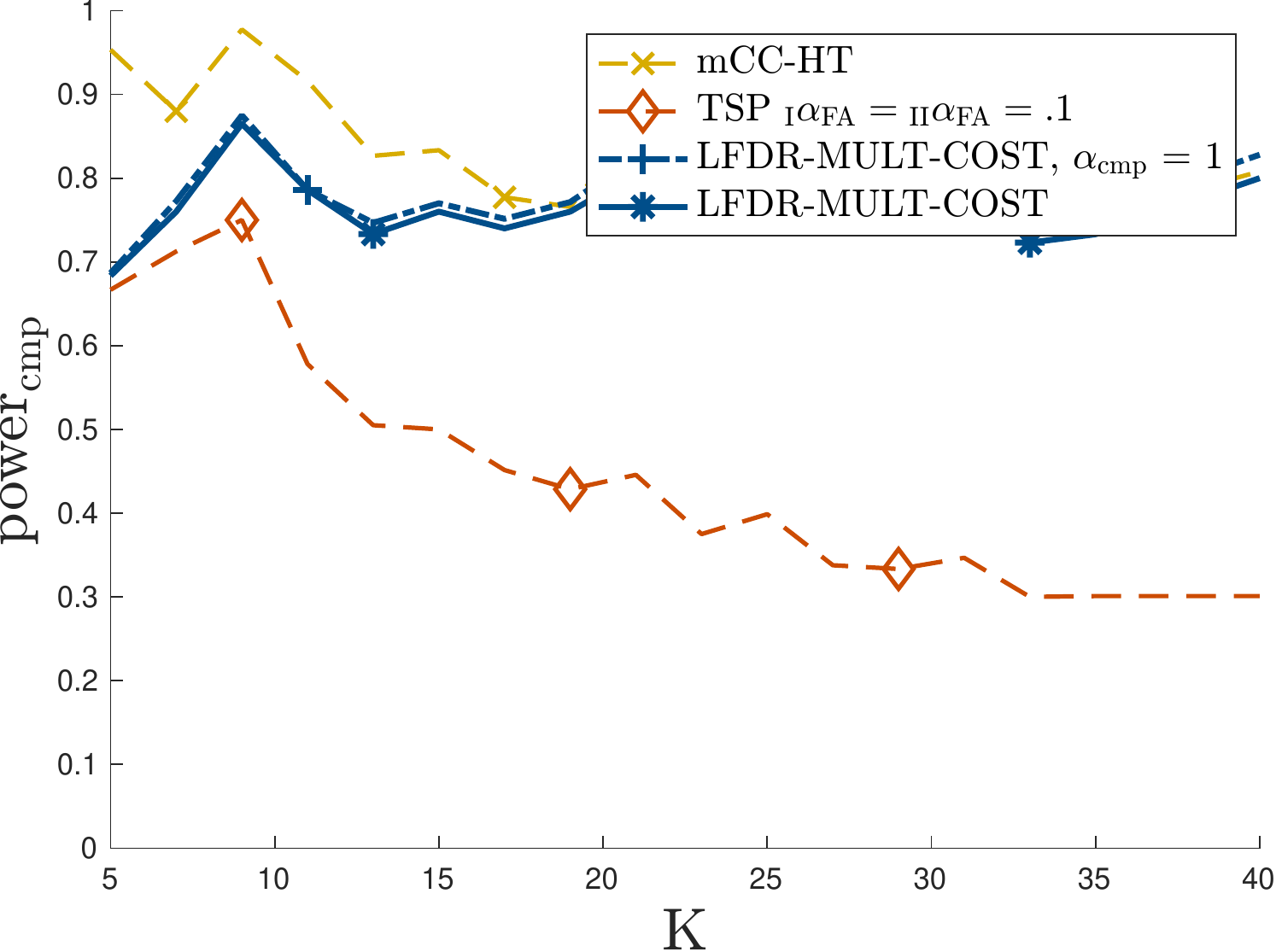}\smallskip
			\caption*{\(\mathrm{Power}_\text{cmp}\)}
			\label{fig:exp-b08_cmp-pow}
		\end{subfigure}
		\medskip
		\caption{performance measures}
		\label{fig:exp-b08_perf}
	\end{subfigure}
	\begin{subfigure}{.44\textwidth}\centering
		\begin{subfigure}{.47\textwidth}
			\centering
			\includegraphics[scale=.2]{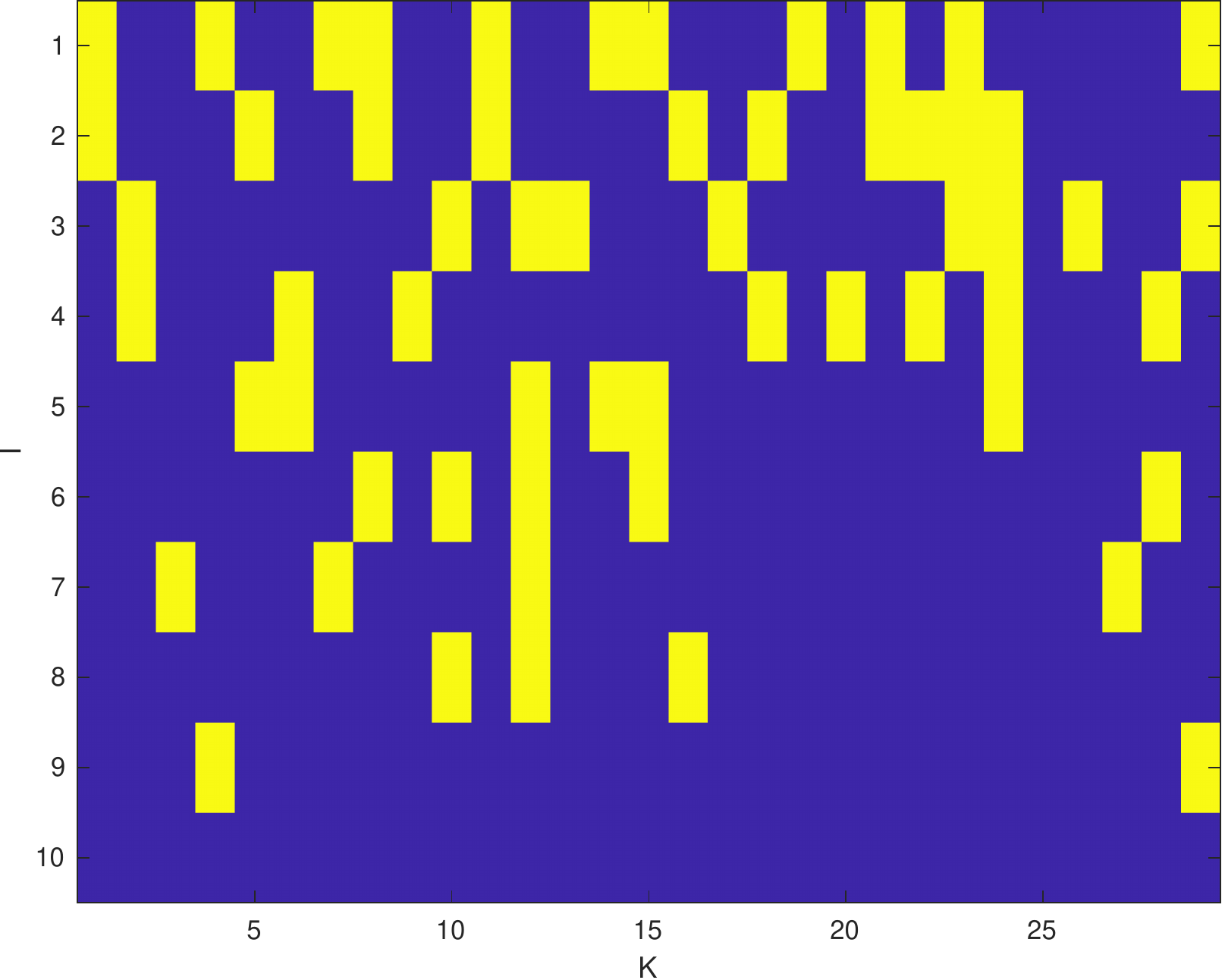}\smallskip
			\caption*{ground truth}
			\label{fig:exp-b08_ex-tru}
		\end{subfigure}
		\begin{subfigure}{.48\textwidth}
			\centering
			\includegraphics[scale=.2]{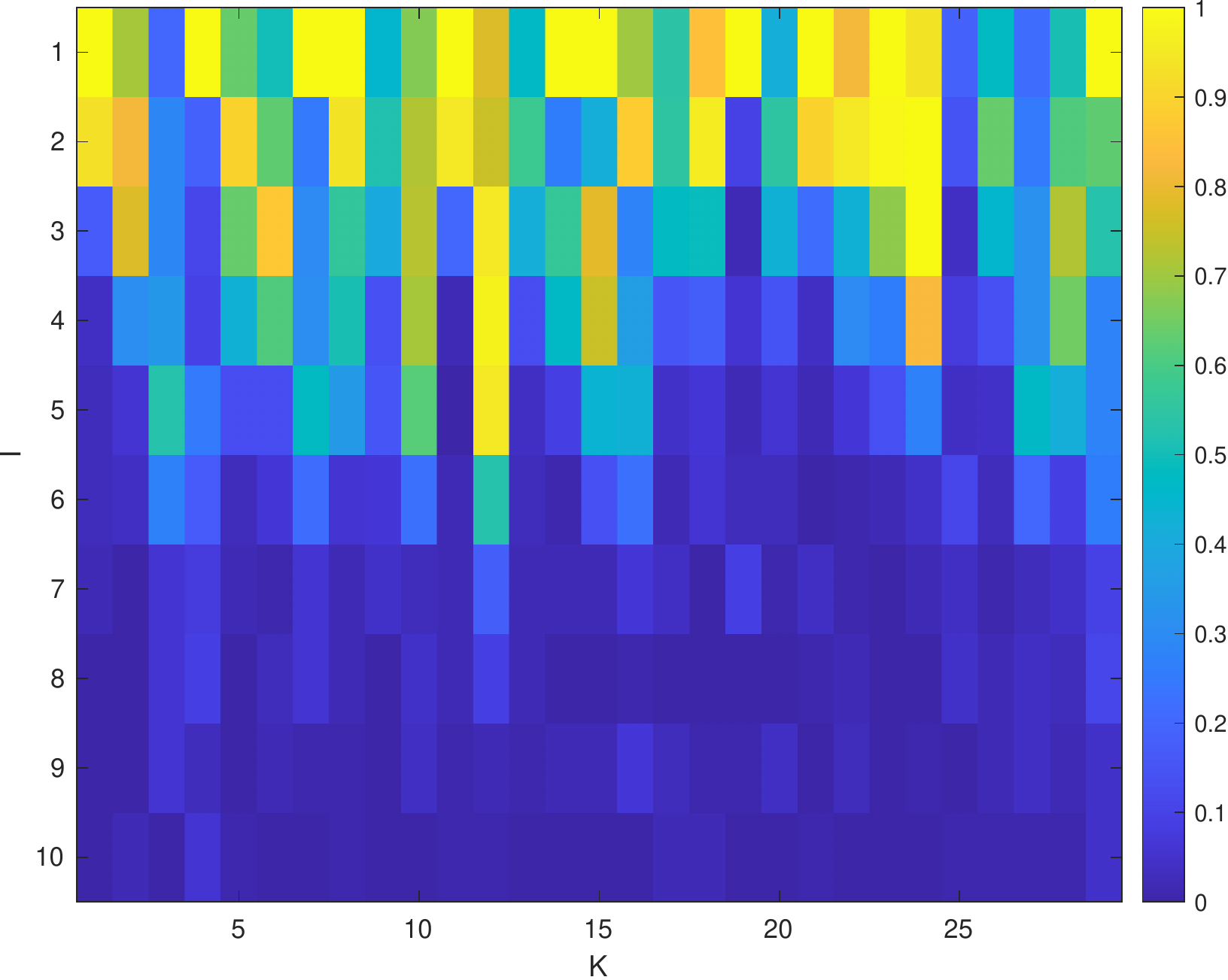}\smallskip
			\caption*{mCCA-HT \cite{Marrinan2018}}
			\label{fig:exp-b08_ex-mar}
		\end{subfigure}\bigskip
		\label{dummy2}
		\begin{subfigure}{.48\textwidth}
			\centering
			\includegraphics[scale=.2]{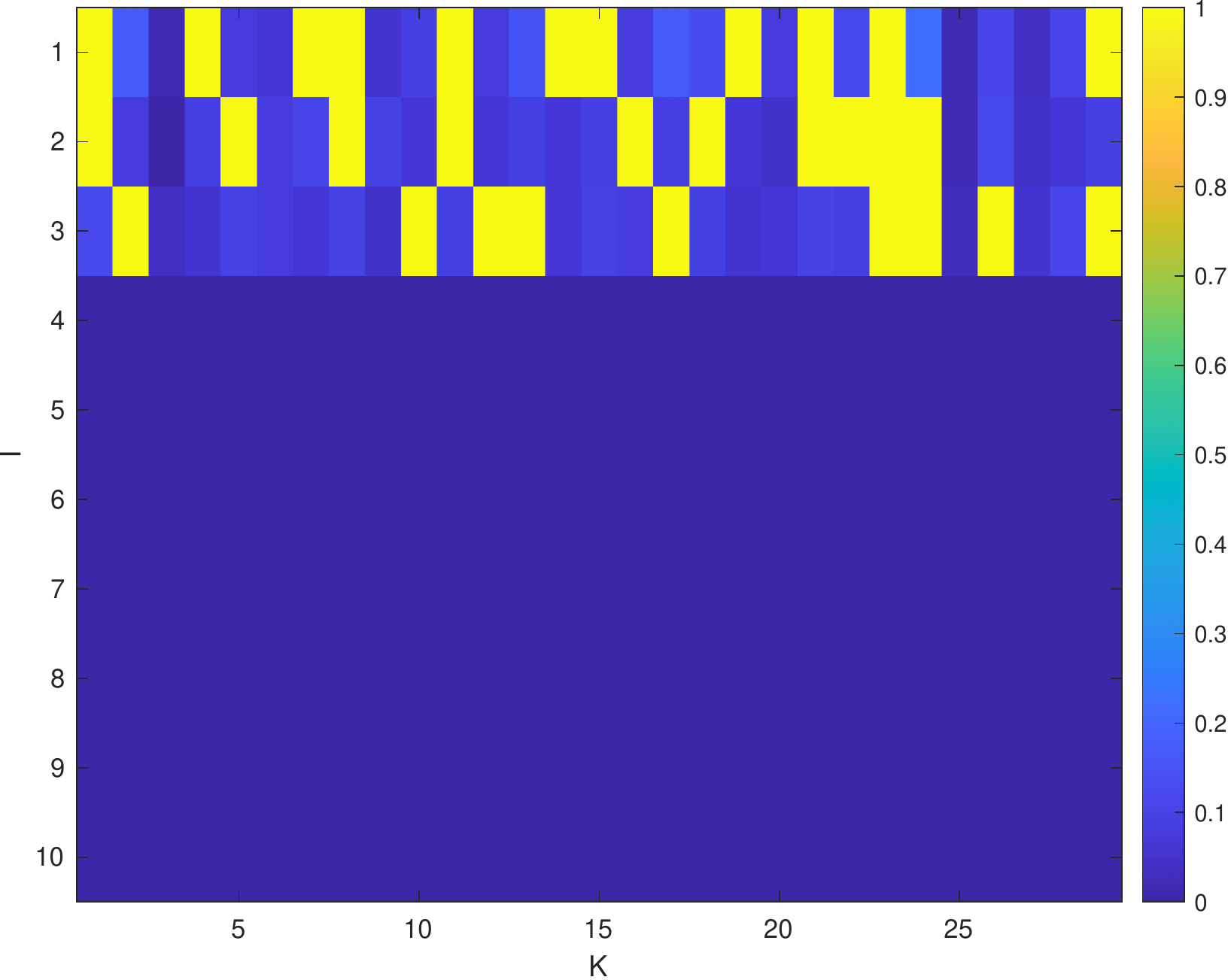}\smallskip
			\caption*{\gls{ts} \cite{Hasija2020}}
			\label{fig:exp-b08_ex-tsp}
		\end{subfigure}
		\begin{subfigure}{.48\textwidth}
			\centering
			\includegraphics[scale=.2]{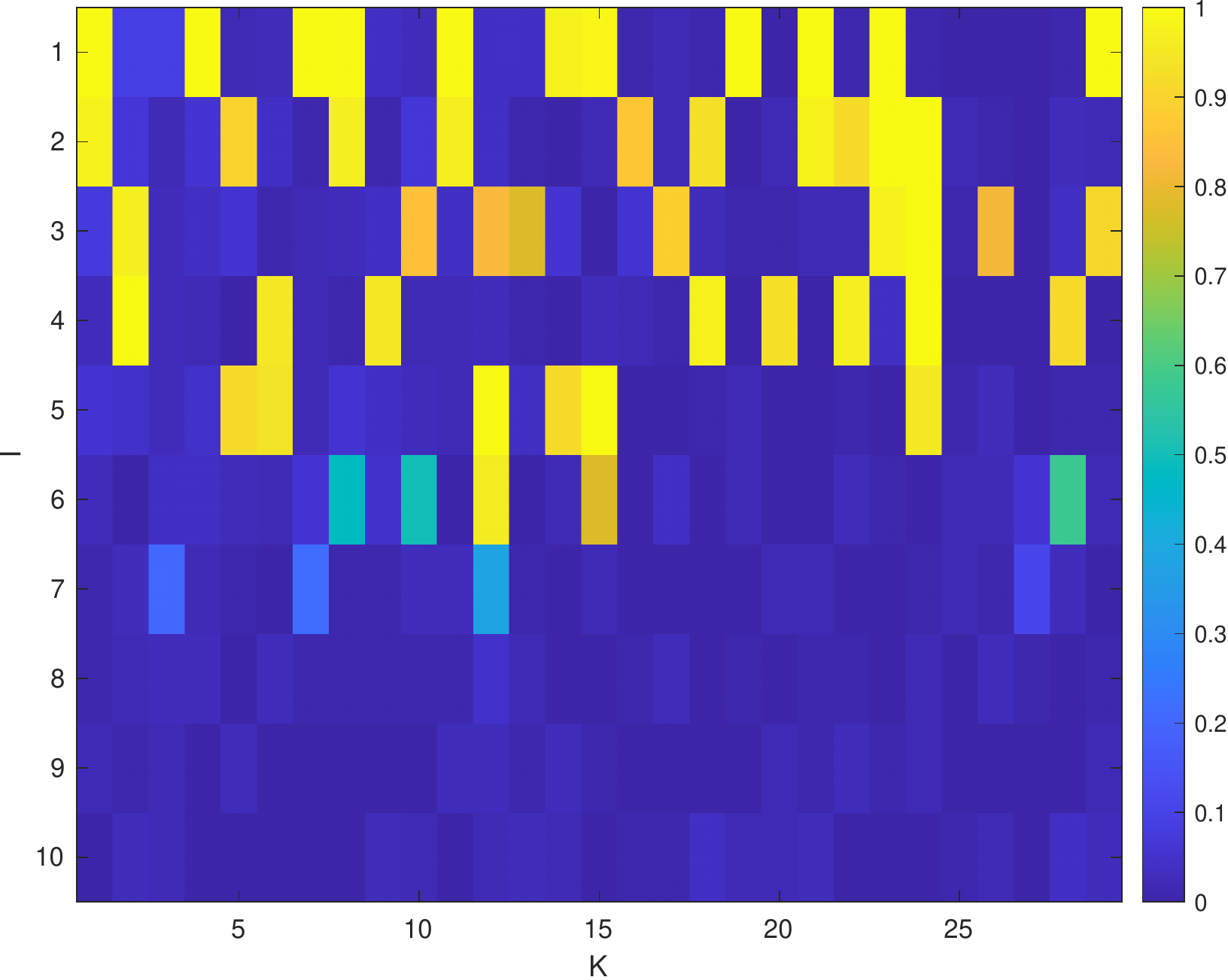}\smallskip
			\caption*{\gls{lfdrmultcost}}
			\label{fig:exp-b08_ex-osp}
		\end{subfigure}
		\bigskip
		\caption{detected correlation structures for $\numSet = 29$}
		\label{fig:exp-b08_examples}
	\end{subfigure}
	\caption{Experiment~3 - Performance comparison depending on the number of data sets for $\pi_0 = 0.8$. mCCA-HT yields unreasonably many false positives. The discoveries of \gls{ts} contain more and more false positives as the number of data sets grow. Hence, \gls{ts} is not well suited for applications where the number of data sets is large. The empirical \glspl{fdr} for our proposed \gls{lfdrmultcost} are nearly constant across all $\numSet$. The detection power on the atom level is constant for a while, before it starts to decline marginally. This is due to the decreasing relative sample size. Setting \(\fdrThrCmp = 1\), i.e., removing the constraint on $\fdrCmp$ has little impact: As the number of sets $\numSet$ grows, more and more components are correlated between at least a few sets. Hence, more and more components become true positives.}
	\label{fig:exp-b08}
\end{figure}

\subsubsection*{Experiment~4} We now study the impact of the proportion of true atom level null hypotheses $\pi_0$ on the performance. We simulate $\numSet = 25$ data sets with $N = 600$ samples of Gaussian components and noise, $\mathrm{SNR} = 5\unit{dB}$ and $\numSrc = 10$ components. The performance measures are shown in Fig.~\ref{fig:exp-C}. mCCA-HT performs insufficiently. With growing $\pi_0$, the number of true correlations decreases and maintaining a high detection power becomes increasingly difficult. The atom level detection power of our proposed \gls{lfdrmultcost} declines much slower than that of \gls{ts}. Remarkably, even for $.9\leq\pi_0\leq.975$, the detection power is only reduced on average by about $15\%$ in comparison to the much less challenging $\pi_0 = 0.7$. Beyond $\pi_0=.975$, also our proposed method detects nothing anymore. With \gls{lfdrmultcost} and $\fdrThrCmp = 0.1$, the component level \gls{fdr} is controlled everywhere. Again, the detection power on the atom level is nearly identical to the results with $\fdrThrCmp = 1$.

\begin{figure}
    \begin{subfigure}{.49\textwidth}
        \includegraphics[scale=.5]{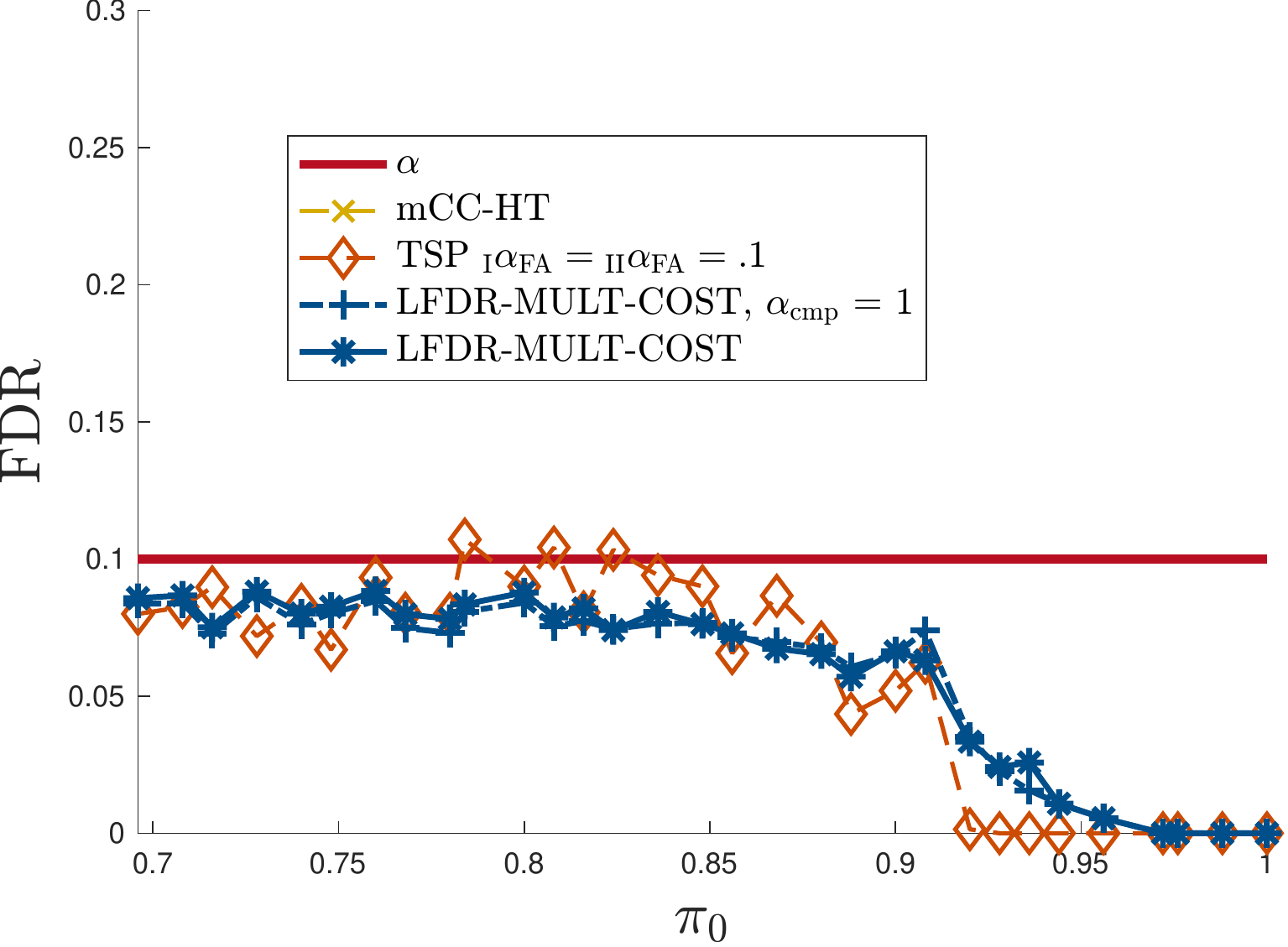}\smallskip
        \caption*{\(\fdr\)}
        \label{fig:exp-c_atom-fdr}
    \end{subfigure}
    \begin{subfigure}{.49\textwidth}
        \includegraphics[scale=.5]{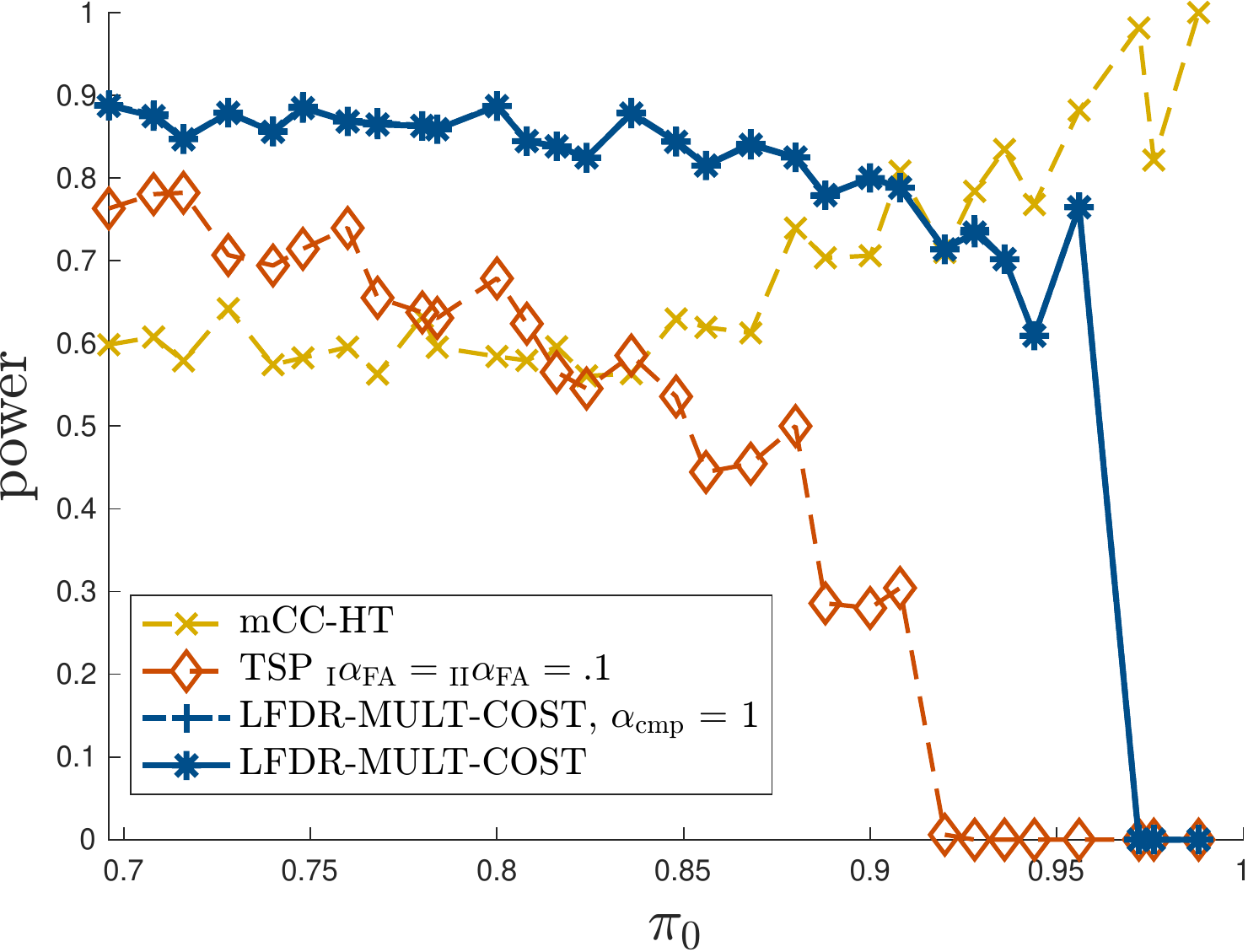}\smallskip
        \caption*{$\mathrm{Power}$}
        \label{fig:exp-c_atom-pow}
    \end{subfigure}\medskip
    \label{dummy}
    \begin{subfigure}{.49\textwidth}
        \includegraphics[scale=.5]{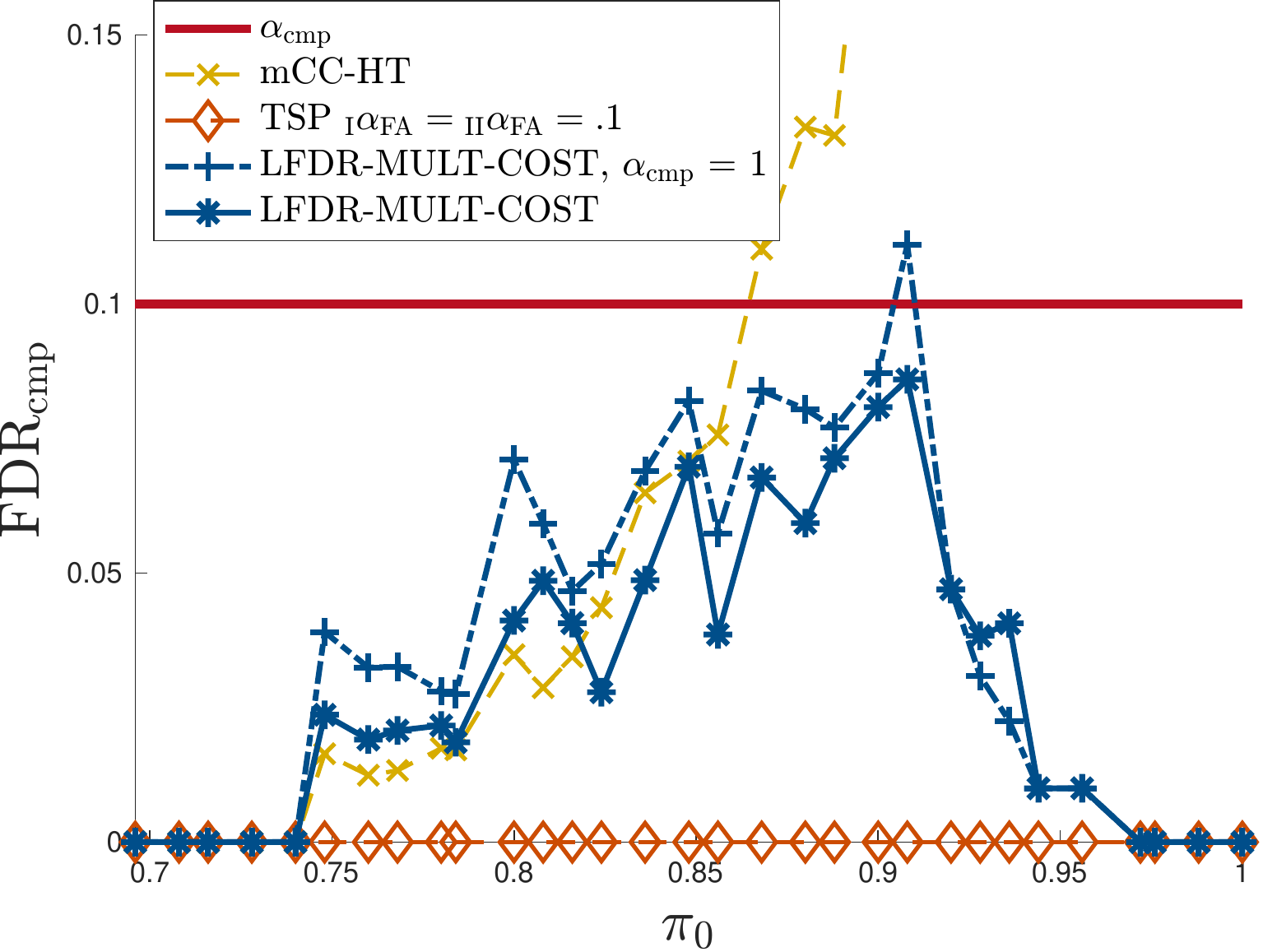}\smallskip
        \caption*{\(\fdrCmp\)}
        \label{fig:exp-c_cmp-fdr}
    \end{subfigure}
    \begin{subfigure}{.49\textwidth}
        \includegraphics[scale=.5]{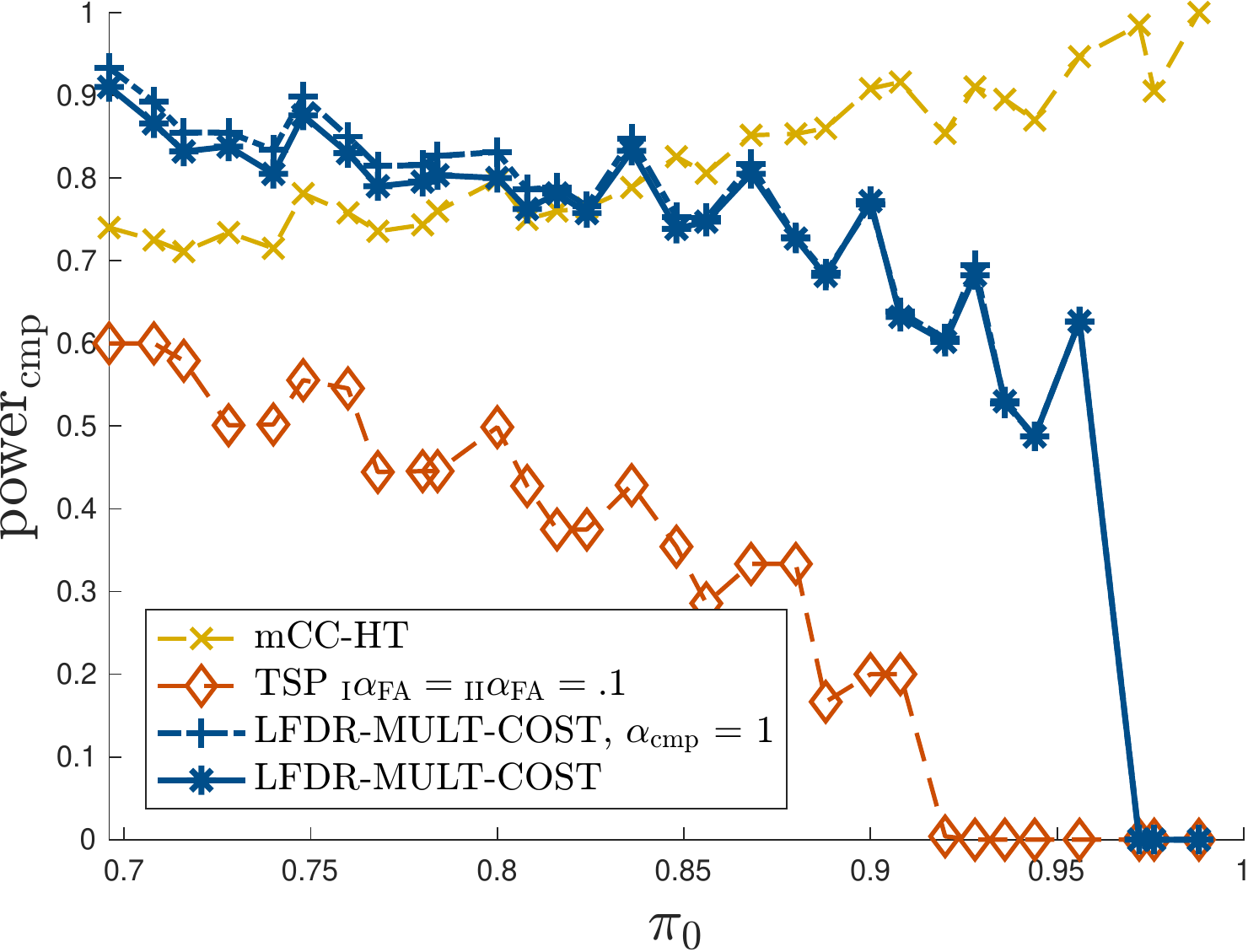}\smallskip
        \caption*{\(\mathrm{Power}_\text{cmp}\)}
        \label{fig:exp-c_cmp-pow}
    \end{subfigure}
	\caption{Experiment 4 - Performance comparison depending on $\pi_0$. As the fraction of true positives declines, maintaining high detection power becomes increasingly challenging. mCCA-HT should not be used due to a very high atom level \gls{fdr}. The detection power of our proposed \gls{lfdrmultcost} declines much slower than that of \gls{ts}. Hence, the proposed method is much more suitable for very challenging applications in which only a few correlations must be identified in large data sets. Again, controlling \(\fdrCmp\) costs close to nothing in terms of detection power. Nevertheless, it keeps \(\fdrCmp\) below the nominal level at all times.}
	\label{fig:exp-C}
\end{figure}


	\FloatBarrier

\section{Conclusion}
\label{sec:con}
The problem of identifying the complete correlation structure in the components of mulitple high-dimensional data sets was considered. A multiple hypothesis testing formulation of this problem was provided. The proposed solution is based on local false discovery rates and fulfills statistical performance guarantees with respect to false positives. In particular, the false discovery rate is controlled on two levels. The test statistic are based in the eigenvectors of the sample coherence matrix, whose probability models are learned from the data using the bootstrap. Our empirical results underline that our method works controls false positives even in the most challenging scenarios when sample size is very small, the number of data sets is very large, the proportion of true alternatives is very small. In addition, it is agnostic to the underlying data and noise probability models, which makes it insensitive to distributional uncertainties and heavy-tailed noise. This makes the proposed method applicable in a large number of practical applications, such as communications engineering, climate science, image processing and biomedicine.

	\clearpage
	\printbibliography[title = {References}]
	\clearpage
	\appendix
\label{app:res}
\FloatBarrier

In this appendix, we provide additional simulation results.

\subsubsection*{Experiment~3-b} In Fig.~\ref{fig:exp-b09}, additional for a variation of Experiment~3 from Section \ref{sec:sim-res} are shown. Here, $\pi_0 = 0.9$. Since the proportion of true atom null hypotheses is only $10\%$, this is a challenging scenario. Indeed, the detection power of \gls{ts} quickly becomes problematic as $\numSet$ increases. Our proposed \gls{lfdrmultcost} proves to be more suitable in high-dimensional, highly challenging scenarios. Its detection power is nearly constant across the entire range of simulated $\numSet$.

\begin{figure}
	\begin{subfigure}{.54\textwidth}
		\begin{subfigure}{.49\textwidth}
			\includegraphics[scale=.28]{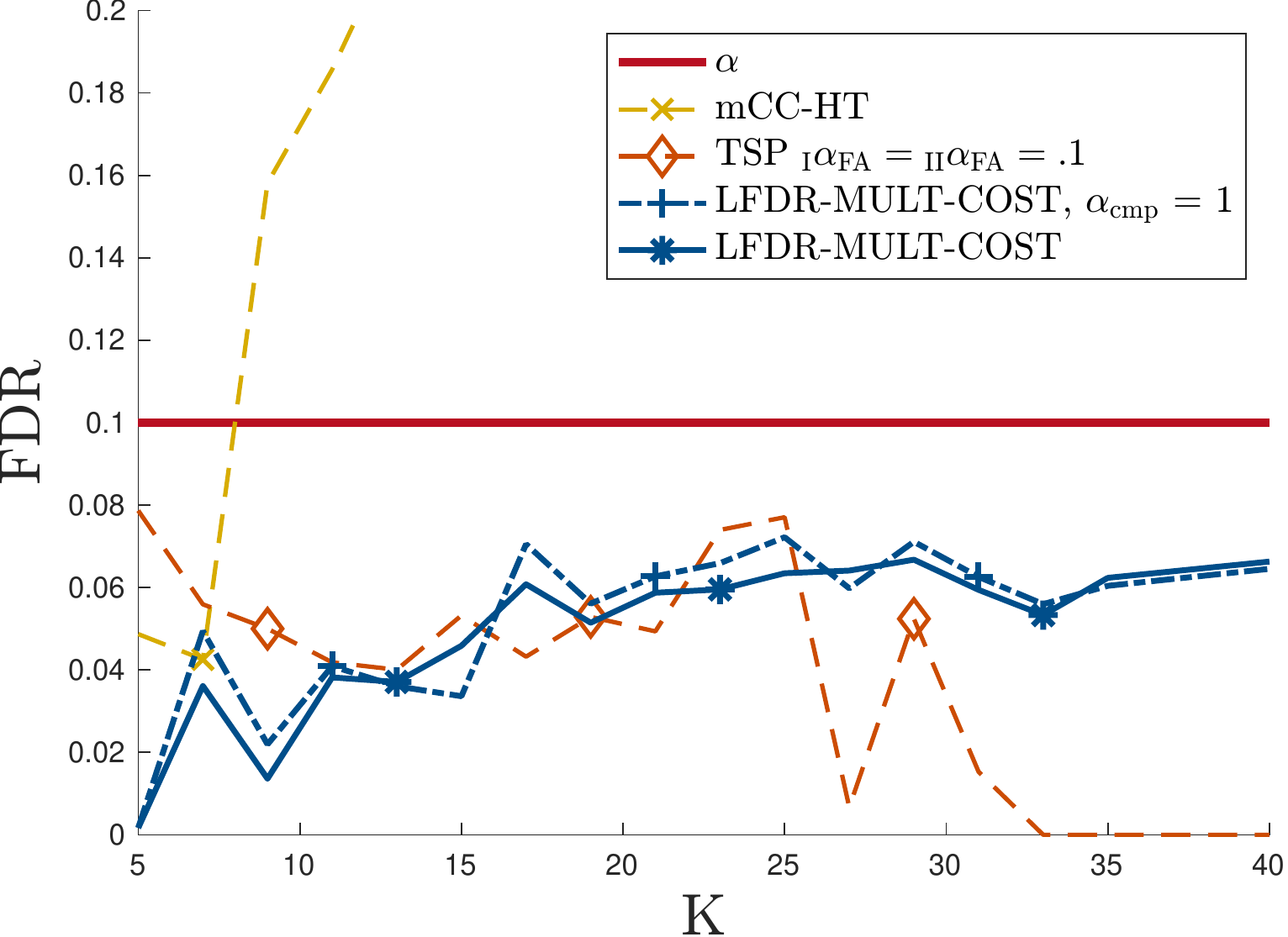}\smallskip
			\caption*{\(\fdr\)}
			\label{fig:exp-b09_atom-fdr}
		\end{subfigure}
		\begin{subfigure}{.49\textwidth}
			\includegraphics[scale=.28]{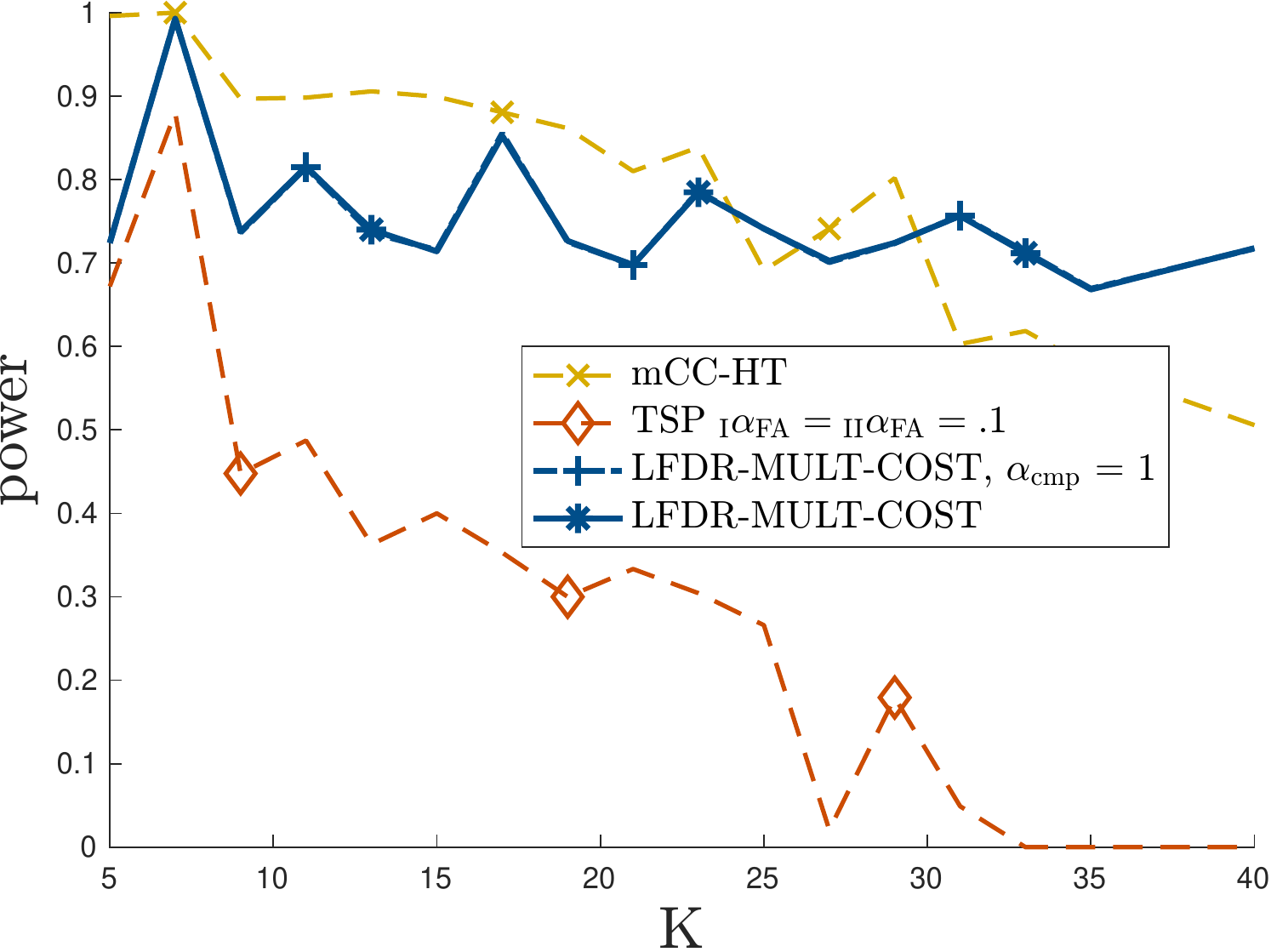}\smallskip
			\caption*{$\mathrm{Power}$}
			\label{fig:exp-b09_atom-pow}
		\end{subfigure}\medskip
		\label{dummy}
		\begin{subfigure}{.49\textwidth}
			\includegraphics[scale=.28]{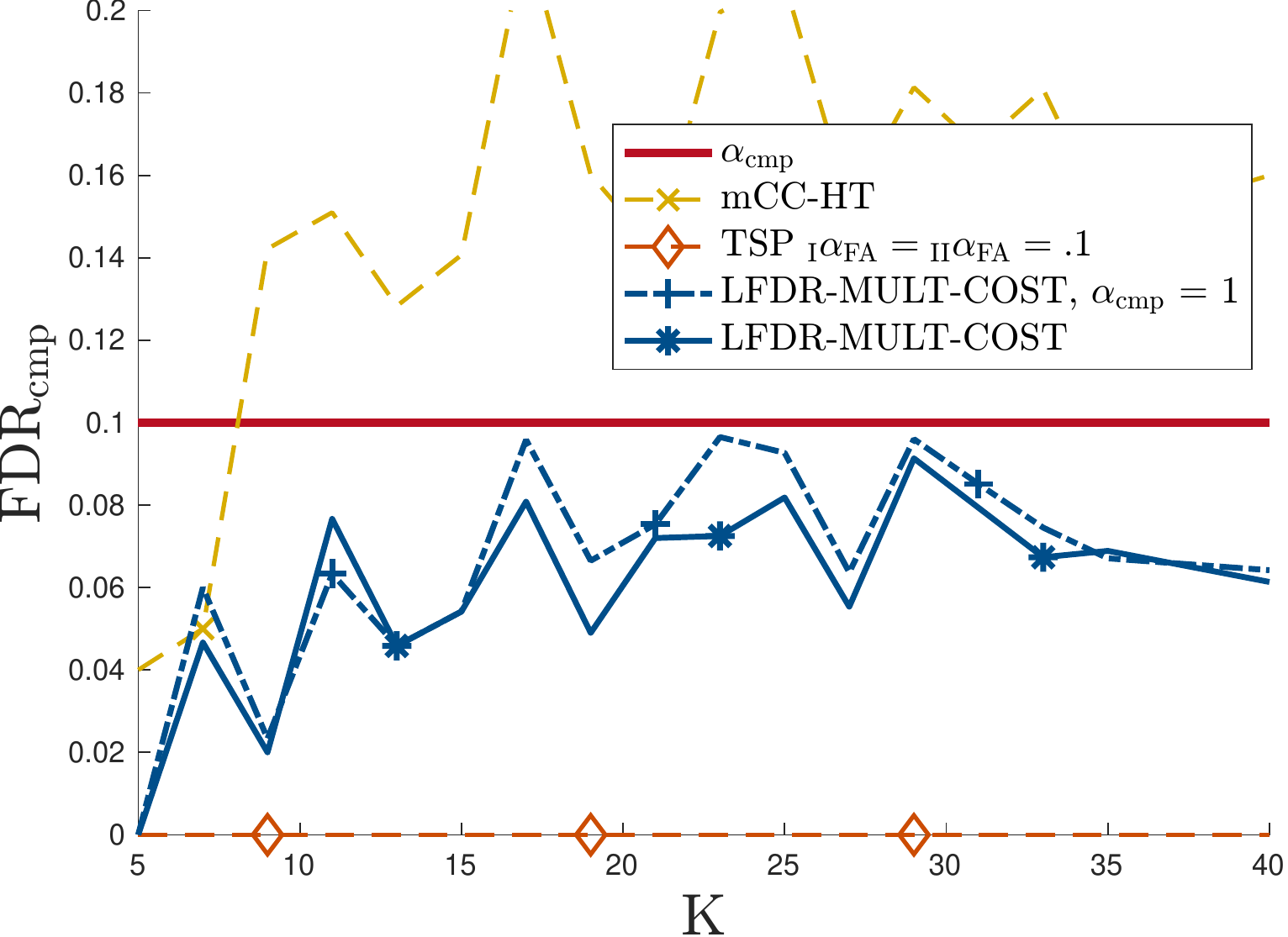}\smallskip
			\caption*{\(\fdrCmp\)}
			\label{fig:exp-b09_cmp-fdr}
		\end{subfigure}
		\begin{subfigure}{.49\textwidth}
			\includegraphics[scale=.28]{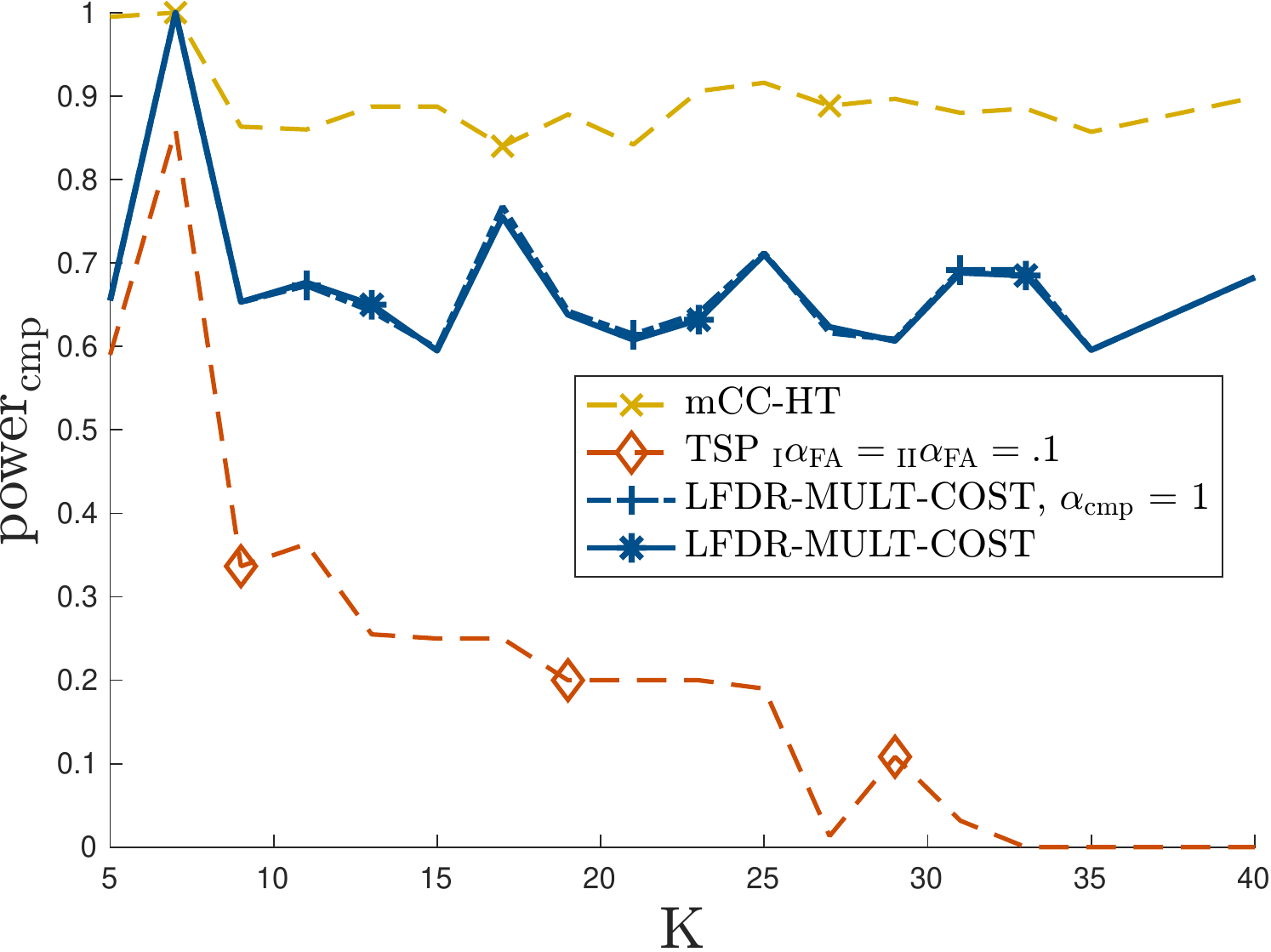}\smallskip
			\caption*{\(\mathrm{Power}_\text{cmp}\)}
			\label{fig:exp-b09_cmp-pow}
		\end{subfigure}
		\medskip
		\caption{performance measures}
		\label{fig:exp-b09_perf}
	\end{subfigure}
	\begin{subfigure}{.44\textwidth}\centering
		\begin{subfigure}{.47\textwidth}
			\centering
			\includegraphics[scale=.2]{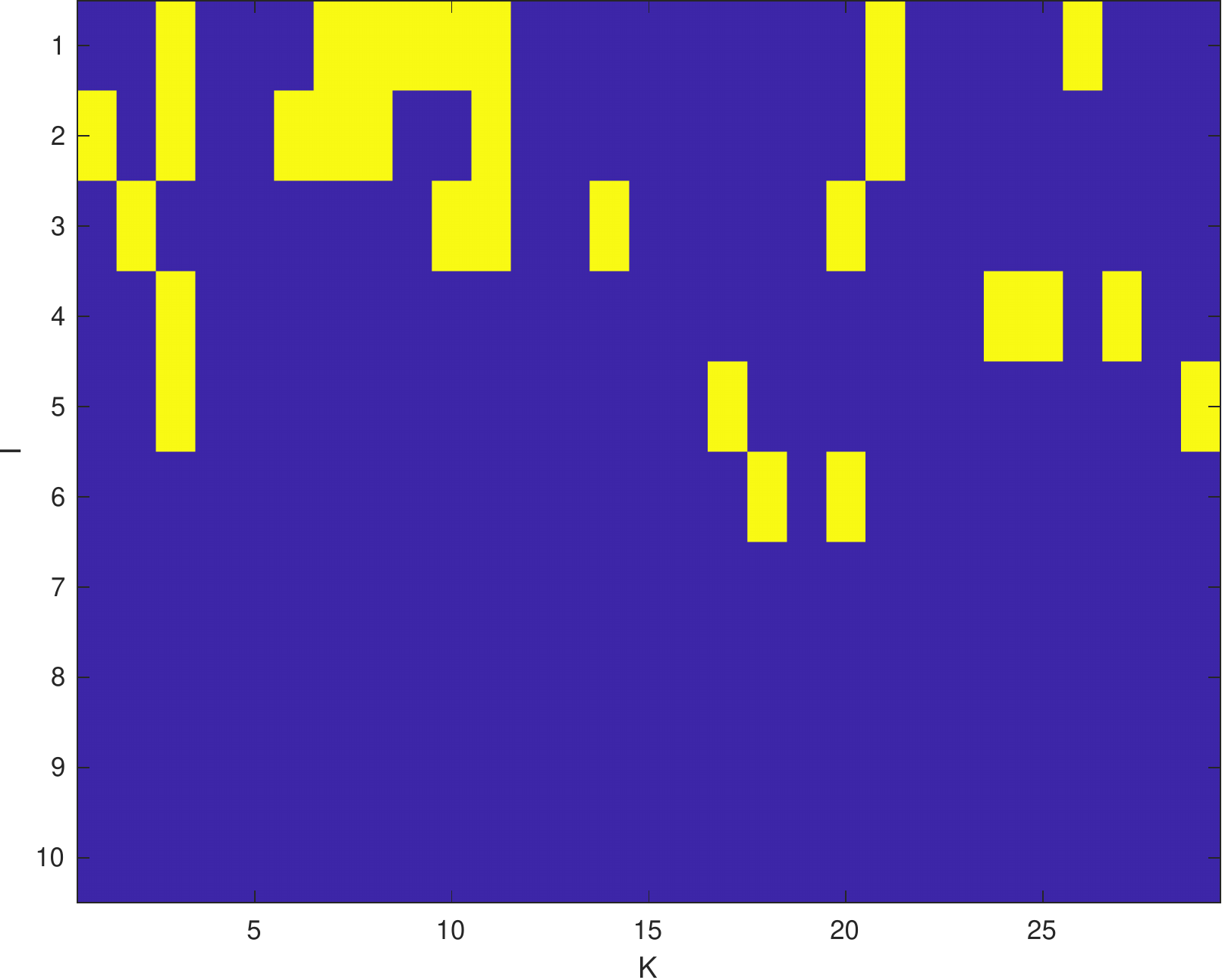}\smallskip
			\caption*{ground truth}
			\label{fig:exp-b09_ex-tru}
		\end{subfigure}
		\begin{subfigure}{.48\textwidth}
			\centering
			\includegraphics[scale=.2]{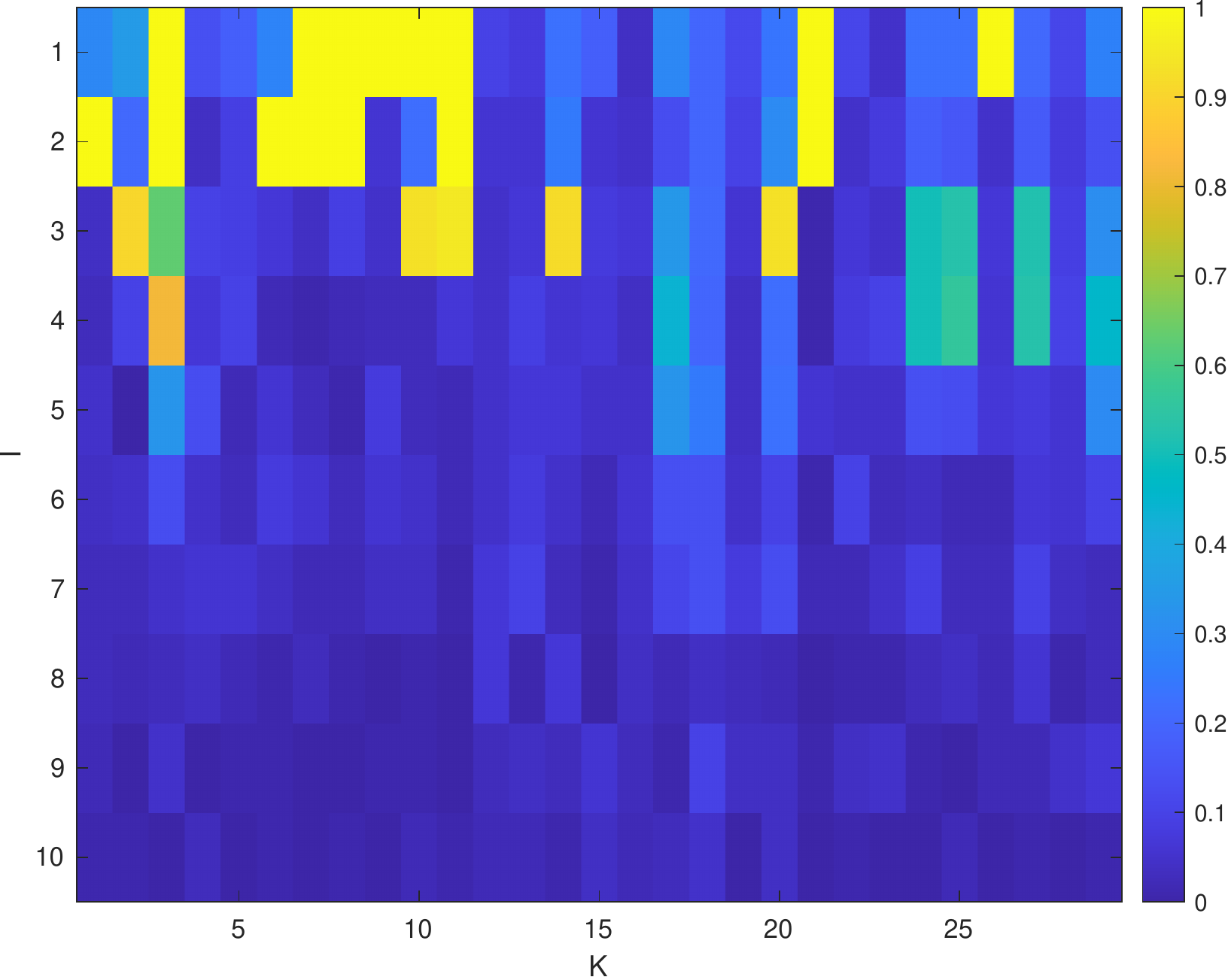}\smallskip
			\caption*{mCCA-HT \cite{Marrinan2018}}
			\label{fig:exp-b09_ex-mar}
		\end{subfigure}\bigskip
		\label{dummy2}
		\begin{subfigure}{.48\textwidth}
			\centering
			\includegraphics[scale=.2]{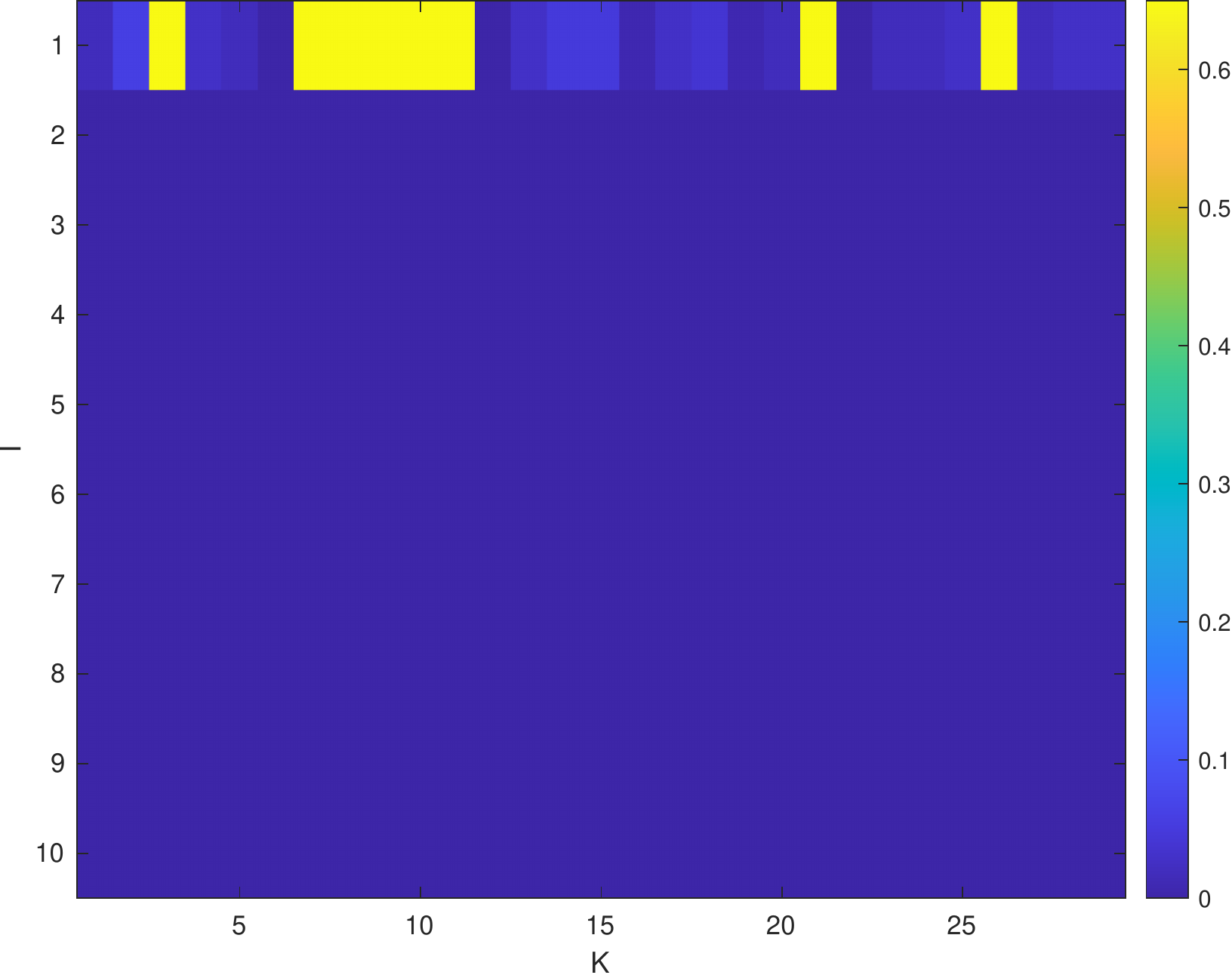}\smallskip
			\caption*{\gls{ts} \cite{Hasija2020}}
			\label{fig:exp-b09_ex-tsp}
		\end{subfigure}
		\begin{subfigure}{.48\textwidth}
			\centering
			\includegraphics[scale=.2]{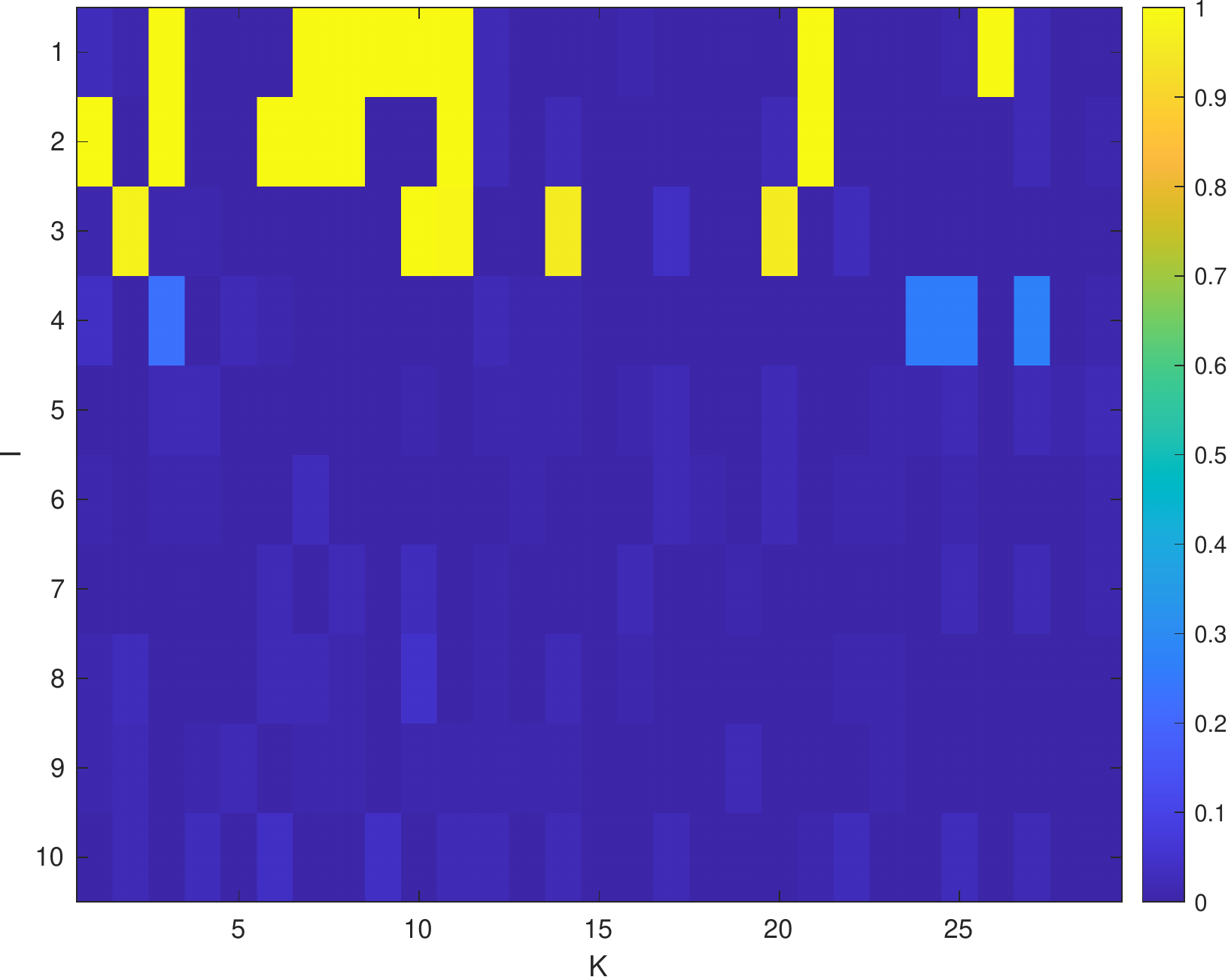}\smallskip
			\caption*{\gls{lfdrmultcost}}
			\label{fig:exp-b09_ex-osp}
		\end{subfigure}
		\bigskip
		\caption{detected correlation structures for $\numSet = 29$}
		\label{fig:exp-b09_examples}
	\end{subfigure}
	\caption{Experiment~3-b - Performance comparison depending on the number of data sets for  $\pi_0 = 0.9$. mCCA-HT yields unreasonably many false positives. \gls{ts} has trouble identifying true correlations as the number of sets grows and as true positives become more and more sparsely distributed across the sets. The detection power with \gls{lfdrmultcost} on the atom level declines very slowly. The decline is due to the decreasing relative sample size. Again, the control of $\fdrCmp\leq 0.1$ comes at no cost regarding detection power while reducing the already well-below the limit residing component FDR marginally in relation to \gls{lfdrmultcost} with nominal level \(\fdrThrCmp = 1\).}
	\label{fig:exp-b09}
\end{figure}

\subsubsection*{Experiment~5} We finally evaluate the robustness of the proposed method to outliers. To this end, we deploy the $\epsilon$-contamination model \cite{Huber2009} where a certain fraction $\epsilon$ of data samples follows a contaminating distribution. Contamination with high-powered Gaussian noise, with a standard deviation of $3$ times the uncontaminated noise standard deviation was proposed in \cite{Huber2009} in Experiment Experiment~5-a. In addition, we also evaluate the much more challenging contamination with a point mass at a fix large value in Experiment~5-b. We generate $\numSam = 1\,000$ samples of $\numSrc = 6$ Gaussian components and noise vectors for $\numSet = 12$ data sets.
\paragraph*{a} All data sets and components get contaminated with a fraction of $\epsilon\in\{0, .25, .5, 1\}$ outliers. For $\epsilon = 0$, there is no contamination. For $\epsilon = 1$, the contaminating distribution has completely replaced the original distribution. The results are shown in Fig.~\ref{fig:exp-R-a}. As we see, both \gls{ts} and \gls{lfdrmultcost} are fairly robust to such contamination, since the detection power and false discovery rates remain close to constant across $\epsilon$.

\begin{figure}
	\begin{subfigure}{.49\textwidth}
		\includegraphics[scale=.5]{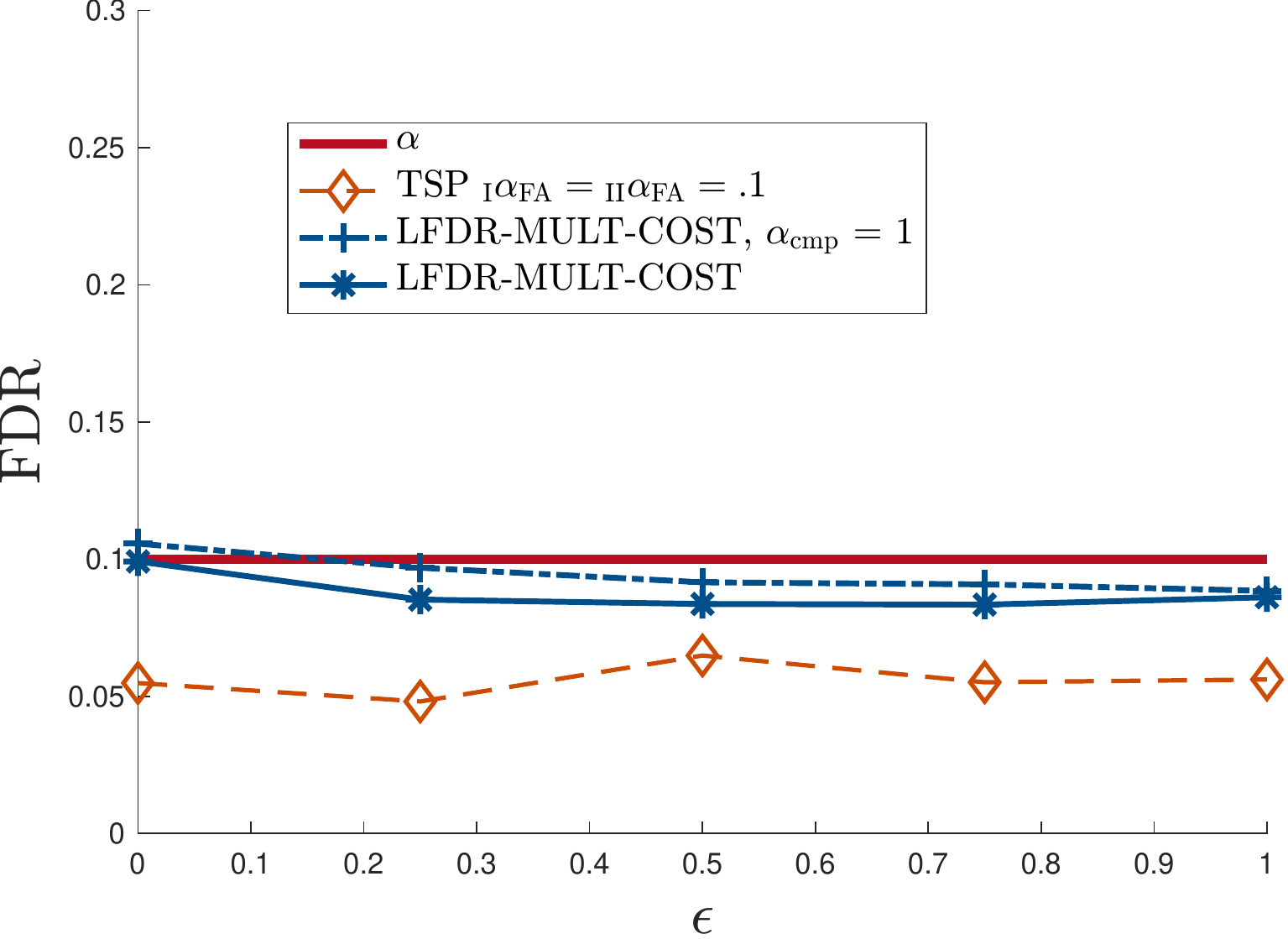}\smallskip
		\caption*{\(\fdr\)}
		\label{fig:exp-R-gauss_atom-fdr}
	\end{subfigure}§
	\begin{subfigure}{.49\textwidth}
		\includegraphics[scale=.5]{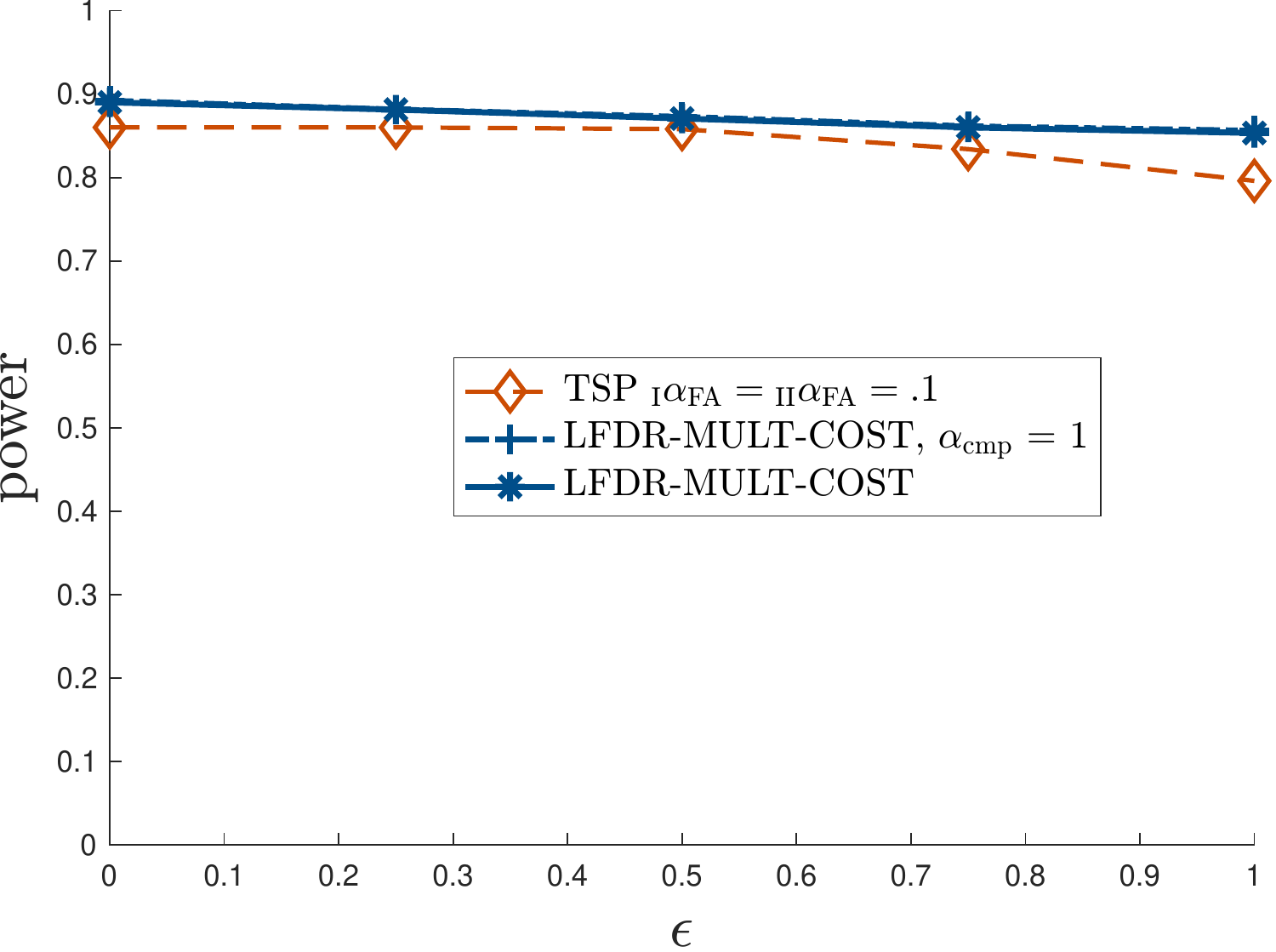}\smallskip
		\caption*{$\mathrm{Power}$}
		\label{fig:exp-R-gauss_atom-pow}
	\end{subfigure}\medskip
	\label{dummy}
	\begin{subfigure}{.49\textwidth}
		\includegraphics[scale=.5]{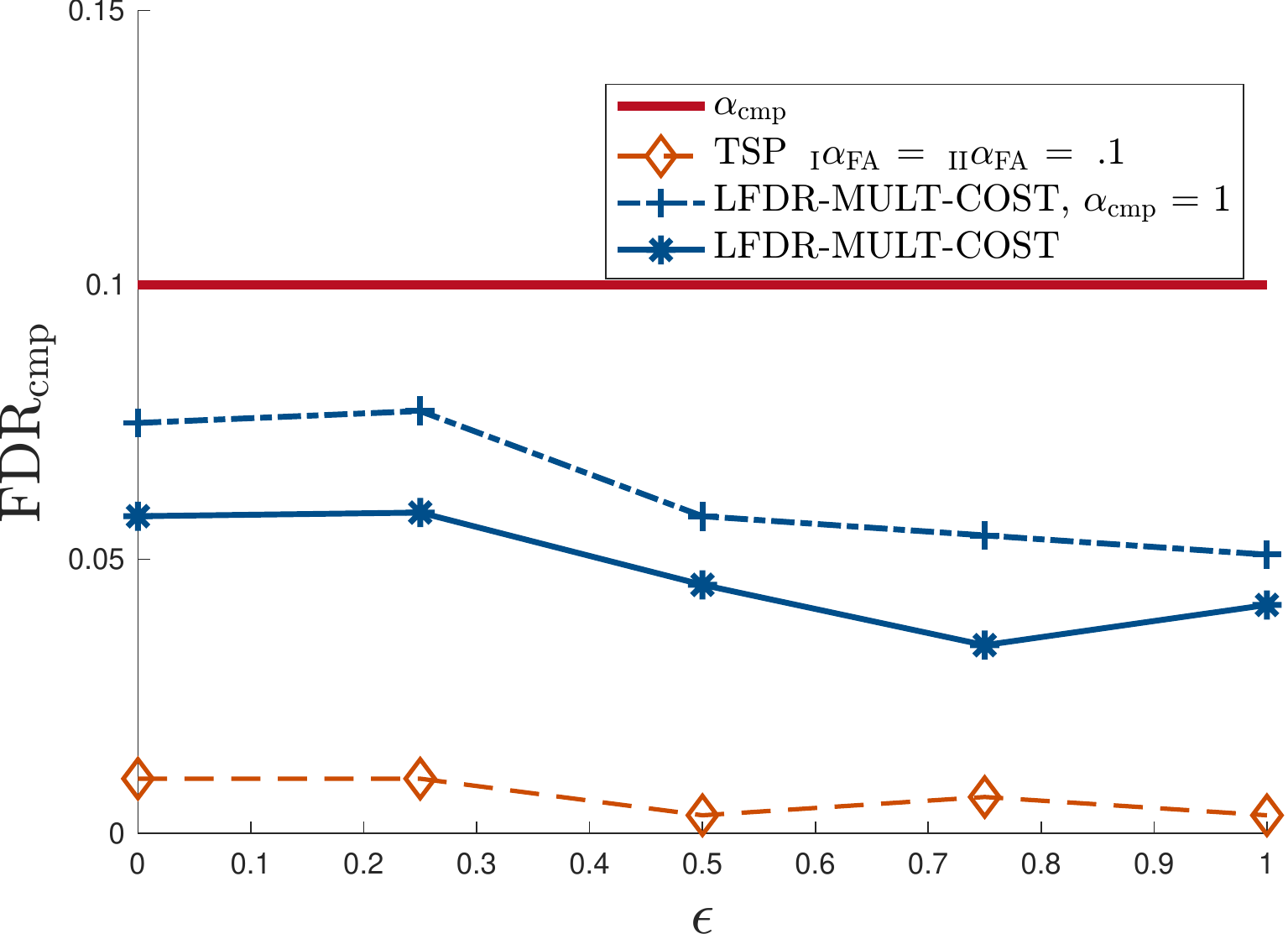}\smallskip
		\caption*{\(\fdrCmp\)}
		\label{fig:exp-R-gauss_cmp-fdr}
	\end{subfigure}
	\begin{subfigure}{.49\textwidth}
		\includegraphics[scale=.5]{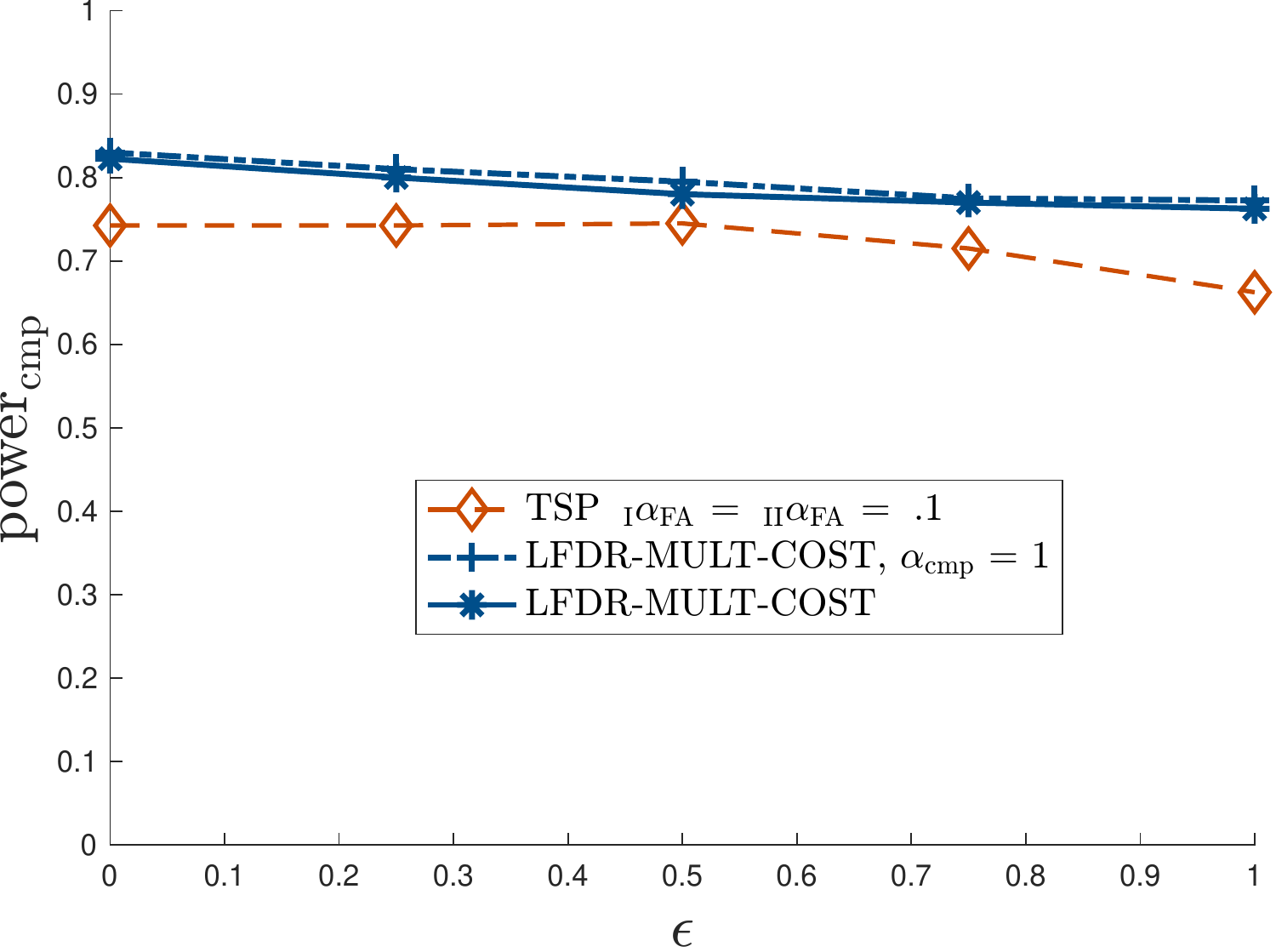}\smallskip
		\caption*{\(\mathrm{Power}_\text{cmp}\)}
		\label{fig:exp-R-gauss_cmp-pow}
	\end{subfigure}
	\caption{Experiment 5-a - Empirical evaluation of the robustness of \gls{ts} \cite{Hasija2020} and the proposed \gls{lfdrmultcost} with and without component-level \gls{fdr} versus outliers. $\epsilon$ is the percentage of contaminated samples by an additive Gaussian contamination with large variance. All methods appear fairly robust to such outliers and yield high detection power and low \gls{fdr}.}
	\label{fig:exp-R-a}
\end{figure}

\paragraph*{b} In this experiment, we contaminate $4$ out of the $6$ rows in the data matrices $\obsMat$ for $8$ out of the $12$ data sets with a point mass of value $10$. This is a very challenging type of outlier, since a point mass creates a strong imbalance in the tails of the contaminated noise distribution. The detection power of both methods is reduced for the contaminated data. Nevertheless, the \gls{fdr} control properties of our proposed \gls{lfdrmultcost} remain intact, as illustrated in Fig.~\ref{fig:exp-R-b}. \gls{ts} produces a higher proportion of false positives under contamination. 

\begin{figure}
	\begin{subfigure}{.49\textwidth}
		\includegraphics[scale=.5]{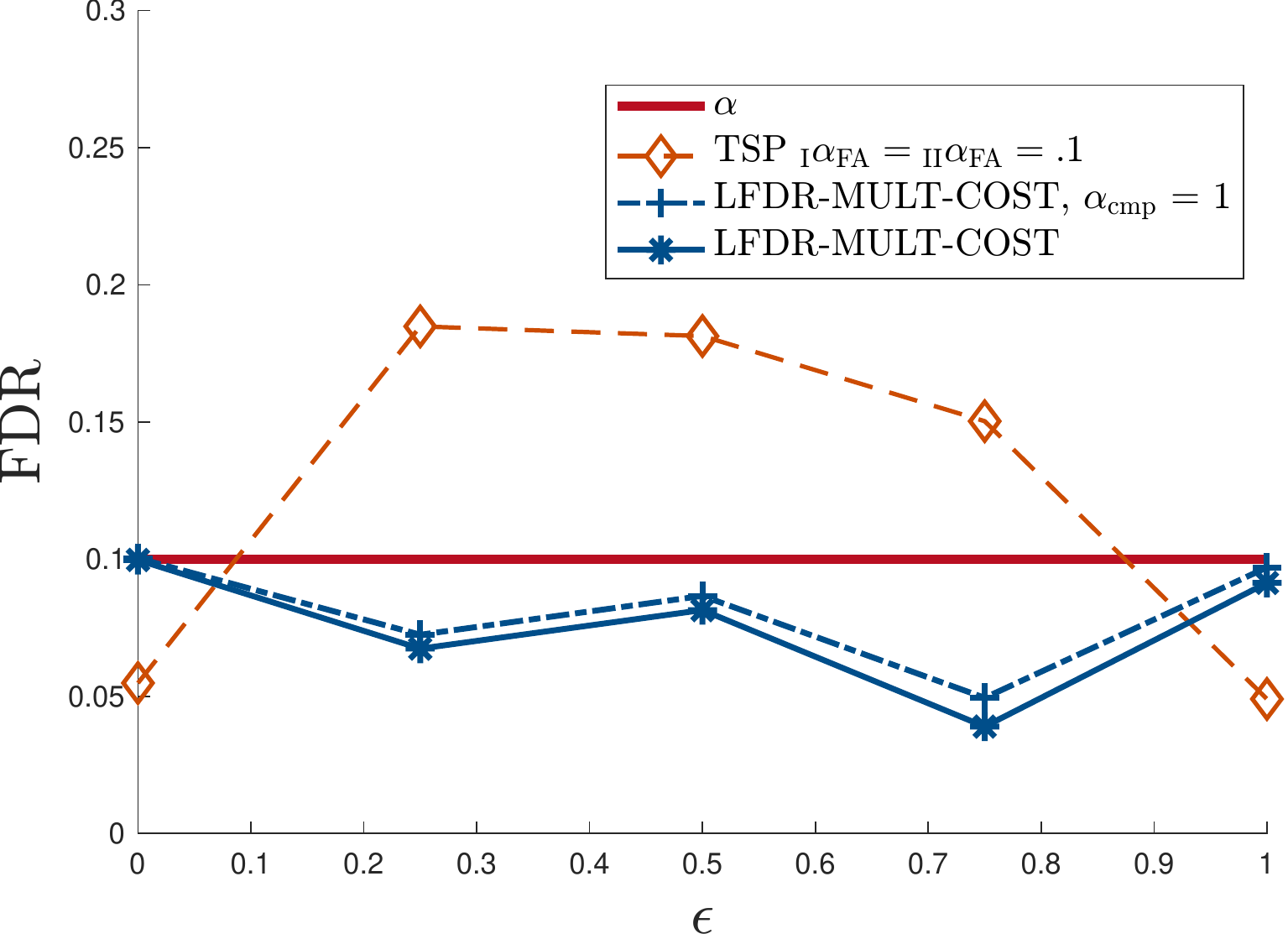}\smallskip
		\caption*{\(\fdr\)}
		\label{fig:exp-R-pm_atom-fdr}
	\end{subfigure}
	\begin{subfigure}{.49\textwidth}
		\includegraphics[scale=.5]{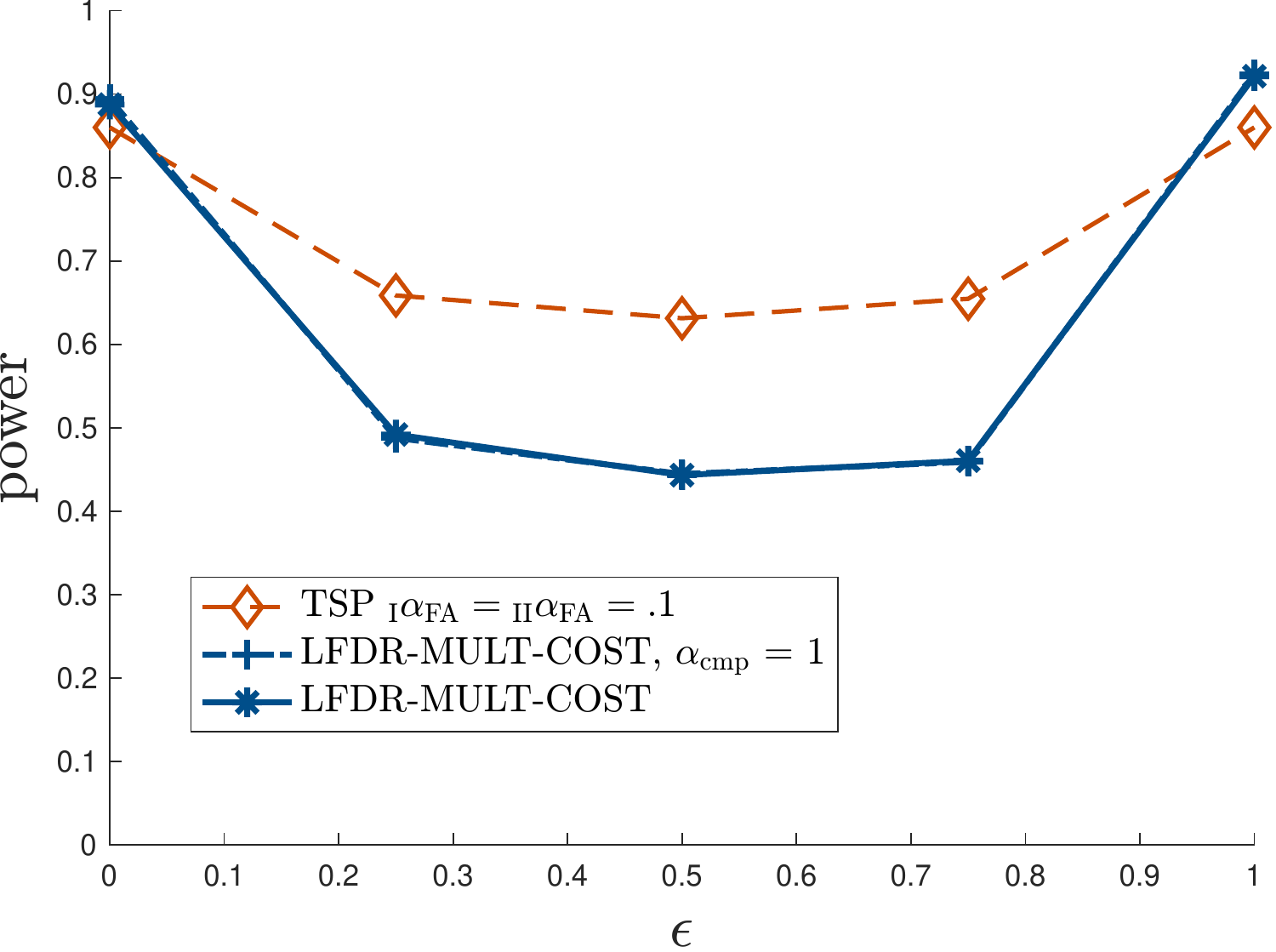}\smallskip
		\caption*{$\mathrm{Power}$}
		\label{fig:exp-R-pm_atom-pow}
	\end{subfigure}\medskip
	\label{dummy}
	\begin{subfigure}{.49\textwidth}
		\includegraphics[scale=.5]{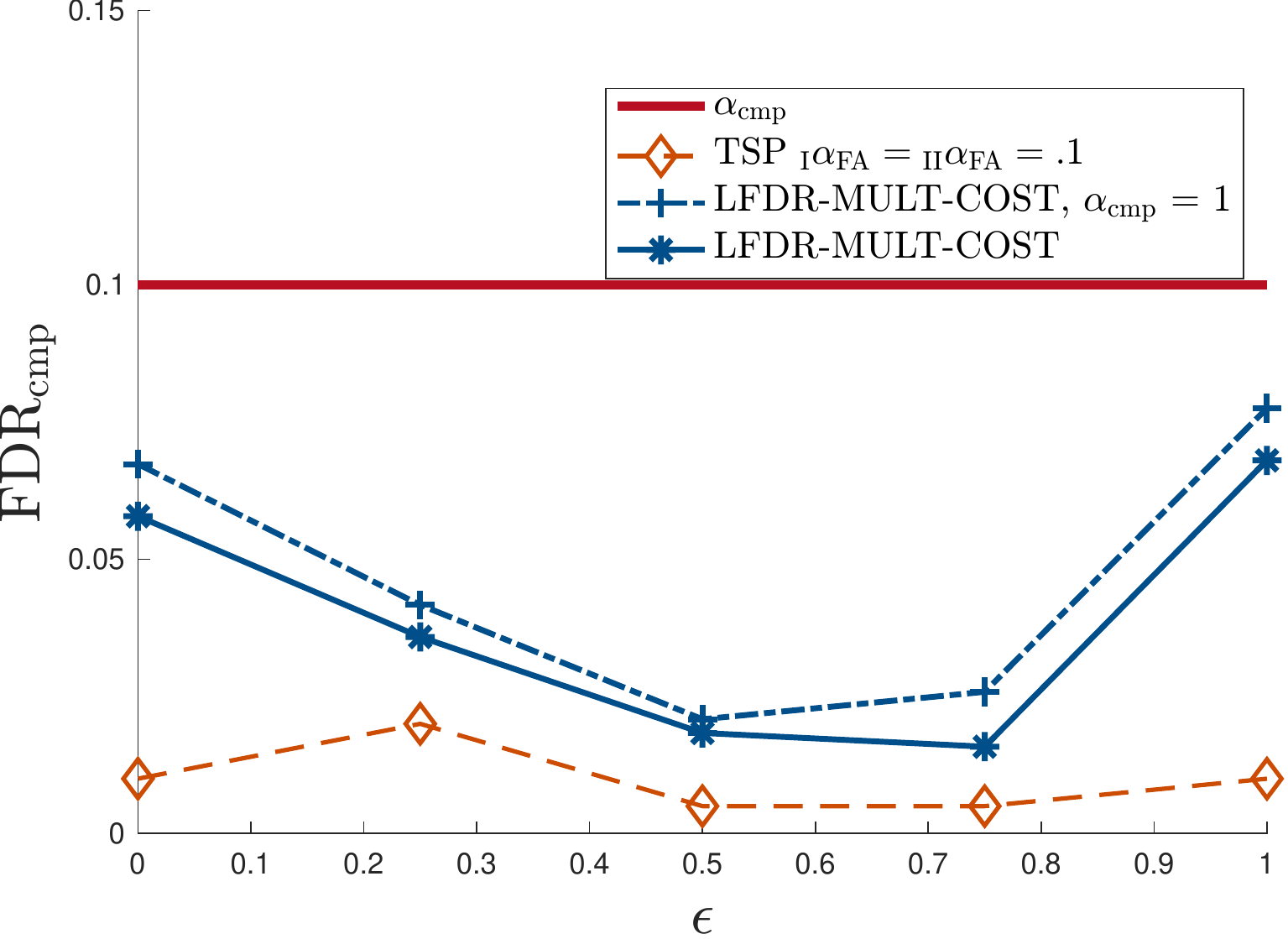}\smallskip
		\caption*{\(\fdrCmp\)}
		\label{fig:exp-R-pm_cmp-fdr}
	\end{subfigure}
	\begin{subfigure}{.49\textwidth}
		\includegraphics[scale=.5]{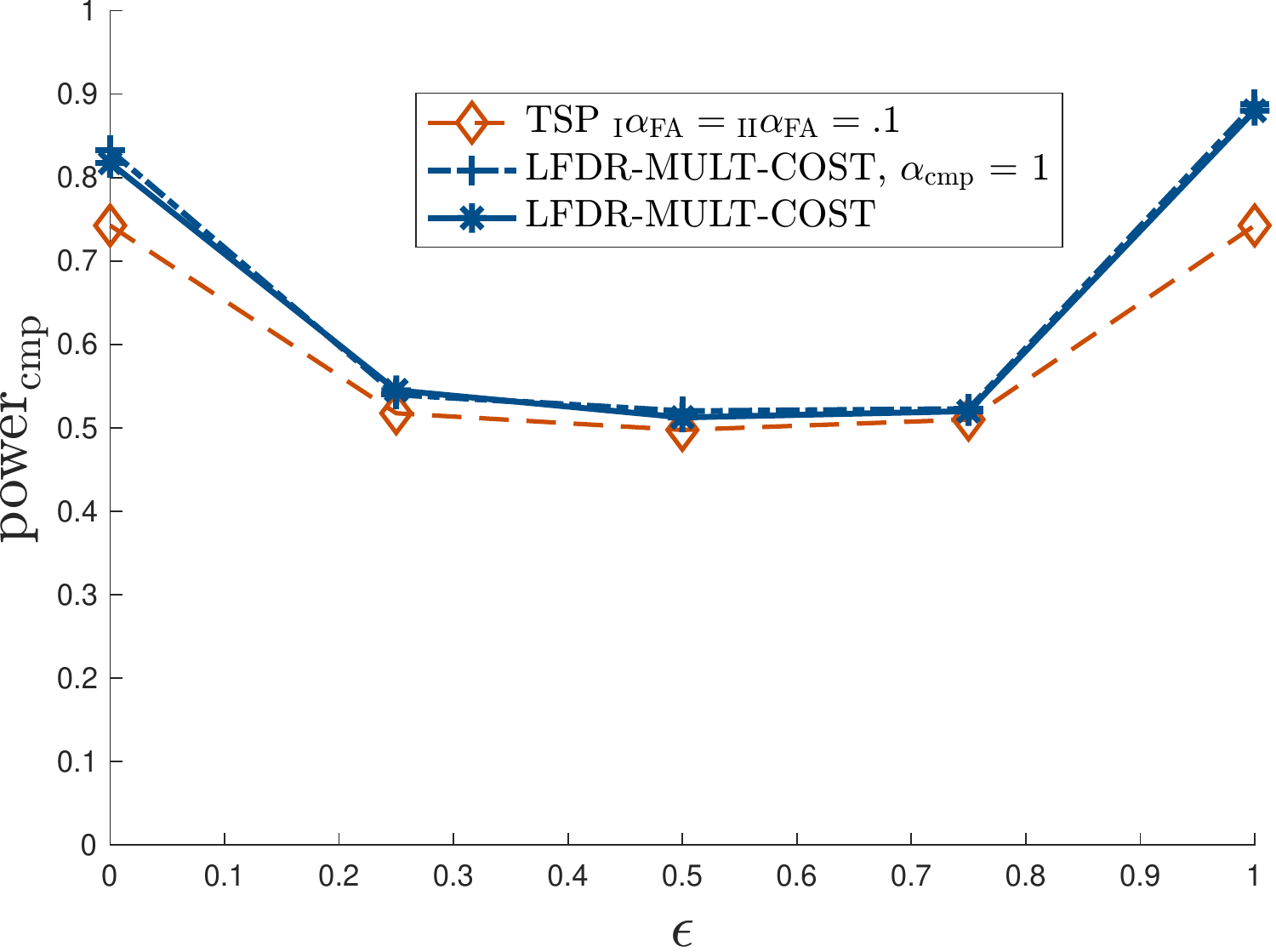}\smallskip
		\caption*{\(\mathrm{Power}_\text{cmp}\)}
		\label{fig:exp-R-pm_cmp-pow}
	\end{subfigure}
	\caption{Experiment 5-b - Empirical evaluation of the robustness of \gls{ts} \cite{Hasija2020} and the proposed \gls{lfdrmultcost} with and without component-level \gls{fdr} versus outliers. $\epsilon$ is the percentage of contaminated samples by an additive point mass contamination at $\delta = 10$. The \gls{fdr} control of our proposed \gls{lfdrmultcost} is robust to this very challenging type of contamination. While the noise distribution is a mixture of the outlier point mass and the uncontaminated noise, the detection power is reduced.}
	\label{fig:exp-R-b}
\end{figure}

\end{document}